%% file: main.tex
\documentclass[acmsmall,screen]{acmart}

\setcopyright{rightsretained} 
\acmJournal{PACMPL}
\acmYear{2019} 
\acmVolume{3} 
\acmNumber{POPL} 
\acmArticle{52} 
\acmMonth{1} 
\acmPrice{}
\acmDOI{10.1145/3290365}
\copyrightyear{2019}
\usepackage{booktabs}   %% For formal tables:
\usepackage{subcaption} %% For complex figures with subfigures/subcaptions
\usepackage{
      afterpage,   % useful to control float positioning in latex
      amsmath,     % math package
      mathtools,   % standard package
      amssymb,     % additional squiggly symbols
      amsfonts,    % additional fancy fonts
      pdfsync,      % cre­ates an aux­il­iary file with ge­o­met­ri­cal in­for­ma­tion to per­mit ref­er­ences back and forth be­tween source and PDF, as­sum­ing a con­form­ing ed­i­tor and PDF viewer
      url,          % To insert urls
      xspace,       % handy for trailing spaces after math items
      enumitem      % for having inlined enumerations
      }

\usepackage{stmaryrd} % used for llbracket and rrbracket

\usepackage{centernot}
\usepackage{hyperref}
\usepackage{cleveref}
\crefname{figure}{Fig.}{Figs.}
\crefname{definition}{Def.}{Defs.}
\crefname{defn}{Def.}{Defs.}
\crefname{equation}{Eq.}{Eqs.}
\crefname{theorem}{Thm.}{Thms.}
\crefname{thm}{Thm.}{Thms.}
\crefname{lemma}{Lem.}{Lems.}
\crefname{lem}{Lem.}{Lems.}
\crefname{corollary}{Cor.}{Cors.}
\crefname{cor}{Cor.}{Cors.}
\crefname{proposition}{Prop.}{Props.}
\crefname{prop}{Prop.}{Props.}
\crefname{enumi}{Item}{Items}
\crefname{example}{E.g.}{E.g.}
\crefname{section}{Sec.}{Secs.}
\crefname{table}{Table}{Tables}
\crefname{appendix}{Appendix}{Apps.}
\crefname{lstlisting}{List.}{List.}
\Crefname{lstlisting}{List.}{List.}

\usepackage{amsthm}  % sometimes clashes with llncs
\theoremstyle{plain}
\newtheorem{thm}{Theorem}[section]
\newtheorem{lem}[thm]{Lemma}
\newtheorem{prop}[thm]{Proposition}
\newtheorem{cor}[thm]{Corollary}

\theoremstyle{definition}
\newtheorem{defn}{Definition}[section]

\newtheorem{exmp}{Example}[section]

\theoremstyle{remark}
\newtheorem*{rem}{Remark}

\theoremstyle{plain}
\newtheorem*{thm*}{Theorem}
\newtheorem*{lem*}{Lemma}
\newtheorem*{prop*}{Proposition}
\newtheorem*{cor*}{Corollary}

\theoremstyle{definition}
\newtheorem*{defn*}{Definition}
\newtheorem*{conj*}{Conjecture}
\newtheorem*{exmp*}{Example}

\newenvironment{lemmabiss}[1]{\begin{trivlist}
\item[\hskip \labelsep {\bfseries Lemma #1}.]
\it}{\end{trivlist}}

\newenvironment{thmbiss}[1]{\begin{trivlist}
\item[\hskip \labelsep {\bfseries Theorem #1}]
\it}{\end{trivlist}}

\newenvironment{propbis}[1]{\begin{trivlist}
\item[\hskip \labelsep {\bfseries Proposition #1}]
\it}{\end{trivlist}}

\newenvironment{propbiss}[2]{\begin{trivlist}
\item[\hskip \labelsep {\bfseries Proposition #1}] {(#2).}
\it}{\end{trivlist}}

\usepackage{tikz}  % Figure, diagrams etc.,
\usetikzlibrary{matrix,shapes,arrows}
\tikzstyle{block} = [rectangle, draw, fill=white!20, text width=4em, text centered, minimum height=3em]
\tikzstyle{line} = [draw, very thick, color=black!90, -latex']

\usepackage[inference]{semantic} % I used it for the proof rules

\usepackage{pic}  % pi calculus, CCS, LTSs and other process calculi stuff
\usepackage{monitor} % additional macros used for monitors.
\usepackage{logic}  % mu calculus, etc
\usepackage{shorthand}  % a bunch of shorthand macros
\usepackage{mathpartir}  % a stylefile

\newcommand{\citeMac}[1]{\citeN{#1}}

\begin{document}

%% Title information
\title[Adventures in Monitorability]{Adventures in Monitorability:\\ From Branching to Linear Time and Back Again}         %% [Short Title] is optional;
                                        %% when present, will be used in
                                        %% header instead of Full Title.
% \titlenote{with title note}             %% \titlenote is optional;
                                        %% can be repeated if necessary;
                                        %% contents suppressed with 'anonymous'
%\subtitle{From Branching to Linear Time and Back Again}                     %% \subtitle is optional
% \subtitlenote{with subtitle note}       %% \subtitlenote is optional;
                                        %% can be repeated if necessary;
                                        %% contents suppressed with 'anonymous'

%% Author information
%% Contents and number of authors suppressed with 'anonymous'.
%% Each author should be introduced by \author, followed by
%% \authornote (optional), \orcid (optional), \affiliation, and
%% \email.
%% An author may have multiple affiliations and/or emails; repeat the
%% appropriate command.
%% Many elements are not rendered, but should be provided for metadata
%% extraction tools.

%% Author with single affiliation.
\author[L. Aceto]{Luca Aceto}
%\authornote{with author1 note}          %% \authornote is optional;
                                        %% can be repeated if necessary
\orcid{0000-0002-2197-3018}             %% \orcid is optional
\affiliation{
%  \position{Position1}
%  \department{Department1}              %% \department is recommended
  \institution{Gran Sasso Science Institute}            %% \institution is required
%  \streetaddress{Street1 Address1}
  \city{L'Aquila}
%  \state{State1}
%  \postcode{Post-Code1}
  \country{Italy}                    %% \country is recommended
}
\affiliation{
	%  \position{Position2a}
	\department{School of Computer Science}             %% \department is recommended
	\institution{Reykjavik University}           %% \institution is required
	%  \streetaddress{Street2a Address2a}
	\city{Reykjavik}
	%  \state{State2a}
	%  \postcode{Post-Code2a}
	\country{Iceland}                   %% \country is recommended
}
\email{luca@ru.is}          %% \email is recommended

%% Author with two affiliations and emails.
\author[A. Achilleos]{Antonis Achilleos}
%\authornote{with author2 note}          %% \authornote is optional;
                                        %% can be repeated if necessary
\orcid{0000-0002-1314-333X}             %% \orcid is optional
\affiliation{
%  \position{Position2a}
  \department{School of Computer Science}             %% \department is recommended
  \institution{Reykjavik University}           %% \institution is required
%  \streetaddress{Street2a Address2a}
  \city{Reykjavik}
%  \state{State2a}
%  \postcode{Post-Code2a}
  \country{Iceland}                   %% \country is recommended
}
\email{antonios@ru.is}         %% \email is recommended
%\affiliation{
%  \position{Position2b}
%  \department{Department2b}             %% \department is recommended
%  \institution{Institution2b}           %% \institution is required
%  \streetaddress{Street3b Address2b}
%  \city{City2b}
%  \state{State2b}
%  \postcode{Post-Code2b}
%  \country{Country2b}                   %% \country is recommended
%}
%\email{first2.last2@inst2b.org}         %% \email is recommended
\author[A. Francalanza]{Adrian Francalanza}%adrian
%\authornote{with author2 note}          %% \authornote is optional;
%% can be repeated if necessary
\orcid{0000-0003-3829-7391}             %% \orcid is optional
\affiliation{
	%  \position{Position2a}
	\department{Department of Computer Science, ICT}             %% \department is recommended
	\institution{University of Malta}           %% \institution is required
	%  \streetaddress{Street2a Address2a}
	\city{Msida}
	%  \state{State2a}
	%  \postcode{Post-Code2a}
	\country{Malta}                   %% \country is recommended
}
\email{adrian.francalanza@um.edu.mt}         %% \email is recommended
\author[A. Ing\'{o}lfsd\'{o}ttir]{Anna Ing\'{o}lfsd\'{o}ttir}%anna
%\authornote{with author2 note}          %% \authornote is optional;
%% can be repeated if necessary
\orcid{0000-0001-8362-3075}             %% \orcid is optional
\affiliation{
	%  \position{Position2a}
	\department{School of Computer Science}             %% \department is recommended
	\institution{Reykjavik University}           %% \institution is required
	%  \streetaddress{Street2a Address2a}
	\city{Reykjavik}
	%  \state{State2a}
	%  \postcode{Post-Code2a}
	\country{Iceland}                   %% \country is recommended
}
\email{annai@ru.is}         %% \email is recommended
\author[K. Lehtinen]{Karoliina Lehtinen}%karoliina
%\authornote{with author2 note}          %% \authornote is optional;
%% can be repeated if necessary
%\orcid{0000-0000-0000-0000}             %% \orcid is optional --- cannot find it; K, it's up to you
\affiliation{
	%  \position{Position2a}
	%\department{}             %% \department is recommended
	\institution{Kiel University}           %% \institution is required
	%  \streetaddress{Street2a Address2a}
	\city{Kiel}
	%  \state{State2a}
	%  \postcode{Post-Code2a}
	\country{Germany}                   %% \country is recommended
}
\affiliation{
	%  \position{Position2a}
	%\department{}             %% \department is recommended
	\institution{University of Liverpool}           %% \institution is required
	%  \streetaddress{Street2a Address2a}
	\city{Liverpool}
	%  \state{State2a}
	%  \postcode{Post-Code2a}
	\country{United Kingdom}                   %% \country is recommended
}
\email{k.lehtinen@liverpool.ac.uk}         %% \email is recommended

%\shortauthors{Luca Aceto, Antonis Achilleos, Adrian Francalanza, } 
%%does not work well

%% Abstract
%% Note: \begin{abstract}...\end{abstract} environment must come
%% before \maketitle command
\begin{abstract}
%

%
% First place.
%
This paper establishes a comprehensive theory of runtime monitorability for Hennessy-Milner logic with recursion, a very expressive variant of the modal $\mu$-calculus.
It investigates the monitorability of that logic with a linear-time semantics and then compares the obtained results with ones that were previously presented in the literature for a branching-time setting.
Our work establishes an expressiveness hierarchy of monitorable fragments of Hennessy-Milner logic with recursion in a linear-time setting and exactly identifies what kinds of guarantees can be given using runtime monitors for each fragment in the hierarchy.
Each fragment is shown to be complete, in the sense that
it
can express all properties
that can be monitored under the corresponding guarantees.
The study is carried out using a principled approach to monitoring that connects the semantics of the logic and the operational semantics of monitors.
The proposed framework supports the automatic, compositional synthesis of correct monitors from monitorable properties.
\end{abstract}

%% 2012 ACM Computing Classification System (CSS) concepts
%% Generate at 'http://dl.acm.org/ccs/ccs.cfm'.
\begin{CCSXML}
	<ccs2012>
	<concept>
	<concept_id>10003752.10003790.10002990</concept_id>
	<concept_desc>Theory of computation~Logic and verification</concept_desc>
	<concept_significance>500</concept_significance>
	</concept>
	<concept>
	<concept_id>10003752.10003790.10003793</concept_id>
	<concept_desc>Theory of computation~Modal and temporal logics</concept_desc>
	<concept_significance>500</concept_significance>
	</concept>
	<concept>
	<concept_id>10003752.10003766.10003770</concept_id>
	<concept_desc>Theory of computation~Automata over infinite objects</concept_desc>
	<concept_significance>300</concept_significance>
	</concept>
	<concept>
	<concept_id>10003752.10003777.10003787</concept_id>
	<concept_desc>Theory of computation~Complexity theory and logic</concept_desc>
	<concept_significance>100</concept_significance>
	</concept>
	</ccs2012>
\end{CCSXML}

\ccsdesc[500]{Theory of computation~Logic and verification}
\ccsdesc[500]{Theory of computation~Modal and temporal logics}
\ccsdesc[100]{Theory of computation~Automata over infinite objects}
\ccsdesc[100]{Theory of computation~Complexity theory and logic}

%% End of generated code

%% Keywords
%% comma separated list
\keywords{monitorability, linear-time and branching-time logics, monitor synthesis}  %% \keywords are mandatory in final camera-ready submission

%% \maketitle
%% Note: \maketitle command must come after title commands, author
%% commands, abstract environment, Computing Classification System
%% environment and commands, and keywords command.
\maketitle

\section{Introduction}
\label{sec:intro}
\input{introduction.tex}
\section{Preliminaries}
\label{sec:preliminaries}
\input{preliminaries.tex}

\section{A monitoring framework}
\label{sec:monitors}
\input{monitors.tex}

\section{Monitorability for \ltmu}
\label{sec:ltmu-monitorability}
\input{monitorability.tex}

\section{Branching-Time Monitorability}
\label{sec:branchingtime}
\input{branching.tex}

\section{Conclusion}
\label{sec:conclusions}
\input{conclusion.tex}

%% Acknowledgments
\begin{acks}                            %% acks environment is optional
                                        %% contents suppressed with 'anonymous'
  %% Commands \grantsponsor{<sponsorID>}{<name>}{<url>} and
  %% \grantnum[<url>]{<sponsorID>}{<number>} should be used to
  %% acknowledge financial support and will be used by metadata
  %% extraction tools.
  The authors thank Orna Kupferman and Moshe Vardi for clarifying remarks about guarded formulae and pointers to the literature.
  We thank the anonymous reviewers for their thorough reading of our paper and their insightful comments.

  This research was partially supported by the  projects  ``TheoFoMon: Theoretical Foundations for Monitorability'' (grant number: ~\grantnum{Icelandic Research Fund}{163406-051}) and ``Epistemic Logic for Distributed Runtime Monitoring'' (grant number: ~\grantnum{Icelandic Research Fund}{184940-051}) of the Icelandic Research Fund, by the BMBF project ``Aramis II'' (project number:~\grantnum{BMBF}{01IS160253}) and the EPSRC project ``Solving parity games in theory and practice'' (project number:~\grantnum{EPSRC}{EP/P020909/1}).

%  \grantsponsor{GS100000001}{National Science
 %   Foundation}{http://dx.doi.org/10.13039/100000001} under Grant
 % No.~\grantnum{GS100000001}{nnnnnnn} and Grant

\end{acks}

%% Bibliography
\bibliography{refs}

\newpage
%
%%% Appendix
\appendix
\part*{Appendix}
% \section{Appendix: Omitted Proofs}
% \label{sec:appendix}
\input{appendix.tex}

\end{document}

%% file: introduction.tex
% !TEX root = main.tex

%
% Doomsday scenario.
%
The ubiquitous proliferation of software---from
% in every aspect of society
 high-frequency stock market trading and autonomous vehicles, down to mundane objects such as mobile phones and household appliances---makes a strong case for stringent software correctness requirements.
This proliferation has also substantially altered the manner in which software is developed and deployed.
Today's software often consists of multiple components (\eg third-party libraries, mobile apps, microservices, cloud services \etc) that are
% typically
developed and maintained by independent software organisations.
In this setting,  access to the components' internal workings varies (\eg open-source versus  proprietary code) and
%they are often
different components may be
subject to diverse quality controls.
Moreover, time-to-market constraints often impose multiple deployment phases where software is rolled out \emph{in stages} and third-party components change without notice from one deployment phase to the next.
Requirements from various stakeholders may also evolve between deployment phases and occasionally become conflicting.
These realities
% and complexities
suggest that there is
% probably
no silver bullet for ensuring software correctness.
Any adequate solution will most likely need to employ \emph{multiple verification techniques}  (\eg testing, model checking, theorem proving, log analysis, type checking, monitoring \etc) in a coherent manner, spanning the various stages of the software development lifecyle.

Runtime Verification (RV) \cite{BartocciFFR:18:RVIntro} is
% the name generally used to refer to a collection of lightweight verification methods that check
% the term generally used for any
a lightweight verification technique that checks
for the correctness of the system under scrutiny by analysing the \emph{current execution} exhibited by the system.
RV generally assumes a logic (or some other formal language) for describing the correctness specifications of the system.
From these specifications, (online) RV generates computational entities called \emph{monitors} that are then instrumented to run with the system so as to \emph{incrementally} analyse its execution (expressed as a trace of captured events) and
reach \emph{(irrevocable) judgements} relating to system violations or satisfactions for these specifications.
These characteristics make RV an ideal candidate to be used in a  multi-pronged approach towards ensuring software correctness:
it can
% be employed to
verify the correctness of components that are either not available for inspection prior to deployment, or are too
expensive
% extensive
to check via more exhaustive
%(
and less scalable
%)
verification techniques such as model checking~\cite{Baier2008Book,Clarke1999Book}.
%
%There are also various economies of scale to be gained from having a unifying language (or set of languages) for specifying correctness properties, and then having separate concerns for determining the most appropriate set of techniques to use for each specification (\eg using automated test-case generation or monitor synthesis, encodings to standard static analysis tools  \etc).
%
Importantly, in settings where multiple verification techniques are used,
one cannot necessarily expect specifications to be expressed in a language tailored specifically to RV.
Indeed, the use of disparate specification logics
specific to every verification technique that is used for validating system correctness is expensive.
Moreover, an RV-specific property language leads to a poor separation of concerns between the effort required to formulate the specifications and the engineering endeavour needed to determine how to best verify them.
%
%Apart from being more expensive to use disparate specification logics for every verification technique used, an RV-specific property language leads to a poor separation of concerns between the effort required to formulate the specifications and the engineering endeavour required to determine how to best verify them.
% , nor can one expect one language to be appropriate for all use-cases.
%
%%%% alternatives: %%% I vhose the first, because it seems to supposrt the choice of recHML more
%We t
Therefore %advocate
it is
%only
natural and important
to
%the
develop
%ment of
RV foundations that are based on \emph{general-purpose specification languages}, which subsume application-specific verification concerns.
%
%We therefore advocate a separation of concerns between the specification language and the monitoring model.

In order for RV to be used effectively in this way, a few
% fundamental
foundational
questions need to be
addressed.
%understood.
%
Principal among them is the question of \emph{monitorability}:
for sufficiently expressive specification logics, it is often the case that
some specifications
%certain expressible specifications
\emph{cannot} be monitored at runtime.
For example, the observation of finite executions does not give sufficient information to decide whether the specification ``every request is eventually followed by an answer'' is satisfied.
%
% For this reason,
It is thus important to identify \emph{which} specifiable properties are monitorable and \emph{which are not}, since this directly impinges on whether to use RV or some other verification technique instead.
Another fundamental question is that of monitor \emph{correctness}.
Monitors are often considered part of the trusted computing base and any errors in their code could either invalidate the runtime analysis they perform or, even worse, compromise the execution of the system itself.
In order to ensure monitor correctness, one must first establish what it means for a monitor to
%\emph{adequately runtime-check for} a specification.
\emph{adequately verify} a specification at runtime.
In fact, there may be a number of plausible definitions for this notion, each contributing to different monitor implementations.
The question of what it means to adequately verify  a specification at runtime directly impacts the question of monitorability as well, and guides the design of algorithms for the synthesis of correct monitors from monitorable properties.
A third fundamental question concerns the limits of monitor \emph{expressiveness}.
After one has established the monitorability of a set of properties from a reasonably general specification logic,
it is important to know whether this set contains \emph{all} properties that can be expressed in the logic and can, at the same time, be monitored at runtime.
This is the question of \emph{maximality} of
the monitorable fragment of the specification language, and
its importance lies in the knowledge that
one can identify a logical sub-language that
syntactically characterises all monitorable properties:
syntactic characterisations of monitorable properties provide a core calculus for conducting further studies and facilitate tool construction.
%Depending on the constraints one imposes on the general specification language and the monitoring system,

In prior work~\cite{FrancalanzaAI15,FraAI:17:FMSD,AceAFI:17:FSTTCS,FrancalanzaAAACDMI:17:RV,AceAFI:18:FOSSACS},
these foundational questions have been investigated for a highly expressive logic called Hennessy-Milner Logic with recursion (\UHML)~\cite{Larsen:90:HMLRec},  a variant of the \UCalc~\cite{Koz:83:TCS}, that can embed a variety of widely used logics such as LTL and CTL, thus guaranteeing a good level of generality for the obtained results.
A distinctive aspect of this programme of study is the differentiation between
%the logical semantics on the one hand, and monitor operational semantics
the \emph{semantics of the logic} on the one hand, and the \emph{operational semantics of monitors}
on the other, which mirrors the separation of concerns required for the multi-pronged verification approach advocated
%for
earlier.
%
%The definitions of monitorability and correctness emerge naturally as relationships between the two semantics, that is, the relationship between the verdicts reached by a monitor and the valuation of the specification it is meant to monitor for.
%retry:
Within the proposed framework, the definitions of monitorability and correctness emerge naturally as relationships between the two semantics.
That is, the relationship between the verdicts reached by a monitor and the satisfaction of a specification by the observed system naturally characterises both the monitor's correctness and the specification's monitorability.

Despite its merits, that body of work remains rather disconnected from the more established classical results on monitorability~\cite{MannaPnueli:91:TCS,ChangMannaPnueli:92:ALP,PnueliZaks:06:FM,BauerLeuckerSchallhart:10:LandC,
FalconeFernandezMounier:STTT:12}.
One major complication obstructing a unified understanding of all these monitoring theories is the fact that the former work on \UHML is carried out for a \emph{branching-time} semantics, whereas the classical theories target specifications for a \emph{linear-time} semantics.
%
% Opportunely,
Propitiously, however, the \UCalc also has a well-established linear-time semantics, which can be easily adapted to \UHML.
This provides us with an opportunity to extend the principled
% programme of study
framework developed in \citeMac{AceAFI:17:FSTTCS} and \citeMac{FrancalanzaAI15,FraAI:17:FMSD}  to a linear-time setting, offering an ideal
basis
%foundation
to better understand the connections between monitorability for branching-time and linear-time specifications.
We
%further see
contend that
this framework
%as
is
general enough to lay the foundations for a potential unified theory of monitorability.

\paragraph{Contributions and Synopsis.}  This paper sets out to establish a comprehensive theory of monitorability for \UHML, by investigating the monitorability of
that logic with a linear-time semantics and then comparing the obtained results  with those presented in the literature in a branching-time setting.
%
%the linear-time \UHML variant and then compare it with that for the branching-time \UHML.
We identify the trade-offs between monitoring guarantees and expressiveness:
In general, the more we expect from monitors, the fewer specifications can be monitored. Here we establish an expressiveness hierarchy within  linear-time \UHML and identify exactly what kind of guarantees can be given for each type of specification.
\begin{itemize}
\item We show that, compared to branching time, linear time allows for a much stronger notion
%on monitorability that requires a monitor that reaches a correct verdict on all executions.
of monitorability requiring that a monitor correctly report both the satisfaction and the violation of the property it checks on all system executions.
We identify a fragment of \UHML that captures exactly linear-time properties with such monitors (\Cref{prop:hml-maximal}), and show how to
synthesise monitors from them
%synthesize them xc
(\Cref{def:mon-synt-complete}).
\item
For any collection of monitors with irrevocable acceptance and rejection verdicts, which are reported after examining a finite prefix of the observed execution, we show a strong maximality result for the above-mentioned logical fragment (\Cref{thm:stronger-HML-maximality}), which guarantees that all monitorable properties of traces can be expressed in that fragment of \UHML.
\item We apply the weaker notion of monitorability called \textit{partial} monitorability from \citeMac{FraAI:17:FMSD}, which guarantees that a monitor does not reach an incorrect verdict and reaches a verdict for either all violations or all satisfactions. Again, we give a syntactic characterisation of linear-time properties that can be monitored with such monitors (\Cref{prop:linear-partial-maximality}), we show how to
synthesise correct monitors from them
%synthesize them
(\Cref{def:formula-to-monitor-part}), and prove maximality results.
%with respect to the linear-time properties from \UHML.
%
\item We establish a
% semantic
relationship between specifications that are partially monitorable in branching-time and in linear-time semantics (\Cref{sec:branchingtime}). To establish this result, we
%establish
study
how considering specifications over \textit{both} finite and infinite executions affects monitorability.
Our main observation here is that the syntactic fragment identified as partially monitorable with respect to branching-time semantics and the one identified as partially monitorable with respect to linear-time semantics are equally expressive under linear-time semantics over a finite set of actions. This bridges the gap in the treatment of monitorability on linear- versus branching-time domains.
%
%\item We coin the term ``finfinite'' that will change all science forever.
\end{itemize}

Our results establish a unified foundation for an increasingly important verification technique, covering \emph{both} branching-time and linear-time specifications.
We establish simple syntactic characterisations for specifications that can be monitored at runtime for various monitor requirements.
% allow for different types of monitors, and, in
For each characterisation, we provide a synthesis function that automates the generation of the corresponding monitors, whose correctness proofs depend on delicate arguments about the monitor semantics.
%
% of which the correctness proofs depend on delicate arguments about monitor semantics
%
This approach facilitates the design and implementation of correct monitors,
% with correctness guarantees,
along the lines of previous work on tool construction~\cite{Attard:16:RV,Attard:17:Book,FrancalanzaS15}.
Throughout our technical development, we also highlight
% some of the more
the subtle aspects of moving between semantics of branching processes, infinite traces, and potentially finite traces, and provide ample discussion on how they affect monitorability.
Crucially, our results are not just limited to our line of work.
For instance, the syntactic characterisations of monitorable properties set maximality limits to a number of existing RV tools using popular logics such as LTL since these logics can be embedded in our general language \UHML.

% the myriad of existing RV tools based on logics like LTL
% they also set maximality limits for

The proofs of all the results in the paper may be found in the 
appendix.

%% file: preliminaries.tex
% !TEX root = main.tex

\begin{figure}[!h]
  \textbf{Syntax}
  \begin{align*}
      \hV,\hVV \in \UHML &\bnfdef  \hTru && (\text{truth})&
                               &\bnfsepp  \hFls && (\text{falsehood})\\
             & \bnfsepp \hOr{\hV\,}{\,\hVV}  && (\text{disjunction}) &
             & \bnfsepp \hAnd{\hV\,}{\,\hVV}  && (\text{conjunction}) \\
             &\bnfsepp \hSuf{\ASet}{\hV} && (\text{possibility}) &
              &\bnfsepp \hNec{\ASet}{\hV} && (\text{necessity}) \\
              & \bnfsepp \hMinX{\hV} && (\text{min. fixpoint}) &
              & \bnfsepp \hMaxX{\hV} && (\text{max. fixpoint}) \\
              & \bnfsepp\; X && (\text{rec. variable})\\
    \end{align*}
    \\
  \textbf{Linear-Time Semantics}
     \[\begin{array}{rlrl}
      \hSemL{\hTru,\sigma}  & \!\!\!\deftxt   \Trc
      &
      % \\
      \hSemL{\hFls,\sigma}  & \!\!\!\deftxt   \emptyset
      \\
      \hSemL{\hOr{\hV_1}{\hV_2},\sigma} & \!\!\!\deftxt   \hSemL{\hV_1,\sigma} \cup \hSemL{\hV_2,\sigma}
      \qquad\qquad
       &
      % \\
      \hSemL{\hAnd{\hV_1}{\hV_2},\sigma} & \!\!\!\deftxt   \hSemL{\hV_1,\sigma} \cap \hSemL{\hV_2,\sigma}
      \\
      \hSemL{\hSuf{\ASet}{\hV},\sigma}  &
      \multicolumn{3}{l}{
          \!\!\!\deftxt \sset{\tr \;|\;
           \exists \trr \cdot \exists
          \acta\in \ASet
           \cdot\tr=\act\trr
          \;\text{ and }\; \trr \in \hSemL{\hV,\sigma}
          }
      }
    \\
    \hSemL{\hNec{\ASet}{\hV},\sigma}  &
    \multicolumn{3}{l}{
      \!\!\!\deftxt \sset{\tr \;|\;
       \forall \trr \cdot \forall
      \acta \in \ASet
       \cdot  \tr=\acta\trr
      \;\text{ implies }\;\trr \in \hSemL{\hV,\sigma}
       }
     }

    \\
    \hSemL{\hMin{\!\hVarX}{\hV},\sigma} & \!\!\!\deftxt
    \bigcap \sset{\TSet \;|\;  \hSemL{\hV,\sigma[\hVarX\mapsto \TSet]} \subseteq \TSet\ }

    \\
    \hSemL{\hMax{\!\hVarX}{\hV},\sigma} & \!\!\!\deftxt
    \bigcup \sset{\TSet \;|\;  \TSet \subseteq \hSemL{\hV,\sigma[\hVarX\mapsto \TSet]}\ }
    \quad\;
    &
    % \\
    \hSemL{\hVarX,\sigma} & \!\!\!\deftxt   \sigma(\hVarX)
    \\
    \\
    \end{array}
    \]
  \\
  \textbf{Branching-Time Semantics}
  \[
  \begin{array}{rlrl}
    \hSemB{\hTru,\rho}  & \!\!\!\deftxt \Proc
      &
    \hSemB{\hFls,\rho}  & \!\!\!\deftxt \emptyset
    \\
    \hSemB{\hOr{\hV_1}{\hV_2},\rho} &
    \!\!\!\deftxt \hSemB{\hV_1,\rho} \cup \hSemB{\hV_2,\rho}
    \qquad\qquad
    &
    \hSemB{\hAnd{\hV_1}{\hV_2},\rho} &
    \!\!\!\deftxt \hSemB{\hV_1,\rho} \cap \hSemB{\hV_2,\rho}
    \\
    \hSemB{\hSuf{\ASet}{\hV},\rho}  &
    \multicolumn{3}{l}{
      \!\!\!\deftxt
          \sset{p \;|\; \exists q\cdot \exists \acta\in\ASet \cdot
          p \wtraS{\acta} q   \;\text{ and }\; q \in \hSemB{\hV,\rho} }
    }
    \\
    \hSemB{\hNec{\ASet}{\hV},\rho}  &
    \multicolumn{3}{l}{
      \!\!\!\deftxt
        \sset{p \;|\; \forall q\cdot \forall \acta \in \ASet \cdot p  \wtraS{\act} q \;\text{ implies }\; q \in \hSemB{\hV,\rho} }
    }
    \\
    \hSemB{\hMin{\!\hVarX}{\hV},\rho} &
    \!\!\!\deftxt
    \bigcap
    \sset{\PSet \;|\;  \hSemB{\hV,\rho[\hVarX\mapsto \PSet]} \subseteq \PSet}
    \\
    \hSemB{\hMax{\!\hVarX}{\hV},\rho} &
    \!\!\!\deftxt
    \bigcup
    \sset{\PSet \;|\;  \PSet \subseteq \hSemB{\hV,\rho[\hVarX\mapsto \PSet]}}
    % % \quad\;
    &
    \hSemB{\hVarX,\rho} &
    \!\!\!\deftxt \rho(\hVarX)
    \end{array}\]
  \caption{\UHML Syntax, Linear-Time and Branching-Time Semantics}
  \label{fig:recHML}
\end{figure}
We provide a brief overview of our touchstone logic, \UHML~\cite{Larsen:90:HMLRec,AceILS:2007}, a reformulation of the highly expressive and extensively studied \UCalc~\cite{Koz:83:TCS}.
% that has been shown to embed a number of widely used specification logics such as LTL and CTL.
%
\subsection{The Syntax}
The logic described in \Cref{fig:recHML} is a mild generalisation of \UHML  \cite{Larsen:90:HMLRec,AceILS:2007}.
It assumes a
% finite
set of actions, $\acta,\actb,\ldots
% , a, b, \ldots
\in \Act$, together with a distinguished \emph{internal} action $\tau$, where $\tau\not\in \Act$.
 We refer to the actions in $\Act$ as \emph{external} actions, as opposed to the
 % silent
action $\tau$, and use ${\actu \in \Act\cup\sset{\tau}}$ to refer to either.
The metavariables $\ASet,\BSet,\ldots \subseteq \Act$ range over sets of (external) actions, where the convenient notation $\coASet$ is occasionally used to denote $\Act\setminus \ASet$;
whenever the context allows us to do so unambiguously, singleton sets $\sset{\acta}$ are also occasionally denoted as \acta, and $\overline{\{\acta \}}$ is occasionally denoted as  $\coact{\acta}$.

The grammar in \Cref{fig:recHML} also assumes a countable set of logical variables $\hVarX,\hVarY \in \LVars$.
Apart from the standard constructs for truth, falsehood, conjunction and disjunction, the logic is equipped with
% possibility and necessity
existential and universal modal operators that use sets of actions, \ASet.
A hallmark of the logic is the use of \emph{two} recursion operators  that express least or greatest fixpoints:
formulae \hMinX{\hV} and  \hMaxX{\hV}  bind free instances of the logical variable \hVarX in \hV, inducing the usual notions of open/closed formulae and  formula equality up to alpha-conversion.
A formula is said to be guarded if every fixpoint variable appears within the scope of a modality within its fixpoint binding.
All formulae are assumed to be guarded (without loss of expressiveness~\cite{kupferman00}).
For a formula $\hV$, we use $l(\hV)$ to denote the length of $\hV$ as a string of symbols.

%\af{Luca asked to restate the actual result from \cite{kupferman00}.  When I checked, it however stated that ``THEOREM 2.1. Given a \actu=calculus formula, we can construct, \emph{\underline{in linear time}}, an equivalent guarded formula.''}

\subsection{The models}

We provide
%\emph{two}
linear- and branching-time
interpretations for the logic.
%The first semantics is defined over \emph{traces}, \Trc, that abstractly represent system runs.
%%
%We refer to it as the linear-time semantics.
%
%
The metavariables $\tV,\tVV{\in}\Trc=\Act^\omega$ range over \emph{infinite} sequences of external actions, abstractly representing complete system runs;  the metavariable $\TSet \subseteq \Trc$ ranges over \emph{sets} of traces.
\emph{Finite traces}, denoted as $\ftV,\ftVV \in \Act^\ast$, represent \emph{finite} prefixes of a system run or finite executions.
\emph{Explicit traces}, denoted as $\etV,\etVV \in {(\Act \cup \{\tau\})}^\ast$, represent  \emph{detailed finite} prefixes of a system run that also include its internal transitions;
the function \filter{\etV} returns the finite trace \ftV\ that is left after dropping all the \actt-actions from \etV.
We say that two explicit traces \emph{agree on the external actions}, denoted as $\etV_1 \equiv_\Act \etV_2$, whenever $\filter{\etV_1}=\filter{\etV_2}$.
A trace (\resp finite trace) with action \acta at its head is denoted as $\acta\tV$ (\resp $\acta\ftV$).
An explicit trace with action \actu at its head is denoted as $\actu\etV$.
%
% An explicit trace with action $\tau$ at its head is denoted as $\tau\etr$.
%
Similarly, a trace with a prefix \ftV\ and continuation \tV\ is denoted as $\ftr\tr$.
%

% if, after removing all occurences of $\tau$, they both become the same trace, \ie $\etr$ and $\etr'$ agree on the external actions if $\etr \equiv_\Act \etr'$, where $\equiv_\Act$ is the smallest equivalence on explicit traces, such that
% for every $\etrr,\etrr'$,  if $\etrr \equiv_\Act \etrr'$, then $\etrr \equiv \tau \etrr'$ and $\act \etrr \equiv_\Act \act \etrr'$.

The denotational semantic function $\hSemL{-}$  in \Cref{fig:recHML} maps a formula to a set of traces, and is referred to as the \emph{linear-time} semantics of \UHML.
It uses valuations that map logical variables to sets of traces,
$\sigma: \LVars \to \powset{\Trc}$, to define the semantics by induction on the structure of the formulae.
Intuitively, $\sigma(\hVarX)$ is the set of traces assumed to satisfy \hVarX.
The cases for truth, falsehood, disjunction  and conjunction are straightforward.
An
% possibility
existential modal formula \hSuf{\ASet}{\hV} denotes all traces with a prefix action $\acta$ from the action set \ASet and a continuation that satisfies \hV.
A universal modal
% necessity
formula \hNec{\ASet}{\hV} denotes all traces that are either \emph{not} prefixed by any \acta in \ASet, or
% else traces with
have a continuation \tVV\ satisfying \hV.
The sets of traces satisfying the least and greatest fixpoint formulae, \hMinX{\hV} and  \hMaxX{\hV}, are  defined as intersection (\resp union) of all the pre-fixpoints (\resp post-fixpoints) of the function induced by the formula \hV.
%
%Finally, variables \hVarX are assigned the semantics mapped by the valuation.
 % $\sigma$.

%
The second interpretation of \UHML, denoted by \hSemB{-}, is defined in terms of \emph{processes}, \Proc, and is referred to as the branching-time semantics.
It assumes a set of process states, $p,q,\ldots \in \Proc$ where $\PSet\subseteq \Proc$,  and a transition relation, $\reduc \subseteq (\Proc\times(\Act\cup\sset{\tau})\times \Proc)$.
The triple $\langle \Proc,(\Act\cup\sset{\tau}),\reduc\rangle$ forms a Labelled Transition System (LTS)~\cite{Keller:1976:CACM}.
The suggestive notation $p \traS{\actu} p'$ denotes $(p,\actu,p') \in\ \reduc$;
we also write $p \traSN{\actu}$ to denote $\neg(\exists p'\cdot \;p\traS{\actu}p')$.
We employ the usual notation for weak transitions and write $p \wreduc p'$ in lieu of $p (\traS{\tau})^{\ast} p'$ and $p \wtraS{\actu} p'$ for
$p \wreduc\cdot\traS{\actu}\cdot\wreduc p'$, referring to $p'$ as a \actu-derivative of $p$.
As we have done for strong transitions, for weak transitions we use
$p \wtraS{\actu}$ to denote $\exists p'\cdot \;p\wtraS{\actu}p'$
and
$p \centernot{\wtraS{\actu}}$ to denote $\neg(\exists p'\cdot \;p\wtraS{\actu}p')$.
%
% We also abuse notation and write $p \reduc p'$ instead of $p \traS{\actu} p'$ for some \actu, whenever the action is not important.
%
Sequences of weak transitions
$p\wtra{\act_1}\cdots\wtra{\act_n} p'$ are written as $p \wtraS{\ftV} p'$,
where $\ftV = \act_1 \cdots \act_n$.
Similarly, for strong transitions, $p\tra{\actu_1}\cdots\tra{\actu_n} p'$ is written as $p \etraS{\etV} p'$,
where $\etV = \actu_1 \cdots \actu_n$. We say that $p$ produces a trace $t=\act_1\act_2\cdots$  if there are processes $p_0, p_1, p_2, \ldots$ such that $p = p_0$ and $p_0 \wtraS{\act_1} p_1 \wtraS{\act_2} p_2 \cdots$.
While an LTS can be used to model a single system,
it can also model all possible system behaviours.

The branching-time semantics in \Cref{fig:recHML} follows the linear-time semantics for most cases, using a valuation from variables to \emph{sets of processes}, $\rho: \LVars \to \powset{\Proc}$, instead.
The main differences are \wrt the modal formulae.
A
% necessity
universal modal
formula \hNec{\ASet}{\hV} requires \emph{all} $\acta$-derivatives of a   process, where $\acta \in \ASet$, to satisfy \hV.
By contrast, an
% possibility
existential modal formula \hSuf{\ASet}{\hV} requires the \emph{existence} of at least one $\acta$-derivative, for some $\acta\in\ASet$, that satisfies \hV.

%We refer to \UHML with a linear-time interpretation as \ltmu, whereas that with a branching-time interpretation as \recHML.
%
For closed formulae, we use \hSemL{\hV} and \hSemB{\hV} in lieu of \hSemL{\hV,\sigma} and \hSemL{\hV,\rho}  (for some $\sigma$ and $\rho$) \resp since the semantics is independent of the valuation.
We also write $\hSem{\hV}$ instead of \hSemL{\hV} or \hSemB{\hV},  whenever the correct interpretation can be
% unambiguously
discerned from the context or the specific interpretation is unimportant.
Unless otherwise stated, we
% shall
assume that the formulae we consider are all \emph{closed}.
%
% ADRIAN: Commented for now
%
% In the case of \ltmu, we say that a finite trace $\ftr$ satisfies \hV, denoted as $\ftr \in \hSem{\hV}$, whenever $\forall \tr \cdot \ftr\tr \in \hSem{\hV}$. Similarly, $\ftr \not\in \hSem{\hV}$, whenever $\forall \tr \cdot \ftr\tr \not\in \hSem{\hV}$.

\begin{exmp}[Expressiveness] \label{ex:rechml-expressivity}
  For arbitrary formulae $\hV,\hVV \in \ltmu$, we can encode the following characteristic  LTL operators \cite{Clarke1999Book} as:
  \begin{align}
    \textsf{X}\,\hV & {\deftxt} \hSuf{\Act}{\hV}
    &
    \hV\,\textsf{U}\,\hVV & {\deftxt} \hMinY{\bigl(\hOr{\hVV}{(\hAnd{\hV}{\hSuf{\Act}\hVarY})}\bigr)}
    &
    \hV\,\textsf{R}\,\hVV & {\deftxt} \hMaxY{\bigl(\hOr{(\hAnd{\hVV}{\hV})}{(\hAnd{\hVV}{\hSuf{\Act}\hVarY})}\bigr)} \tag*{\qedd}
  \end{align}
\end{exmp}

\begin{exmp}[Comparison] \label{ex:lin-vs-bra-rechml}
   % An short example to illustrate the nuances, between a linear-time and a branching-time interpretation of a formula.
   %
   Assume $\Act = \sset{a,b,c}$.
   Consider the two formulae
   \begin{align*}
      \hV_1 &=  \hNec{a}{\hNec{a}{\hFls}}
      &
      \hV_2 &= \hNec{a}{(\hOr{\hSuf{a}{\hTru}}{\hSuf{\sset{b,c}}{\hTru}})}
   \end{align*}
   together with
   the trace (denoted by the $\omega$-regular expression)
   \begin{math}
     \tV = (a.b)^\omega
   \end{math},
   and the (non-deterministic) process (described by the regular CCS syntax \cite{milner:89:CCS})
   \begin{math}
     p  = \rec{x}{(\ch{\prf{a}{\prf{b}{x}}}{\ch{\prf{a}{\prf{a}{x}}}{
     \prf{a}{\nil}}})}
   \end{math}.
   % \begin{align*}
   %   \tV &=
   %     (a.b)^\ast
   %    &
   %    p & =
   %    \rec{x}{\bigl(
   %      \ch{
   %        \prf{a}{\prf{b}{x}}
   %      }{
   %        \ch{
   %          \prf{a}{\prf{a}{x}}
   %        }{
   %          \prf{a}{\nil}
   %        }
   %      }
   %    \bigr)}
   % \end{align*}
   %
   In particular, we note that $p$ can produce the infinite trace \tV.

    Whereas $\tV \in \hSemL{\hV_1}$, we have $p \not\in \hSemB{\hV_1}$ because along one branch we have $p \wtraS{a} \prf{a}{p}$ and $\prf{a}{p} \not\in \hSemB{\hNec{a}{\hFls}}$.
    In linear-time semantics, the equality
    $\hSemL{\hOr{\hSuf{\ASet}{\hTru}}{\hSuf{\coASet}{\hTru}}} = \hSemL{\hTru}$ holds for \emph{each} \ASet.
    One can also easily deduce that $\hSem{\hNec{\ASet}{\hTru}} = \hSem{\hTru}$  for both linear- and branching-time semantics from the semantics of \Cref{fig:recHML}.
    Hence, in our case (where $\Act = \sset{a,b,c}$), we obtain  $\hSemL{\hOr{\hSuf{a}{\hTru}}{\hSuf{\sset{b,c}}{\hTru}}}  = \hSemL{\hTru}$ by instantiating
    $\hSemL{\hOr{\hSuf{\ASet}{\hTru}}{\hSuf{\coASet}{\hTru}}} = \hSemL{\hTru}$ with $\ASet = \sset{a}$.
    As a result,
%    the semantics of
    $\hV_2$
%    collapses
    is equivalent
    to \hTru under linear-time semantics and we have $\tV {\in} \hSemL{\hV_2}$ for \emph{every} trace \tV.
    However, under branching-time semantics $\hSemB{\hOr{\hSuf{a}{\hTru}}{\hSuf{\sset{b,c}}{\hTru}}} \neq \hSemB{\hTru}$
    (one witness for the inequality is the deadlocked process \nil,
    $\hSemB{\hTru} \ni \nil \not\in \hSemB{\hOr{\hSuf{a}{\hTru}}{\hSuf{\sset{b,c}}{\hTru}}}$). In fact, $p \traS{a} \nil$ and thus $p \not\in \hSemB{\hV_2}$.   \qedd
\end{exmp}

\begin{rem}
  Action sets \ASet in \hNec{\ASet}{\hV} and \hSuf{\ASet}{\hV} are typically
  expressed using
%  denoted
%  as predicates in related tools
%  as
  predicates in tools such as those described in
  \citeMac{Attard:16:RV}, \citeMac{Attard:17:Book} and \citeMac{AcetoCFI18}.
  For example, modalities can be labelled by an output action on port $x$ carrying payload $\langle 8,y\rangle$ where the
  data variables $x$ and $y$ are constrained by conditions, as in
%  \eg an output action where the data variables $x$ and $y$ are constrained by
%  conditions
  $\hNec{ \textbf{out}(x,\langle 8,y\rangle),
  (192.188.34.42 \geq x \geq 192.188.34.1) \wedge  \textbf{mod}(y) = 1
  }{\hV}$.
  In the sequel, we shall assume that \Act (and thus any action set \ASet) is a \emph{finite} set of actions.
  This helps to simplify our technical development and enables us to focus on the core issues being studied.
  However, finite action sets are not necessarily a limitation since, in most cases,  infinite data sets can be treated in a finite manner using standard symbolic techniques
%  and equivalence classes
  (\eg see \citeMac{Fra:17:CONCUR} for a recent treatment of the subject in the context of monitors). \qedd
\end{rem}

\begin{rem}
  For a finite set $I$ of indices, the (standard) notation $\bigwedge_{i \in I} \hV_i$  denotes $\hTru$ when $I = \emptyset$, and
  a conjunction of the formulae in $\{ \hV_i \mid i \in I \}$ when $I \neq \emptyset$.
  Similarly $\bigvee_{i \in I} \hV_i$ denotes $\hFls$ when $I = \emptyset$, and
  a disjunction of the formulae in $\{ \hV_i \mid i \in I \}$ when $I \neq \emptyset$.
  These notations are justified by the fact that $\vee$ and $\wedge$ are commutative and associative with respect to all the semantics considered in the paper.
  We also observe that, for both semantics, $\hNec{A}{\hV}$ is equivalent to $\bigwedge_{\act \in A}\hNec{\act}\hV$, and $\hSuf{A}{\hV}$ is equivalent to $\bigvee_{\act \in A}\hSuf{\act}\hV$ for finite $A$, so we use these equivalent notations interchangeably.
  \qedd
\end{rem}

%% file: monitors.tex
% !TEX root = main.tex

% Monitors, instrumentation and synthesis

\begin{figure}[t]
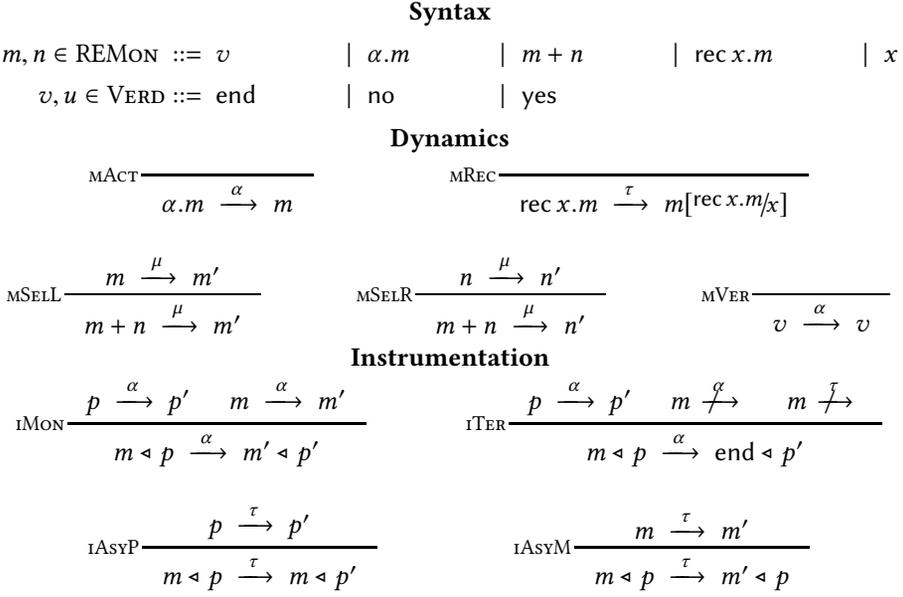

  % \centering
  \textbf{Syntax}
  \begin{align*}
    \mV,\mVV\in\REMon\ &\bnfdef\  \vV
    && \bnfsepp \prf{\acta}{\mV}
    && \bnfsepp \ch{\mV}{\mVV}
    && \bnfsepp \rec{x}{\mV}
    &&\bnfsepp x \\
% & \bnfsepp \yes && (\text{satisfaction}) && \bnfsepp \no && (\text{violation})\\
    \vV,\vVV \in \Verd & \bnfdef\ \stp
    &&  \bnfsepp \no
    && \bnfsepp \yes
  \end{align*}
  \textbf{Dynamics}
   \begin{mathpar}
 \inference[\rtit{mAct}]{
 % \acta\in\ASet
 }{\prf{\acta}{\mV}  \traSS{\acta} \mV}
 \and
 %
 % \inference[\rtit{mRec}]{\mV\subS{\rec{x}{\mV}}{x} \traSS{\acta} \mV'}{\rec{x}{\mV} \traSS{\acta} \mV'}
\inference[\rtit{mRec}]{}{\rec{x}{\mV} \traSS{\tau} \mV\subS{\rec{x}{\mV}}{x}}
  \\\\\\
 % \and
   \inference[\rtit{mSelL}]{\mV \traSS{\actu} \mV'}{\ch{\mV}{\mVV}  \traSS{\actu} \mV'}
   \and
  \inference[\rtit{mSelR}]{\mVV \traSS{\actu} \mVV'}{\ch{\mV}{\mVV}  \traSS{\actu} \mVV'}
  \and
  \inference[\rtit{mVer}]{
  % \verd{\mV}
  }{\vV \traSS{\acta} \vV}
\end{mathpar}
 \textbf{Instrumentation}
 % \begin{mathpar}
 %   \inference[\rtit{iMon}]{p \traSS{\acta} p' & \mV \traSS{\acta} \mV'}{\sys{\mV}{p} \traSS{\acta} \sys{\mV'}{p'}} \and
 %   \inference[\rtit{iTer}]{p \traSS{\acta} p' & \mV \traSSN{\acta}  % \mV' & \neg\verd{\mV}
 %   }{\sys{\mV}{p} \traSS{\acta} \sys{\nil}{p'}}
 % \end{mathpar}
 \begin{mathpar}
   \inference[\rtit{iMon}]{p \traSS{\acta} p' & \mV \traSS{\acta} \mV'}{\sys{\mV}{p} \traSS{\acta} \sys{\mV'}{p'}}
  \and
\inference[\rtit{iTer}]{p \traSS{\acta} p' & \mV \traSSN{\acta}  & \mV \traSSN{\tau} }{\sys{\mV}{p} \traSS{\acta} \sys{\stp}{p'}}
\\\\\\
\inference[\rtit{iAsyP}]{p \traSS{\tau} p'}{\sys{\mV}{p} \traSS{\tau} \sys{\mV}{p'}}
   \and
\inference[\rtit{iAsyM}]{\mV \traSS{\tau} \mV'}{\sys{\mV}{p} \traSS{\tau} \sys{\mV'}{p}}
 \end{mathpar}
  \caption{Monitors and Instrumentation}
  \label{fig:monit-instr}
\end{figure}

A distinctive feature of the work in \citeMac{AceAFI:17:FSTTCS,AceAFI:18:FOSSACS} and \citeMac{FraAI:17:FMSD} is the full description of the monitoring setup used, which incorporates the monitor definition together with the system instrumentation mechanism---monitor compositionality results have shown that the semantics of monitors in an instrumented setup differs substantially from that given for monitors in isolation~\cite{Fra:16:FOSSACS,Fra:17:CONCUR}.
Here we follow this comprehensive approach.

\subsection{Regular Monitors}

% \begin{defn}[Regular Monitors] \label{def:regular-monitors}
Regular monitors are LTSs defined by the grammar and transition rules in \cref{fig:monit-instr},
% \qedd
% \end{defn}
 used already in \citeMac{AceAFI:17:FSTTCS} and \citeMac{FraAI:17:FMSD}.
A transition $\mV \traS{\acta} \mVV$ denotes that the monitor in state \mV can \emph{analyse} the (external) action \acta and transition to state \mVV.
%
% ;this is lifted to sequences of actions $\mV \wtraS{\ftV} \mVV$ in the obvious way.
%
Monitors may reach any one of \emph{three} verdicts after analysing a finite trace: \emph{acceptance}, \yes, \emph{rejection}, \no, and the \emph{inconclusive} verdict \stp.
We highlight the transition rule for verdicts in \Cref{fig:monit-instr}, describing the fact that from a verdict state any action can be analysed by transitioning to the same state;  verdicts are thus \emph{irrevocable}.
The remaining constructs and transitions  are standard.
  If \emph{at most one} of the verdicts $\yes,\no$ appears in $\mV$, then $\mV$ is called a \emph{single-verdict monitor}.
Otherwise, $\mV$ is called a dual-verdict monitor.
Just like for formulae,
% for a monitor $\mV$,
we use  $l(\mV)$ to denote the length of $\mV$ as a string of symbols.
In the sequel, for a finite nonempty set of indices $I$, we use notation $\sum_{i \in I}\mV_i$ to denote a combination of the monitors in
$\{\mV_i \mid i \in I\}$
%$\setof{\mV_i}{i \in I}$
% sorry F, but it did looked weird
using the operator $+$.
The notation is justified, because $+$ is commutative and associative with respect to the transitions that a resulting monitor can exhibit.
We also use the shorthand notation \prf{\ASet}{\mV} to denote $\sum_{\acta{\in}\ASet} \prf{\acta}{\mV}$ (for finite non-empty \ASet).
The regular monitors in \Cref{fig:monit-instr} have an important property, namely that their state space, \ie the set of reachable states, is finite.
This is a valuable property for ensuring reasonable overheads in terms of the amount of memory the monitor will use at runtime (see \Cref{prop:reg-mon-fin-state}%
, whose proof is in  \Cref{sec:monitors-appendix}).

\begin{lem}[Verdict Persistence]\label{lem:ver-persistence}
  \begin{math}
    \vV \etraS{\etV} \mV \text{ implies } \mV=\vV.
  \end{math}\qed 
\end{lem}
% \begin{proof}
%   By structural induction on \etV.
% \end{proof}

\begin{defn}[Monitor Reachable States] \label{def:mon-reach}
  \begin{math}
    \reach{\mV} \deftxt \setof{\mVV}{ \exists \etV \cdot \mV \etraS{\etV} \mVV} .
  \end{math}
  \qedd
\end{defn}

\begin{prop} \label{prop:reg-mon-fin-state} Regular monitors are finite state \ie for all $\mV\in\REMon$, \reach{\mV} is finite.
%	systems.
	\qed
\end{prop}
% \begin{proof} Follows from \Cref{lem:state-space-charact-reach,lem:state-space-charact-bounded}
% from \Cref{sec:monitors-appendix}.
% \end{proof}

We define the following behavioural predicate on monitors, which
relates to
%could be associated with
their correctness.

\begin{defn}[Monitor Consistency] \label{def:consitency-totality}
\quad
%\begin{itemize}
 A monitor $\mV$ is \emph{consistent} when there is \emph{no} finite trace $\ftr$ such that  $\mV \wtraS{\ftr} \yes$ and $\mV \wtraS{\ftr} \no$.
 % , for every finite trace $\ftr$, it is \emph{not} the case that $\mV \wtraS{\ftr} \yes$ and $\mV \wtraS{\ftr} \no$.
%  Otherwise, we call $\mV$ inconsistent.
%  \item A monitor is $\mV$ \emph{total} when it reaches a verdict for \emph{every} infinite trace, \ie
%  for every infinite trace $\tV$, there is a prefix $\ftV$ such that $\tV=\ftV\tV'$ and
  % such that
%  $\mV \wtraS{\ftr} \yes$ or $\mV \wtraS{\ftr} \no$.
  \qedd
%\end{itemize}
\end{defn}

Monitors are intended to run in conjunction with the system (\ie process) they are analysing.
Following \citeMac{Fra:16:FOSSACS,Fra:17:CONCUR} and \citeMac{FraAI:17:FMSD}, \Cref{fig:monit-instr} defines a transition relation for a process $p$ instrumented with a monitor \mV, denoted as \sys{\mV}{p}.
The relation is parametric with respect to the transition semantics of the process $p$ and the monitor, as long as the latter includes the inconclusive verdict \stp\ (\eg the monitor transition semantics given in \Cref{fig:monit-instr} does).
The semantics relegates the monitor \mV to a \emph{passive role} in an instrumented system \sys{\mV}{p}, meaning that  \sys{\mV}{p} transitions with an external action \acta only when $p$ transitions with that action.
For instance, when $p$ transitions with action \acta to some $p'$, and \mV can analyse this action and transition to state $\mV'$, the instrumented pair transitions in lockstep to \sys{\mV'}{p'}; see rule \rtit{iMon}.
Conversely, if $p$ wants to transition with an action \acta that the instrumented monitor is \emph{not} able to analyse (perhaps due to underspecification), the instrumented system is \emph{still} allowed to transition with \acta, but the monitor analysis is prematurely aborted to the inconclusive state; see rule \rtit{iTer}.
The other rules allow monitors and processes to execute independently of one another with respect to internal (\actt-)moves.

\begin{exmp} \label{ex:mon-acc-n-rej}
%   A short example to illustrate the nuances, if any for monitor detection. Ideally, it builds/links with \Cref{ex:lin-vs-bra-rechml}.
When the monitor $\rec{x}{(\ch{\prf{a}{x}}{\prf{b}{\yes}})}$
is instrumented with the process $\prf{a}{\rec{x}{\prf{b}{x}}}$, it can reach an acceptance verdict
thus:
%should the process exhibit a trace of the form $ab\tV$:
%
   \begin{align*}
    &\textbf{}\sys{\rec{x}{(\ch{\prf{a}{x}}{\prf{b}{\yes}})}}
    {\prf{a}{\rec{x}{\prf{b}{x}}}}
    % \\
    % \quad
    \traS{\tau}
    % \quad
    \sys{(\ch{\prf{a}{(\rec{x}{(\ch{\prf{a}{x}}{\prf{b}{\yes}})})}}{\prf{b}{\yes}})}
    {\prf{a}{\rec{x}{\prf{b}{x}}}}
    % \quad
    \traS{a}
    % \quad
    \\
    & \sys{{\rec{x}{(\ch{\prf{a}{x}}{\prf{b}{\yes}}})}}
    {\rec{x}{\prf{b}{x}}}
%    \\
    % \quad
    \etraS{\actt\actt}
    % \cdot\traS{\actt}
    % \quad
    % \sys{{(\rec{x}{(\ch{\prf{a}{x}}{\prf{b}{\yes}}})})}
    % {\rec{x}{\prf{b}{x}}}
    % \\
    % &
    % \quad\traS{\actt}\quad
    \sys{\ch{\prf{a}{\rec{x}{(\ch{\prf{a}{x}}{\prf{b}{\yes}})}}}{\prf{b}{\yes}}}
    {\prf{b}{\rec{x}{\prf{b}{x}}}}
%    \textbf{}\\
    % \quad
    \traS{b}
    % \quad
    \sys{{\yes}}
    {{\rec{x}{\prf{b}{x}}}} .
     % \tag*{\qedd}
   \end{align*}
   However, if the same process is instrumented with a slightly different monitor $\rec{x}{(\ch{\prf{a}{\prf{a}{x}}}{\prf{b}{\yes}})}$
   we obtain  a different verdict.
%   for the same trace prefix.
   \begin{align*}
    &
    \sys
     {\rec{x}{
      (\ch
        {\prf{a}{\prf{a}{x}}}
        {\prf{b}{\yes}}
      )
      }}
    {\prf{a}{\rec{x}{\prf{b}{x}}}}
    % \\
    % \quad
    \traS{\actt}
    % \quad
    \sys{(
      \ch
      {\prf{a}{
        \prf{a}{
          (
            \rec{x}{
              (
                \ch{\prf{a}{\prf{a}{x}}}{\prf{b}{\yes}}
              )
            }
          )
        }
       }
      }
      {\prf{b}{\yes}}
    )}
    {\prf{a}{\rec{x}{\prf{b}{x}}}}
    % \quad
    \traS{a}
    % \quad
    \\
    & \sys{
    \prf{a}{\rec{x}{(\ch{\prf{a}{x}}{\prf{b}{\yes}}})}}
    {\rec{x}{\prf{b}{x}}}
%    \\
    % \quad
    \traS{\actt}
    %\quad
    \sys{
      \prf{a}{
        \rec{x}{(
        \ch
        {\prf{a}{x}}
        {\prf{b}{\yes}}
        )}
      }
    }
    {\prf{b}{\rec{x}{\prf{b}{x}}}}
%    \textbf{}\\
    % \quad
    \traS{b}
    % \quad
    \sys{{\stp}}
    {{\rec{x}{\prf{b}{x}}}}
     % \tag*{\qedd}
   \end{align*}
   The last transition is obtained via rule \rtit{iTer}, whereby the process exhibited an action  that the current monitor state was unable to analyse (\ie it could only analyse action $a$, not $b$).  \qedd
\end{exmp}

The following lemmata describe how the respective monitor and system LTSs  can be composed and decomposed according to instrumentation
% , and are due to
~\cite{Fra:16:FOSSACS,FraAI:17:FMSD}.

\begin{lem}[General Unzipping
% ~\cite{Fra:16:FOSSACS,FraAI:17:FMSD}
] \label{lem:unzipping}
  \sys{\mV}{p} \wtraS{\ftV} \sys{\mVV}{q} implies
  \begin{itemize}
  \item $ p \wtraS{\ftV} q$ and
  \item $\mV \wtraS{\ftV} \mVV$ or ($\exists\ftV_1,\ftV_2,\acta,\mV' \cdot \ftV=\ftV_1\acta\ftV_2$ and $\mV \wtraS{\ftV_1} \mV' \traSN{\actt}$ and $\mV' \traSN{\acta}$ and $\mVV=\stp$). \qed
  \end{itemize}
\end{lem}

\begin{lem}[Zipping
% ~\cite{Fra:16:FOSSACS,FraAI:17:FMSD}
] \label{lem:zipping}
  \begin{math}
    (p \wtraS{\ftV} q\text{ and }    \mV \wtraS{\ftV} \mVV)
  \end{math}
      implies
  \begin{math}
      {\sys{\mV}{p} \wtraS{\ftV} \sys{\mVV}{q}}
  \end{math}. \qed
\end{lem}

Within this framework, we
% are able to
can formalise our understanding of process and trace acceptance and rejection by a monitor.
Acceptances and rejections will constitute the monitoring counterpart to formula satisfactions and violations from \Cref{sec:preliminaries} when we consider our definitions of monitorability.

\begin{defn}[Process and Trace Acceptance and Rejection]  \label{def:acc-n-rej}
 A monitor \emph{\mV rejects $p$ along \ftr}, denoted as $\rej{\mV,p,\ftr}$, if $\sys{\mV}{p} \wtraS{\ftr} \sys{\no}{p'}$ for some $p'$.
 Similarly,
 \emph{\mV accepts $p$ along \ftr}, denoted as $\acc{\mV,p,\ftr}$, if $\sys{\mV}{p} \wtraS{\ftr} \sys{\yes}{p'}$  for some $p'$.
  \begin{itemize}
    \item A monitor \mV rejects (\resp accepts) \tr, using the abuse of notation \rej{\mV,\tr} (\resp \acc{\mV,\tr}), if $\exists p,\ftr,\trr$  such that
    $ {\tr = \ftr\trr}$ and  $\rej{\mV,p,\ftr}$ (\resp $\acc{\mV,p,\ftr}$).
    \item A monitor \mV rejects (\resp accepts) $p$, using the abuse of notation \rej{\mV,p} (\resp \acc{\mV,p}), if $\exists \ftr$ such that $\rej{\mV,p,\ftr}$ (\resp $\acc{\mV,p,\ftr}$).
  \end{itemize}
 We also say that \emph{\mV rejects \ftr} as a shorthand for $\exists p \cdot \rej{\mV,p,\ftr}$, and similarly, \emph{\mV accepts \ftr} is a shorthand for $\exists p \cdot \acc{\mV,p,\ftr}$. \qedd
\end{defn}

As
\Cref{def:acc-n-rej,lem:unzipping,lem:zipping}
%the above definition
%and \Cref{lem:unzipping,lem:zipping}
make clear, a monitor accepts or rejects a finite trace \ftV\ iff it can transition to the appropriate verdict by reading \ftV. This hints at the fact that each monitor might be ``equivalent to a deterministic one''. As we will see in \Cref{prop:determinization},
%Proposition 3.11,
this is indeed the case.

\subsection{Parallel Composition of Monitors}
\label{sec:parallel-mon}

\begin{figure}[t]
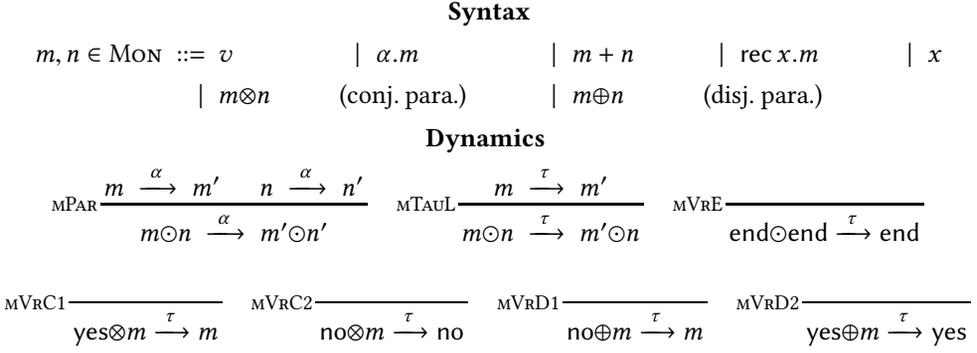

  % \centering
  \textbf{Syntax}
  \begin{align*}
    \mV,\mVV\in\Mon\ \bnfdef\
    &~~\vV
    && \bnfsepp \prf{\acta}{\mV}
    && \bnfsepp \ch{\mV}{\mVV}
    && \bnfsepp \rec{x}{\mV}
    &&\bnfsepp x \\
    % & \bnfsepp \yes && (\text{satisfaction}) && \bnfsepp \no && (\text{violation})\\
%    \vV,\vVV \in \Verd & \bnfdef\ \stp
%    &&  \bnfsepp \no
%    && \bnfsepp \yes
%    \ldots
%\\
    \bnfsepp &~~\mV \paralC \mVV  && \text{(conj. para.)}&& \bnfsepp \mV \paralD \mVV && \text{(disj. para.)}
  \end{align*}
  \textbf{Dynamics}
  \setpremisesend{0.3ex} %to make the rules less wide
   \begin{mathpar}
 \inference[\rtit{mPar}]{\mV\traSS{\acta}\mV' & \mVV\traSS{\acta}\mVV'}
 {\mV\paralG\mVV\traSS{\acta}\mV'\paralG\mVV'}
  \quad
 % \inference[\rtit{mBlk}]{\mV\traSS{\acta}\mV' & \mVV\traSSN{\acta} & \mVV\traSSN{\tau}}
 % {\mV\paralG\mVV\traSS{\acta}\mV'\paralG\stp}
% \\\\\\
 \inference[\rtit{mTauL}]{\mV\traSS{\tau}\mV'}
 {\mV\paralG\mVV\traSS{\tau}\mV'\paralG\mVV}
 \quad
 % \inference[\rtit{mTauR}]{\mVV\traSS{\tau}\mVV'}
 % {\mV\paralG\mVV\traSS{\tau}\mV\paralG\mVV'}
 % \\\\\\
  \inference[\rtit{mVrE}]{}{\stp\paralG\stp \traS{\tau} \stp}
 % \quad
  \\\\\\
 \inference[\rtit{mVrC1}]{}{{{\yes\paralC\mV} \traS{\tau} {\mV}}}
 \quad
 \inference[\rtit{mVrC2}]{}{{{\no\paralC\mV} \traS{\tau} {\no}}}
% \\\\\\
\quad
 \inference[\rtit{mVrD1}]{}{{{\no\paralD\mV} \traS{\tau} {\mV}}}
 % \\\\\\
 \quad
 \inference[\rtit{mVrD2}]{}{{{\yes\paralD\mV} \traS{\tau} {\yes}}}
\end{mathpar}
\setpremisesend{.75ex}%reset to default value
  \caption{Parallel Monitors. The syntax and dynamics of parallel monitors are extensions of the ones for regular monitors, as presented in \Cref{fig:monit-instr}. Parallel monitors use the same instrumentation as regular monitors.}
  \label{fig:monit-parallel}
\end{figure}

When relating monitors to formulae, it may be convenient \emph{not} to view monitors as one monolithic
%block
entity
but rather as a \emph{system of sub-monitors} where the constituent submonitors are concerned with checking specific subformulae.
For instance, the use of sub-monitors executing in parallel facilitates the synthesis of monitors from formulae in a \emph{compositional} fashion.
Monitors with parallel composition,
$\mV,\mVV \in \Mon$, are defined by the grammar and transition rules in \cref{fig:monit-parallel}.  In particular, we endow monitors with conjunctive parallelism, $\paralC$, and disjunctive parallelism, $\paralD$. We use the notation $\paralG$ to range over either $\paralC$ or $\paralD$ (\ie $\paralG \in \sset{\paralC,\paralD}$).

\Cref{fig:monit-parallel} also outlines the behaviour of parallel monitors.
Rule \rtit{mPar} states that \emph{both} submonitors need to be able to analyse an external action \acta for their parallel composition to transition with that action.
The rules in \Cref{fig:monit-parallel} also allow \actt-transitions for the reconfiguration of parallel compositions of monitors.
For instance, rules \rtit{mVrC1} and \rtit{mVrC2} describe the fact that, whereas $\yes$  verdicts are uninfluential in conjunctive parallel compositions, $\no$  verdicts supersede the verdicts of other monitors in a conjunctive parallel compositions (\Cref{fig:monit-parallel} omits  the obvious symmetric rules).
The dual applies for $\yes$  and $\no$  verdicts in a disjunctive parallel composition, as described by rules \rtit{mVrD1} and \rtit{mVrD2}.
Rule \rtit{mVrE} applies to both forms of parallel composition and consolidates multiple inconclusive verdicts.
Finally, rules \rtit{mTauL} and its dual \rtit{mTauR} (omitted) are contextual rules for these monitor reconfiguration steps.

We identify a useful monitor predicate that obviates the need for the rule \rtit{iTer} of \Cref{fig:monit-instr} that prematurely terminates monitors; see \Cref{ex:mon-acc-n-rej}. In the case of parallel monitors, it also allows us to neatly decompose the monitor behaviour in terms of the respective sub-monitors.

\begin{defn}[Monitor Reactivity]
\label{def:reactive-mon}
We call a monitor $\mV$ \emph{reactive} when for every $\mVV \in \reach{\mV}$ and $\act \in \Act$, there is some $\mVV'$ such that $\mVV \wtraS{\acta} \mVV'$. \qedd
\end{defn}

\Cref{ex:reactiveness} below indicates why the assumption that $\mV_1$ and $\mV_2$ are reactive is needed in \Cref{cor:parallel-and-monitors}, which states that parallel monitors behave as expected with respect to the acceptance and rejection of traces as long as the consitituent submonitors are reactive.

\begin{exmp} \label{ex:reactiveness}
 Assume that $\Act=\sset{a,b}$.
  The monitors $\ch{\prf{a}{\yes}}{\prf{b}{\no}}$ and $\rec{x}{(\ch{\prf{a}{x}}{\prf{b}{\yes}})}$
  % and
  % $\ch{\prf{a}{\yes}\paralC\prf{a}{\no}}{\prf{b}{\no}}$
  are both reactive.
  The monitor $\mV = \prf{a}{\yes}\paralC\prf{b}{\no}$, however, is \emph{not} reactive.
 Since the submonitor $\prf{a}{\yes}$ can only transition with $a$, according to the rules of \Cref{fig:monit-parallel}, $\mV$ cannot transition with any action that is not $a$.
 Similarly, as the submonitor $\prf{b}{\no}$ can only transition with $b$, $\mV$ cannot transition with any action that is not $b$.
 Thus, $\mV$ cannot transition to any monitor, and therefore it cannot reject or accept any trace.
% upon analysing $\actb$, the monitor $\act.\yes \paralC \actb.\no$ does not transition to $\stp \paralC \no$ because only the submonitor $\actb.\no$ can transition with $\actb$.
% It therefore does not reject traces beginning with $\actb$.
 By contrast, the  monitor $\mVV = (\ch{\prf{a}{\yes}}{\prf{b}{\stp}}) \paralC (\ch{\prf{b}{\yes}}{\prf{a}{\stp}})$ is reactive, because its constituent submonitors are reactive as well. \qedd
%  , and $\mVV$ can transition upon analysing $\actb$ to $\stp \paralC \no$ and then to $\no$.
%    Therefore, we see that for the parallel operators to behave as intended, the monitors need to be reactive. This observation is behind the conditions for the claims of
% %   is formalised in
%    \Cref{lem:mon-combinators}.
%     \qedd
\end{exmp}
\begin{lem}[Monitor Composition and Decomposition] \label{cor:parallel-and-monitors}
For reactive $\mV_1$ and $\mV_2$:
\begin{itemize}
\item
$\mV_1 \paralC \mV_2$ rejects $\tV$ if and only if either $\mV_1$ or $\mV_2$ rejects $\tV$.
% For reactive $\mV$ and $\mVV$,
\item
$\mV_1 \paralC \mV_2$ accepts $\tV$ if and only if both $\mV_1$ and $\mV_2$ accept $\tV$.
% \end{cor}

% \begin{cor} \label{cor:parallel-or-monitors}
% For reactive $\mV$ and $\mVV$,
\item
$\mV_1 \paralD \mV_2$ rejects $\tV$ if and only if both $\mV_1$ and $\mV_2$ reject $\tV$.
% For reactive $\mV$ and $\mVV$,
\item
$\mV_1 \paralD \mV_2$ accepts $\tV$ if and only if either $\mV_1$ or $\mV_2$ accepts $\tV$.
 \qed
\end{itemize}
\end{lem}

% If for a monitor $\mV$, there is a finite trace $\ftr$, such that $\mV \wtraS{\ftr} \yes$ and $\mV \wtraS{\ftr} \no$, then we call $\mV$ inconsistent.
% We are interested only in consistent monitors (\ie not inconsistent monitors), and we assume the ones we handle are such.
%
Parallel monitors are a convenient formalism for constructing monitors in a compositional fashion and facilitate the definition of monitor synthesis functions from a specification logic.
However, these monitors are
%not more
only
as
expressive
%than
as
regular monitors, as \Cref{prop:extended-mon-to-reg-mon} demonstrates. \Cref{subsec:transformations} is devoted to the proof of this result.
% Nonetheless, parallel monitors present a convenient framework that will help us monitor for certain properties in a more straightforward manner than directly synthesizing regular monitors would allow.

%
%
%
%\begin{prop} \label{prop:extended-mon-to-reg-mon}
%For all reactive and consistent parallel monitors \mV, there exists a regular monitor \mVV such that \mV and \mVV are verdict equivalent, \ie $\mV\mveq\mVV$.
% % in the sense that they accept, reject and terminate on exactly the same linear processes.
%\end{prop}
%% \begin{proof} Antonis
%% \end{proof}
%
%Subsection \ref{subsec:transformations} is devoted to the proof of Proposition \ref{prop:extended-mon-to-reg-mon}.
%% \ac{...and possibly other things as well}

%\begin{rem}
%	 \Cref{lem:mon-combinators} indirectly describes three different kinds of non-determinism for reactive parallel monitors.
%	Operator $\paralD$ can be thought of as an existential monitor choice, as
%	$\mV_1 \paralD \mV_2$ will accept (\resp reject) iff either (\resp both) of its components accepts (\resp reject). Dually, $\paralC$ can be thought of as a universal choice. The Operator $+$ is a different choice that favours neither acceptance nor rejection, but generates either verdict, as long as one of its component monitors can reach it. \qedd
%\end{rem}

\subsection{Monitor Transformations: Parallel to Regular}
\label{subsec:transformations}

We describe how one can transform a parallel monitor to a verdict-equivalent regular one.
For this, we use known results about alternating finite automata, restated here for completeness.

\begin{defn}[Alternating Automata]
  An alternating finite automaton is a quintuple $A = (Q, \Sigma, q_0, \delta, F)$, where $Q$ is a finite set of states,
  $\Sigma$ is a finite alphabet, $q_0$ is the starting state, $F \subseteq Q$ is the set of accepting/final states, and $\delta : (Q\times \Sigma) \to (2^Q \to \{0,1\})$ is the transition function.
An alternating finite automaton is non-deterministic (NFA) if for each $\act \in \Sigma$ and $q\in Q$, there is some $S_{q,\act} \subseteq Q$, such that for all $S \subseteq Q$, $\delta(q,\act)(S)=1$ if and only if $S \cap S_{q,a} \neq \emptyset$.
\qedd
\end{defn}

  Intuitively, given a state $q \in Q$ and a symbol $\act \in \Sigma$,
  $\delta$ returns a boolean function on $2^Q$ that evaluates,
  given a truth-assignment on the states of $Q$ (represented by a subset of $Q$),
  an assigned truth-value for $q$. We can extend the transition function
  to
  $\delta^*: (Q\times \Sigma^*) \to (2^Q \to \{0,1\})$, so that $\delta^*(q, \varepsilon)(R) = 1$ iff $q \in R$,
  and
  $\delta^*(q,\act w)(R) = \delta(q,\act)(\{ q' \in Q \mid  \delta^*(q',w)(R) = 1 \})$.
  We say that the automaton \emph{accepts} $w \in \Sigma^*$ when $\delta^*(q_0,w)(F) = 1$, and that
  it recognizes $L\subseteq \Sigma^*$ when $L$ is the set of strings accepted by the automaton.

\begin{defn}[Monitor Language Acceptance and Rejection]\label{def:mon-language}
A monitor $\mV$ accepts (\resp rejects) a set of finite traces (\ie a language) $L \subseteq \Act^*$ when for every $\ftr \in \Act^*$, $\ftr \in L$ if and only if $\mV$ accepts (\resp rejects) $\ftr$. We call the set that $\mV$ accepts (\resp rejects) $L_a(\mV)$ (\resp $L_r(\mV)$).
\qedd
\end{defn}

\begin{prop}\label{prop:monitor2automaton}
    For every reactive parallel monitor $\mV$, there is an alternating automaton that
    accepts $L_a(\mV)$
    % the set of finite traces that $\mV$ accepts
    and one that
    accepts $L_r(\mV)$.
    % the set of finite traces that $\mV$ rejects.
    % Furthermore,  has at most $|m|$ states.
  \end{prop}
  \begin{proof}
    We describe the
    % recursive
    process of constructing an alternating automaton that accepts $L_a(\mV)$
    % exactly the finite traces that $\mV$ accepts
    --- the case for $L_r(\mV)$ is similar.
    We assume that for every variable $x$ that appears in $\mV$, there is a unique submonitor of $\mV$ of the form $\rec x \mVV$, such that $x$ appears in $\mVV$.
    % automaton that accepts exactly the traces that $\mV$ rejects is constructed in a similar way.
    % that may have free variables.
    % For every variable $x$ of the monitor, there is a designated state $q(x)$ of the resulting automaton.
%
    The automaton for $\mV$ is $A_\mV = (Q,\Act,\mV,\delta,F)$, where
    \begin{itemize}
    \item $Q$ is the set of submonitors of $\mV$;
    % \item $q_0 = \mV$;
    \item $F = \{ \mVV \in Q \mid n \text{ accepts } \varepsilon \} $;
    \item Let for every $S \subseteq Q$, $\delta_0 (q,\act) (S) = 1$ iff $q \in F$;  $\delta$ is the closure of $\delta_0$ under the following conditions.
    % (we can view $\delta$ as the set of triples $(\mVV,\act,S)$ for which $\delta(\mVV,\act,S)=1$).
    For every $S \subseteq Q$:
      \begin{itemize}
        \item if $\mVV \in S$, then $\delta (\act.n,\act) (S) = 1$; % iff $n \in S$;
        \item if $\delta (\mVV,\act) (S) = 1$ or $\delta (\mVV',\act) (S) = 1$, then $\delta (\mVV + \mVV',\act) (S) = 1$;
        \item if $\delta (\mVV,\act) (S) = 1$ or $\delta (\mVV',\act) (S) = 1$,
              and
%              $\exists q.\mVV \wtraS{\act} q$
%              and
%              $\exists q.\mVV' \wtraS{\act} q$,
              $\mVV \wtraS{\act}$
              and
              $\mVV' \wtraS{\act}$,
              then $\delta (\mVV \paralD \mVV',\act) (S) = 1$;
        \item if $\delta (\mVV,\act) (S) = 1$ and $\delta (\mVV',\act) (S) = 1$, then $\delta (\mVV \paralC \mVV',\act) (S) = 1$;
        \item if $\delta (\mVV,\act) (S) = 1$ and $\rec x \mVV \in Q$, then $\delta (\rec x \mVV,\act) (S) = \delta (x,\act) (S) = 1$.
      \end{itemize}
    \end{itemize}
    % We observe that $\delta^*(\mVV,\ftr)(F)$ also satisfies these closure conditions --- by induction on $\ftr$:
    % if $\ftr = \varepsilon$, then $\delta^*(\mVV,\ftr)(F) = 1$ iff $\mVV \in F$
In \Cref{sec:transformations-appendix}, we present the remaining
%To complete the
proof, 
%we show 
that
$\mV$ accepts $\ftr$ if and only if
$\delta^*(\mV,\ftr) (F) = 1$.
  \end{proof}

\begin{rem}
%	\Cref{prop:monitor2automaton} demonstrates that parallel monitors, which are potentially infinite-state can be
%
	The assumption that the monitor is reactive is necessary for the construction in the proof of \Cref{prop:monitor2automaton} to be correct.
Consider, for example, the monitor $\mV_1 = \prf{a}{\prf{a}{\yes}} \paralD \prf{a}{\prf{b}{\yes}}$.
Although $\prf{a}{\prf{a}{\yes}} \traS{a} \prf{a}{\yes} \traS{a} \yes$, the monitor does not accept any trace since $\prf{b}{\yes} \centernot{\traS{a}}$.
By the construction, in the resulting alternating automaton, $F = \{\yes\}$,  and therefore $\delta(\yes,a)(F) = 1$, implying that $\delta(a.\yes,a)(F) = 1$, in turn implying that $\delta^*(\mV_1,aa)(F) = 1$, according to the closure conditions for $\delta$.
Therefore, $aa$ is a finite trace that the automaton accepts and the monitor does not.

In light of our assumption that monitor $\mV$ in \Cref{prop:monitor2automaton} is reactive, the third condition for $\delta$ in the construction in the proof of the proposition may seem superfluous.
However, reactivity does not transfer to submonitors. For example, let
$\mV_2 = \ch{(\prf{a}{\yes} \paralD \prf{b}{\yes})}{\ch{\prf{a}{\stp}}{\prf{b}{\stp}}}$.
Reasoning similarly to the above argument for $\mV_1$, $\mV_2$
is a reactive parallel monitor, which accepts no traces. On the other hand, a more naive construction that ensures that $\delta (\mVV \paralD \mVV',\act) (S) = 1$ whenever $\delta (\mVV,\act) (S) = 1$ or $\delta (\mVV',\act) (S) = 1$, would result in an automaton that accepts the finite trace $a$.
%The enhanced condition is required, because

As we see in the remainder of this section, \Cref{prop:monitor2automaton} implies that potentially \emph{infinite-state} parallel monitors
are equivalent to \emph{finite-state} regular monitors.
% We find that t
The subtleties that we
%have
pointed out are the trade-off for keeping the construction of the alternating automaton
%fairly
straightforward.
%We also note that
%%%%%%% I need to think about this with a clearer head, because m IS reactive...
\qedd
\end{rem}

  \begin{cor}\label{cor:parallel2NFA}
    For every reactive parallel monitor $\mV$, there is an NFA that
    accepts $L_a(\mV)$
    and an NFA that
    accepts $L_r(\mV)$, and each has at most $2^{l(\mV)}$ states.
  \end{cor}
  \begin{proof}
  % We observe that the
  The alternating automaton that is constructed in the proof of \Cref{prop:monitor2automaton} has at most as many states as there are submonitors in $\mV$ which, in turn, are not more than $l(\mV)$.
  Furthermore, it is a known result  that every alternating automaton with $k$ states can be converted into an NFA with at most $2^k$ states that accepts the same language~\cite{alternation,FJY90}.
  \end{proof}

We now have all the ingredients to complete the proof of \Cref{prop:extended-mon-to-reg-mon}.  This relies on a notion of monitor equivalence from \citeN{AceAFI:2017:CIAA} that focusses on how monitors can reach verdicts.
% the verdict capacity.

\begin{defn}[Verdict Equivalence] \label{def:verdict-eq}
	Monitors \mV and \mVV are \emph{acceptance equivalent} (\resp rejection equivalent), denoted as $\mV\maeq\mVV$ (\resp $\mV\mreq\mVV$), if for every finite trace \ftr,
%and $\vV \in \{\yes,\no \}$,
$\mV \wtra{\ftr} \yes$  iff  $\mVV \wtra{\ftr} \yes$ (\resp $\mV \wtra{\ftr} \no$  iff  $\mVV \wtra{\ftr} \no$).
They are
\emph{verdict equivalent}, denoted as $\mV\mveq\mVV$, if
they are both acceptance- and rejection-equivalent.
%for every finite trace \ftr\
%and $\vV \in \{\yes,\no \}$,
%$\mV \wtra{\ftr} \vV$  iff  $\mVV \wtra{\ftr} \vV$
 \qedd
\end{defn}

%, which we restate below for ease of reference.

\begin{prop}\label{prop:extended-mon-to-reg-mon}
For all reactive
% and consistent
parallel monitors \mV, there exist regular monitors $\mVV_1,\mVV_2$, and \mVV such that
$\mVV_1$ and $\mVV_2$ are single-verdict monitors that are respectively acceptance-equivalent and rejection-equivalent to $\mV$, and
\mV and \mVV are verdict equivalent,
%\ie $\mV\mveq\mVV$,
and $l(\mVV_1), l(\mVV_2), l(\mVV) = 2^{O\left(l(\mV)\cdot 2^{l(\mV)}\right)}$.
 % in the sense that they accept, reject and terminate on exactly the same linear processes.
\end{prop}
\begin{proof}
Let $A_\mV^a$ be an NFA for $L_a(\mV)$ with at most $2^{l(\mV)}$ states, and let $A_\mV^r$ be an NFA for $L_r(\mV)$ with at most $2^{l(\mV)}$ states, which exist by \Cref{cor:parallel2NFA}.
From these NFAs, we can construct regular monitors $\mV_R^a$ and $\mV_R^r$,
%of length  $2^{O\left(l(\mV)\cdot 2^{l(\mV)}\right)}$,
such that
$\mV_R^a$ accepts $L_a(\mV)$ and $\mV_R^r$ rejects $L_r(\mV)$, and $l(\mV_R^a), l(\mV_R^r) = 2^{O\left(l(\mV)\cdot 2^{l(\mV)}\right)}$ \cite{determinization}.
Therefore,
$\mV_R^a \maeq \mV$ and $\mV_R^r \mreq \mV$, and
$\mV_R^a + \mV_R^r$ is regular and verdict-equivalent to $\mV$, and $l(\mV_R^a + \mV_R^r) =  2^{O\left(l(\mV)\cdot 2^{l(\mV)}\right)}$.
% , completing
% the proof of Proposition \ref{prop:extended-mon-to-reg-mon}.
\end{proof}

The techniques of \citeMac{determinization} can also be used to produce deterministic monitors.

\begin{defn}[\cite{determinization}]\label{def:determinism}
  A regular monitor $\mV$ is \emph{syntactically deterministic} iff
  every sum of at least two summands
    which appears in $\mV$
  is of the form
  $\sum_{\act \in A} \act.m_\act$,
  where $A \subseteq \Act$.
  \qedd
\end{defn}

\begin{exmp}
The monitor \ch{\prf{a}{\prf{b}{\yes}}}{\prf{a}{\prf{a}}{\no}} is not syntactically deterministic while the verdict-equivalent monitor \prf{a}{(\ch{\prf{b}{\yes}}{\prf{a}{\no}})} is syntactically deterministic.
\qedd
\end{exmp}

One can also consider non-syntactic notions of determinism, such as if $\mV \wtraS{\ftr} \mVV $ and $\mV \wtraS{\ftr} \mVV'$, then $\mVV \mveq \mVV'$. \Cref{lem:det_mon_is_det} shows that syntactic determinism implies this semantic notion. Henceforth we will simply say deterministic to mean syntactically deterministic.

\begin{lem}[\cite{determinization}]\label{lem:det_mon_is_det}
  If $\mV$ is deterministic, $\mV \wtraS{\ftr} \mVV $, and $\mV \wtraS{\ftr} \mVV'$, then $\mVV \mveq \mVV'$.
   \qed
\end{lem}

\begin{thm}[\cite{determinization}]\label{thm:determ_from_determ}
	For every consistent regular monitor $\mV$, there is a verdict-equivalent deterministic regular monitor $\mVV$ such that $l(\mVV) = 2^{2^{O(l(\mV))}}$.
	\qed
\end{thm}

\begin{prop}\label{prop:determinization}
For every consistent reactive parallel monitor $\mV$, there is a verdict-equivalent deterministic regular monitor $\mVV$ such that $l(\mVV) = 2^{2^{2^{O\left(l(\mV)\cdot 2^{l(\mV)}\right)}}}$.
\end{prop}

\begin{proof}
Using \Cref{prop:extended-mon-to-reg-mon},
$\mV$
%a consistent reactive parallel monitor
can be translated into  a (possibly nondeterministic) verdict-equivalent (hence consistent) regular monitor $\mVV_r$, such that $l(\mVV_r) = 2^{O\left(l(\mV)\cdot 2^{l(\mV)}\right)}$.
%A result from \cite{determinization} is
\Cref{thm:determ_from_determ} can
then be used to convert $\mVV_r$
%this consistent regular monitor
into a verdict-equivalent deterministic regular monitor $\mVV$, such that $l(\mVV) = 2^{2^{O(l(\mVV_r))}}$.
Therefore, $l(\mVV) = 2^{2^{2^{O\left(l(\mV)\cdot 2^{l(\mV)}\right)}}}$.
\end{proof}

%% file: monitorability.tex
% !TEX root = main.tex

\emph{Monitorability} is the study of the relationship between the semantics of a logic on the one hand (\ie satisfactions and violations), and the verdicts that can be discerned by the monitoring setup on the other (\ie acceptances and rejections).
The concept relies on what
%constitutes to be
a \emph{correct} monitor for a particular formula
is,
which, in turn, defines what it means for a formula to be \emph{monitorable}.
In this section we focus on the monitorability of \ltmu.
Based on the definition of trace acceptance and rejection of \Cref{def:acc-n-rej}, we adapt the concepts of monitor soundness and completeness (\wrt a formula) from \citeMac{FraAI:17:FMSD} to the linear-time setting.
% of \Cref{fig:recHML}.

\begin{defn}[Linear-time Monitor Soundness and Completeness] \label{def:soundness-n-completeness}
\quad
\begin{itemize}
  \item A monitor \mV is \emph{sound} for a (closed) formula \hV of \ltmu over
  % (infinite)
  traces if, \emph{for all} $\tV\in \Trc$:
  \begin{itemize}
  \item  \rej{\mV,\tV} implies $\tV\not \in \hSemL{\hV}$;
  \item  \acc{\mV,\tV} implies $\tV\in \hSemL{\hV}$.
  \end{itemize}
  \item A monitor \mV is \emph{violation-complete} for a (closed) formula \hV of \ltmu over traces if \emph{for all} $\tV\in \Trc$, $\tV\not\in\hSemL{\hV}$ implies  \rej{\mV,\tV}.
  It is \emph{satisfaction-complete} if $\tV\in \hSemL{\hV}$ implies \acc{\mV,\tV}.
  \item A monitor \mV is \emph{complete} for a (closed) formula \hV of \ltmu if it is \emph{both} violation- and satisfaction-complete for it. \qedd
\end{itemize}
\end{defn}

The definition of soundness and completeness for monitors
% , as in \Cref{def:soundness-n-completeness},
depends on the semantics given to the formulae.
Since we focus on linear-time semantics
% and therefore on monitorability over traces,
in this section,
instead of saying that a monitor is sound or violation- or satisfaction-complete, or complete for a formula over traces, we respectively simply say that it is sound or violation- or satisfaction-complete, or complete for the formula.
In \Cref{sec:branchingtime}, we will introduce variations of \Cref{def:soundness-n-completeness} that depend on different semantics for \ltmu.
Observe that a monitor that is \emph{sound} for some formula must be \emph{consistent}. % and a monitor that is \emph{complete} for a formula must be \emph{total}.

%%% Trivial: proofs below are boring.
%\begin{lem}
%A monitor that is sound for some formula is consistent
%\end{lem}
%\begin{proof}
%Is a monitor $\mV$ is sound for a formula $\hV$ then there is not trace $\tV$ such that $\rej{\mV,\tV}$ and \acc{\mV,\tV} hence $\mV$ is consistent.
%\end{proof}

%begin{lem}
%A% monitor that is complete for some formula is total.
%\end{lem}
%\begin{proof}
%Is a monitor $\mV$ is complete for a formula $\hV$ then there is not trace $\tV$ such that neither $\rej{\mV,\tV}$ nor \acc{\mV,\tV} hence $\mV$ is total.
%\end{proof}

Following \citeMac{FraAI:17:FMSD}, we assume that the minimum requirement for a monitor to correctly correlate to a formula is for it to be \emph{sound}.
% \wrt it.
%
It can be however argued that, depending on the circumstance of the application requirements, different notions of completeness may be deemed adequate enough.
%
%As i
It turns out
%, it is generally the case
that \emph{not} all formulae can be monitored adequately at runtime.
Moreover,
%and
the more stringent the requirement for adequate monitoring, the more are the formulae that \emph{cannot} be monitored.
In the remainder of the section, we consider different definitions for adequate monitoring and establish \ltmu monitorability results in each case.

In \Cref{sec:complete-monitorability}, we
present
%establish
monitorability results with respect to complete monitors.
In \Cref{sec:strictly-complete}, we introduce the additional requirement of \emph{tightness} for a monitor, under which the monitor reaches a verdict as soon as it has read sufficient information from the input trace and not later. We explain what one needs to do to construct a tight monitor.
In \Cref{sec:partially-complete}, we establish monitorability results for \emph{partially complete} monitors, which are satisfaction-complete or violation-complete for their respective formulae, but are not required to be both.
This relaxation
%of the requirements
allows us to monitor for more formulae.
Finally, in \Cref{sec:strictly-complete-w-recursion}, we
examine what one must do to construct tight partially complete monitors, and we explain why the methods of \Cref{sec:strictly-complete}
are not likely to apply for this case.
%do not apply for this case as well.

\subsection{Complete Monitorability}
\label{sec:complete-monitorability}
\input{monitorability-pt1}

\subsection{Tightly-Complete Monitors}
\label{sec:strictly-complete}
\input{monitorability-pt2}

\subsection{Partially-Complete Monitors}
\label{sec:partially-complete}

 \input{monitorability-pt3}

\subsection{Tightly-Complete Monitors for Recursion}
\label{sec:strictly-complete-w-recursion}
\input{monitorability-pt4}

%% file: monitorability-pt1.tex
% !TEX root = main.tex

We first consider \emph{(sound and) complete} monitors as our notion of adequate monitoring for a particular formula.
This induces the following definition of monitorable formula and (sub)logic.

\begin{defn}[Complete Monitorability]
  \label{def:complete-monitorability}
%
% We say that
 A
% a
 formula $\hV\in\UHML$ is
% said to be
 \emph{complete-monitorable} over traces iff there \emph{exists} a monitor \mV that is  sound and complete for it.
 A (sub)logic $\LSet\subseteq \UHML$ is  \emph{complete-monitorable} over traces iff
% \emph{for all} formulae $\hV\in\LSet$, \hV is complete-monitorable.
 each formula $\hV\in\LSet$ is complete-monitorable.
 \qedd
\end{defn}

\begin{rem}
  In this section we only use \Cref{def:complete-monitorability} for the linear-time interpretation of \UHML.
  However, its general form allows it to be used for other interpretations of the logic, with the appropriate adaptation of complete monitors (\eg along the lines of \citeMac{FraAI:17:FMSD}). \qedd
\end{rem}

%We are already able to establish results linking prior work on monitors \cite{AceAFI:2017:CIAA} with complete monitorability, such as that of monitor verdict equivalence, \Cref{def:verdict-eq}, and others.

As the following results highlight, soundness and completeness for monitors are invariant
under verdict equivalence.

\begin{prop}\label{prop:vedict-equiv-implies-complete-formula}\label{prop:complete-semantic-equivalence}
  If \mV is sound and complete for \hV then
  \begin{enumerate}
    \item $\mV \mveq \mVV$ implies \mVV is sound and complete for \hV;
    \item \mV is a sound and complete monitor for $\hV'$ implies $\hSemL{\hV}=\hSemL{\hV'}$. \qed
  \end{enumerate}
\end{prop}

In line with other works on monitorability \cite{MannaPnueli:91:TCS,ChangMannaPnueli:92:ALP,PnueliZaks:06:FM,BauerLeuckerSchallhart:10:LandC,FalconeFernandezMounier:STTT:12,CiniFrancalanza:15:TACAS,FraAI:17:FMSD},  \emph{not} all properties in \ltmu are complete monitorable.

\begin{exmp} \label{ex:not-complete-monitorable}
  % An example to show that certain formulae are not complete monitorable.
  %
  The formula $\hV_1=\hSuf{a}{\hTru}\,\textsf{U}\,\hSuf{b}{\hTru}$ is \emph{not} complete-monitorable.
  For if, by contradiction, we assume that it was then there must \emph{exist} some sound and complete monitor \mV for $\hV_1$.
  Since the trace $a^\omega \not\in \hSemL{\hV_1}$, this monitor \mV rejects $a^\omega$ which, by \Cref{def:acc-n-rej}, means that it must reach a violation after observing a \emph{finite} prefix  $a^k$ (for $k\geq 0$).
  But this would also mean that \mV rejects \emph{all} traces of the form $a^k b \tV$, which clearly satisfy $\hV_1$, thereby contradicting the assumption that \mV is sound.
  Similarly, it can be argued that the formula $\hV_2=\hSuf{a}{\hSuf{b}{\hSuf{b}{\hTru}}}\,\textsf{R}\,\hSuf{a}{\hTru}$
  is \emph{not} complete-monitorable either.
  For if it was, a sound and complete monitor $\mV_2$ would accept the trace $a^\omega$ after analysing some
  prefix
%  specified lenght
  $a^n$
  of it;
  this would also mean that this monitor would also accept any trace of the form $a^{n}ba\tV$, which clearly violates the property.  Thus, no such monitor exists.
  \qedd
\end{exmp}

\Cref{ex:not-complete-monitorable} raises the question as to which \ltmu properties can be monitored according to \Cref{def:complete-monitorability}.
To answer this question,
we first identify a  fragment of \ltmu that is guaranteed to be complete-monitorable
and then show its maximality.

\begin{defn}[The complete-monitorable fragment of \ltmu]
  \label{def:complete-fragment-HML}
   The recursion-free syntactic fragment of \ltmu (a syntactic variant of HML \cite{HennessyM:1985:JACM}) is defined as:
  % by the following grammar.
\begin{align*}
    \hV,\hVV \in \HML &\bnfdef  \hTru &
           &\bnfsepp  \hFls &
           & \bnfsepp \hOr{\hV\,}{\,\hVV}  &
           & \bnfsepp \hAnd{\hV\,}{\,\hVV} &
           &\bnfsepp \hSuf{\ASet}{\hV} &
            &\bnfsepp \hNec{\ASet}{\hV}
            . \tag*{\qedd}
  \end{align*}
\end{defn}

For every formula $\hV \in \HML$, we can define a monitor synthesis function as follows.

\begin{defn}[Complete Monitor Synthesis] \label{def:mon-synt-complete}
The function $\hSyn{-}: \HML \to \Mon$ is defined inductively as follows:
\begin{align*}
  \hSyn{\hFls} & \defeq \no &
  \hSyn{\hAndF} &\defeq \hSyn{\hV_1} \paralC \hSyn{\hV_2} &
  \hSyn{\hNec{\ASet}{\hV}} &\defeq \ch{\prf{\ASet}{\hSyn{\hV}}}{\uprf{\ASet}{\yes}} \\
  \hSyn{\hTru} & \defeq \yes &
  \hSyn{\hOrF} &\defeq \hSyn{\hV_1} \paralD \hSyn{\hV_2} &
  \hSyn{\hSuf{\ASet} \hV} &\defeq \ch{\prf{\ASet}{\hSyn{\hV}}}{\uprf{\ASet}{\no}}
  .
  \tag*{\qedd}
\end{align*}
\end{defn}

\begin{lem} \label{lem:synt-complete-mon-reactive}
  For all $\hV \in \HML$, \hSyn{\hV} is reactive. \qed
\end{lem}

\begin{exmp}\label{ex:mon-synt-complete}
Assuming $\Act = \sset{a,b,c}$, the synthesised monitor for $\hV=\hAnd{\hNec{a}{\hSuf{b}{\hTru}}}{\hSuf{a}{\hNec{c}{\hFls}}}$, where $\hSemL{\hV} = \setof{ab\tV}{\tV\in\Act^\omega}$, is
\begin{align*}
  \hSyn{\hV} = \mV&= \bigl(\ch{\prf{a}{(\ch{\prf{b}{\yes}}{\prf{\sset{a,c}}{\no}})}}{\prf{\sset{b,c}}{\yes}}\bigr) \paralC \bigl(\ch{\prf{a}{(\ch{\prf{c}{\no}}{\prf{\sset{a,b}}{\yes}})}}{\prf{\sset{b,c}}{\no}}\bigr)
  .
\end{align*}
When we compose $\mV$ with $p=\rec{x}{\prf{a}{\prf{b}{x}}}$, we observe the following monitored behaviour:
\begin{align*}
  \!\!\sys{\mV}{p} \etraS{\actt a} \sys{\bigl((\ch{\prf{b}{\yes}}{\prf{\sset{a,c}}{\no}})\paralC (\ch{\prf{c}{\no}}{\prf{\sset{a,b}}{\yes}})\bigr)}{\prf{b}{p}} \traS{b} \sys{\yes\paralC\yes}{p} \traS{\actt} \sys{\yes}{p}
  .\tag*{{\qedd}}
\end{align*}
\end{exmp}

We show that, for each formula $\hV \in \HML$, the monitor \hSyn{\hV} is the witness sound and complete monitor for it. This, in turn, shows that \HML is complete-monitorable,
%as
in the sense of
%was formalised in
\Cref{def:complete-monitorability}.
%
%For  this to hold, though, we need to

\begin{prop} \label{prop:hml-monitorable}
For all $\hV \in \HML$, \hSyn{\hV} is a sound and complete monitor for \hV.
%\qed
\end{prop}
\begin{proof}
  From \Cref{def:soundness-n-completeness}, \emph{soundness} requires us to show that
  $(i)$ \rej{\hSyn{\hV},\tV}   implies $\tV\not \in \hSem{\hV}$ and
  $(ii)$ \acc{\hSyn{\hV},\tV} implies $\tV\in \hSem{\hV}$.
  \emph{Completeness}, requires us to show
  $(i)$  $\tV\not \in \hSem{\hV}$  implies \rej{\hSyn{\hV},\tV} and
  $(ii)$  $\tV\in \hSem{\hV}$  implies \acc{\hSyn{\hV},\tV}.
  See \Cref{sec:complete-mon-app}.
\end{proof}

\begin{cor} \label{cor:hml-monitorable} \HML is complete monitorable. \qed
\end{cor}

% \medskip

\newcommand{\noVar}[1]{\ensuremath{\textsf{noV}(#1)}}
\newcommand{\noRec}[1]{\ensuremath{\textsf{noR}(#1)}}
\newcommand{\degree}[1]{\ensuremath{\textsf{deg}(#1)}}
\newcommand{\deter}[1]{\ensuremath{\textsf{det}(#1)}}

Following \citeMac{FraAI:17:FMSD}, we go one step further and show that the fragment \HML of \Cref{def:complete-fragment-HML} is \emph{maximally expressive} with respect to sound and complete monitors.
By this we mean that
%for
\emph{every} formula $\hV\in\UHML$ that
is complete-monitorable, in the sense of \Cref{def:complete-monitorability},
is semantically equivalent to
%we can express the same property denoted by it using a (semantically equivalent)
a formula from \HML.
%
% Stated otherwise,
Thus, we can limit ourselves to the syntactic fragment \HML without sacrificing any expressiveness in terms of complete-monitorable properties.

We show this claim in two steps.
First, we tighten expressiveness results from \Cref{sec:monitors} for the specific case of complete monitoring.
Concretely, we argue that every complete-monitorable formula (\Cref{def:complete-monitorability}) can  be monitored adequately by a \emph{recursion-free} syntactically deterministic monitor (see \Cref{def:determinism}).
This is shown via \Cref{lem:remove-rec-complete-monitor}, which relies on \Cref{def:no-rec-mon}.
In the second step, we devise an inverse synthesis function
%is devised
to obtain complete-monitorable \HML\ formulae from \emph{recursion-free}  deterministic monitors, \Cref{lem:consistent-implies-sound-and-complete}.
This formula synthesis function is then used for \Cref{prop:hml-maximal}, the last main result of \Cref{sec:complete-monitorability}.

\begin{defn}[Removing Monitor Recursion]
 \label{def:no-rec-mon}
 For each monitor \mV, we define \noRec{\mV} thus:
 \begin{align*}
   \noRec{x} & \defeq \stp &
   \noRec{\vV} &  \defeq \vV &
   \noRec{\rec{x}{\mVV}} &\defeq  \noRec{\mVV}
   \\
   \noRec{\ch{\mVV_1}{\mVV_2}} &\defeq \ch{\noRec{\mVV_1}}{\noRec{\mVV_2}} &
   \noRec{\prf{\acta}{\mVV}} & \defeq \prf{\acta}{\noRec{\mVV}}
   .\tag*{\qedd}
 \end{align*}
\end{defn}

\begin{lem}
  \label{lem:remove-rec-complete-monitor}
  If \mV is a syntactically deterministic monitor that is sound and complete for \hV, then $\noRec{\mV}$ is also a sound and complete monitor for \hV.
%  \qed
\end{lem}
\begin{proof}
  Using \Cref{prop:vedict-equiv-implies-complete-formula}, the result follows if we show that $\mV\mveq\noRec{\mV}$. See \Cref{sec:complete-mon-app}.
\end{proof}

The next step towards proving \Cref{prop:hml-maximal} is that of synthesising formulae from any
%finite (\ie
recursion-free
%)
syntactically deterministic monitor, which can be described by the following grammar.

\begin{defn}[%
%	Finite
Recursion-free
	Deterministic Monitors] \label{def:finite-determ-monitors}\
  % \begin{align*}
  %     \mV,\mVV\in\FMon\ &\bnfdef\  \no && \bnfsepp \yes && \bnfsepp \textstyle \sum_{\acta \in \ASet} \prf{\acta}{\mV_\acta}
  %     % \tag*{\qedd}
  % \end{align*}
  \begin{align*}
      \mV,\mVV\in\FMon\ \bnfdef\  \no \; \bnfsepp \yes \; \bnfsepp \textstyle \sum_{\acta \in \ASet} \prf{\acta}{\mV_\acta}
      . \tag*{\qedd}
  \end{align*}
%  \begin{math}
%  \mV,\mVV\in\FMon\ \bnfdef\  \no \; \bnfsepp \yes \; \bnfsepp \textstyle \sum_{\acta \in \ASet} \prf{\acta}{\mV_\acta}
%  % \tag*{\qedd}
%  .
%  \end{math} \qedd
\end{defn}
%
% We show that the synthesised formulae, which are all instances of $\HML$ from \Cref{def:complete-fragment-HML}, can be monitored for soundly and completely by any \emph{reactive} $\mV\in\FMon$.

We now show how to convert any recursion-free monitor \mV into an \HML formula $\mSyn{\mV}$.
We then argue that a \emph{reactive} \mV monitors soundly and completely for $\mSyn{\mV}$.

\begin{defn}
  % [Formula synthesis]
  \label{def:formula-synt-complete} The synthesis function $\mSyn{-}: \FMon \to \HML$ is defined as follows:
  \begin{align*}
    \mSyn{\yes} & = \hTru
    &
    \mSyn{\no} & = \hFls
    &
    {\textstyle \mSyn{\sum_{\acta\in\ASet} \prf{\acta}{\mV_\acta}} }
    & = {\textstyle \bigwedge_{\acta\in \ASet} \hNec{\acta}{\mSyn{\mV_\acta}}}
    \tag*{\qedd}
  \end{align*}
% \begin{align*}
%   \mSyn{\mV} &\defeq
%   \begin{cases}
%     \hTru & \text{ if }\mV=\yes \\
%     \hFls & \text{ if }\mV=\no\\
%     \bigwedge_{\acta\in \ASet} \hNec{\acta}{\mSyn{\mV_\acta}}
%     & \text{ if }\mV=\sum_{\acta\in\ASet} \prf{\acta}{\mV_\acta}
%     .
%   \end{cases}
%   \tag*{\begin{tabular}{r}
%   	\\
% %  	\hline
%   	\\
% %  	\hline
% %  	\\
% \hfill \qedd\\ \end{tabular}}
% \end{align*}
\end{defn}

\begin{lem}\label{lem:consistent-implies-sound-and-complete}
  Every \emph{reactive} monitor  $\mV\in\FMon$  is a sound and complete monitor for \mSyn{\mV}. \qed
\end{lem}

% \begin{exmp}
%   \af{add an example}  \ch{\prf{a}{\yes}}{\prf{b}{\no}}. \qedd
% \end{exmp}

We are now in a position to prove the
expressive
maximality of \HML from \Cref{def:complete-fragment-HML}.

\begin{prop}[Maximality for \HML] \label{prop:hml-maximal}
  For each $\hV\in\UHML$, if \hV is complete-monitorable, then there exists some $\hVV\in\HML$ such that $\hSemL{\hV}=\hSemL{\hVV}$.
\end{prop}
% \af{Check about reactivity}
\begin{proof}
  % \af{ToDo}
  From the results in \Cref{sec:monitors} and \Cref{lem:remove-rec-complete-monitor}, each complete-monitorable $\hV\in\UHML$ has a recursion-free deterministic monitor \mV that is sound and complete for it.
  By
   \Cref{lem:consistent-implies-sound-and-complete}, $\mV$ is sound and complete for $\mSyn{\mV}$ as well which is in \HML.
  \Cref{prop:complete-semantic-equivalence} thus yields $\hSemL{\hV}=\hSemL{\mSyn{\mV}}$ as required.
  See \Cref{sec:complete-mon-app} for more details.
\end{proof}

  The proof of \Cref{prop:hml-maximal} is constructive.  We are also able to prove (albeit in a non-constructive manner) an even stronger result (\Cref{thm:stronger-HML-maximality}) with respect to complete monitoring for any \emph{arbitrary} logic defined over traces. This increases the importance of the  fragment identified in \Cref{def:complete-fragment-HML} for the linear-time interpretation.
   % interpretation of \Cref{fig:recHML}.
The proof of \Cref{thm:stronger-HML-maximality} can be found in \Cref{sec:complete-mon-app}.

\begin{thm}
  \label{thm:stronger-HML-maximality}
  Let $\mV$ be a monitor from a monitoring system with the following two properties:
  \begin{enumerate}
  	\item
  	%	(1)
  	verdicts are irrevocable, that is, if $\mV$ accepts (respectively, rejects) a finite trace $\ftV$, then it accepts (respectively, rejects) all its extensions, and
  	%	(2)
  	\item
  	$\mV$ accepts (respectively, rejects) a trace $\tV$ if, and only if, it accepts (respectively, rejects) some finite prefix $\ftV$ of $\tV$.
  \end{enumerate}
  For any property \hV with a trace interpretation (not necessarily syntactically represented using \UHML), if \mV is sound and complete for \hV then \hV can be expressed via the syntactic fragment \HML of \Cref{def:complete-fragment-HML}. \qed
\end{thm}

%% file: monitorability-pt2.tex
% !TEX root = main.tex

% \af{Antonis?} \ac{Antonis, indeed.}
The sound and complete monitoring studied in \Cref{sec:complete-monitorability} does not specify \emph{when} a monitor should reach a verdict while it analyses a trace, as illustrated by the following example.

\begin{exmp}\label{exm:non-tightness}
Assume $\Act=\sset{a,b}$ and consider the formula $\hV=\hSuf{a}{\hSuf{a}{\hFls}}$, which is equivalent to $\hFls$.
Following \Cref{def:mon-synt-complete}, the synthesised monitor for $\hV$ is $\mV = \ch{\prf{a}{(\ch{\prf{a}{\no}}{\prf{b}{\no}})}}{\prf{b}{\no}}$.
%
% It expects an action;
After at most two consecutive actions, \mV will definitely reject, and therefore it correctly rejects all traces.
However, a more \emph{``efficient''} correct monitor for $\hV$ is \no, which rejects immediately. \qedd
\end{exmp}

A finite trace for which every extension violates (\resp satisfies) a property $\hV$ is often called a \emph{bad prefix} (\resp a \emph{good prefix}) for $\hV$~\cite{AlpernS85,PnueliZaks:06:FM,BauerLeuckerSchallhart:10:LandC}; good/bad prefixes provide sufficient finite information for acceptance/rejection.
\begin{exmp} \label{exm:non-tightness-2}
  $\hNec{\ASet}\hTru$ is equivalent to $\hTru$, and thus $\varepsilon$ is a good prefix for it.
  But $\hSyn{\hNec{\ASet}\hTru}$ from \Cref{def:mon-synt-complete} would first need to observe one action before accepting.
  Similarly,
  % $\bigwedge_{\act \in \Act} \hNec{\act} \hFls$
  \hNec{\Act}{\hFls}
  is equivalent to $\hFls$ and  $\varepsilon$ is a valid bad prefix. Yet the synthesised monitor only rejects after observing one action. \qedd
\end{exmp}

Although  the monitors synthesised in \Cref{sec:complete-monitorability} are
% sound and
complete, there may be a delay from the moment a good/bad prefix is seen to the point when a verdict is reached.
This observation does not affect monitor completeness:  the assurance that the stream of events is \emph{infinite} guarantees that any delay in reporting a verdict will not affect the formula's monitorability.
However,
% there are cases where
it may be important for a monitor to report a verdict as soon as it gathers sufficient information to do so.

% \ac{a bit roundabout, perhaps; we could also define tightness independently of the formulae; let me try.}

% \begin{defn}
%  We say that a monitor $\mV$ is \emph{tight-complete} for a formula $\hV$ when for every $\ftr$, if
%  $\forall\tV.
%   \ftr \tr \notin \hmeaning{\hV}_L$, then $\mV \wtraS{\ftr} \no$,
%   and
%   if
%  $\forall\tV.
%   \ftr \tr \in \hmeaning{\hV}_L$, then $\mV \wtraS{\ftr} \yes$.
%  % $\mV$ rejects (\resp accepts) all extensions $\ftr \tr$, then
%  % $\mV \wtraS{\ftr} \no$ (\resp $\mV \wtraS{\ftr} \yes$).
%  % , then there is some $\ftrr$, such that $\mV \traS{\ftr\ftrr} \vV' \neq \vV$.
% \end{defn}

%\ac{we cannot really call this a bad prefix yet, because there is no formula:}

\begin{defn}\label{def:tight}
  A monitor $\mV$ is \emph{tight} when, for every $\ftr \in \Act^*$, if
 $\mV$ rejects (\resp accepts) $\ftr\tV$ for every $\tV \in \Trc$,
  % \ftr \tr \notin \hmeaning{\hV}_L$,
  then $\mV \wtraS{\ftr} \no$ (\resp $\mV \wtraS{\ftr} \yes$).
 %  and
 %  if
 % $\forall\tV.
 %  \ftr \tr \in \hmeaning{\hV}_L$, then $\mV \wtraS{\ftr} \yes$.
 % $\mV$ rejects (\resp accepts) all extensions $\ftr \tr$, then
 % $\mV \wtraS{\ftr} \no$ (\resp $\mV \wtraS{\ftr} \yes$).
 % , then there is some $\ftrr$, such that $\mV \traS{\ftr\ftrr} \vV' \neq \vV$.
 \qedd
\end{defn}

% A tight monitor for $\hV$ reports a bad and a good prefix as soon as it is observed.
%
Although, as \Cref{exm:non-tightness} demonstrates, \Cref{def:mon-synt-complete}
% our monitor synthesis for \HML\
 does not always yield tight monitors we can identify a fragment of \HML\ for which it does.

\begin{defn}\label{def:slim}
 A \emph{slim} formula is defined by the following grammar:
  \begin{align*}
      \hV  ~~~~\bnfdef  ~~~~\hTru  ~~~~\bnfsepp  ~~~~\hFls
             ~~~~ \bnfsepp ~~~~{\textstyle \bigwedge_{\act \in B} \hNec{\act} \hV_\act ~~~~ \bnfsepp ~~~~\bigvee_{\act \in D} \hSuf{\act} \hV_\act} ,
   \end{align*}
    where $B,D \neq \emptyset$,
    $\forall \act \in B. \hV_\act \neq \hTru$,
    $\forall \act \in D. \hV_\act \neq \hFls$,
    either $B \neq \Act$ or $\exists \act \in B. \hV_\act \neq \hFls$, and
    either $D \neq \Act$ or $\exists \act \in D. \hV_\act \neq \hTru$.
    \qedd
\end{defn}

% As we proceed to demonstrate, any monitor synthesised from a slim formula will be tight.
% It is not hard to see that
All slim formulae are \HML formulae.
However, the conditions imposed on their syntax exclude redundancies that yield non-tight monitors.
We proceed to show that if $\hV$ is slim, then $\hSyn{\hV}$ is tight. To this end, we prove a lemma  showing the absence of redundancy in slim formulae.

\newcommand{\rewriteL}{\ensuremath{\Rrightarrow_L}}
\begin{figure}[!t]
\begingroup
\allowdisplaybreaks
\begin{align}
%   redundant, because they are subcases of the big ones:
    % [\alpha] \hV \land [\alpha] \hVV &\equiv_L [\alpha] (\hV \land \hVV) %\\
    % &[\alpha] \hV \lor [\alpha] \hVV &\equiv_L [\alpha] (\hV \lor \hVV) \\
    % \hSuf{\alpha} \hV \land \hSuf{\alpha} \hVV &\equiv_L \hSuf{\alpha} (\hV \land \hVV) %\\
    % &\hSuf{\alpha} \hV \lor \hSuf{\alpha} \hVV &\equiv_L \hSuf{\alpha} (\hV \lor \hVV) \\
    % [\alpha] \hV \lor [\beta] \hVV &\equiv_L \hTru %\\
    % &\hSuf{\alpha} \hV \land \hSuf{\beta} \hVV &\equiv_L \hFls \\
    % [\alpha] \hV \land \hSuf{\alpha} \hVV &\equiv_L \hSuf{\alpha} (\hV \land \hVV) %\\
    % &[\alpha] \hV \lor \hSuf{\alpha} \hVV &\equiv_L [\alpha] (\hV \lor \hVV) \\
    %
    \label{eq:modal-trivial}
    [\Act] \hFls &\rewriteL \hFls%\\
    &\hSuf{\Act} \hTru &\rewriteL \hTru%\\
    &\hSuf{\act} \hFls &\rewriteL \hFls%\\
    &\hNec{\act}\hTru &\rewriteL \hTru\\
    \label{eq:absorb-trivial}
    \hTru \land \hV &\rewriteL \hV%\\
    &\hFls \land \hV &\rewriteL \hFls%\\
    &\hFls \lor \hV &\rewriteL \hV%\\
    &\hTru \lor \hV &\rewriteL \hTru
\end{align}
\begin{align}
  \label{eq:box-and-box}
    \bigwedge_{\act \in A}\hNec{\act}\hV_\act \land \bigwedge_{\act \in B}\hNec{\act}\hVV_\act &\rewriteL
    \bigwedge_{\act \in A\cap B}\hNec{\act} (\hV_\act \land \hVV_\act)
    \land
    \bigwedge_{\act \in A\setminus B}\hNec{\act} \hV_\act
    \land
    \bigwedge_{\act \in B \setminus A}\hNec{\act} \hVV_\act
    % \text{ enlarge $A$ w/ $[\_]\hTru$ to have same set for both conjuncts}
    \\
    \label{eq:box-or-box}
    \bigwedge_{\act \in A}\hNec{\act}\hV_\act \lor \bigwedge_{\act \in B}\hNec{\act}\hVV_\act &\rewriteL
    \begin{cases}
    \bigwedge_{\act \in A\cap B}\hNec{\act}(\hV_\act \vee \hVV_\act) &\text{ if $A\cap B \neq \emptyset$}\\
    $\hTru$ &\text{ otherwise}
    \end{cases}
     \\
     \label{eq:diamond-or-diamond}
    \bigvee_{\act \in A}\hSuf{\act}\hV_\act \lor \bigvee_{\act \in B}\hSuf{\act}\hVV_\act &\rewriteL \bigvee_{\act \in A \cap B}\hSuf{\act} (\hV_\act \lor \hVV_\act)
    \lor
    \bigvee_{\act \in A \setminus B}\hSuf{\act} \hV_\act
    \lor
    \bigvee_{\act \in B \setminus A}\hSuf{\act} \hVV_\act
    % \text{     enlarge $A$ w/ $\hSuf{\_}\hFls$ to have same set for both disjuncts}
    \\
    \label{eq:diamond-and-diamond}
    \bigvee_{\act \in A}\hSuf{\act}\hV_\act \land \bigvee_{\act \in B}\hSuf{\act}\hVV_\act &\rewriteL
    \begin{cases}
    \bigvee_{\act \in A\cap B}\hSuf{\act}(\hV_\act \wedge \hVV_\act) &\text{ if $A\cap B \neq \emptyset$}\\
    $\hFls$ &\text{ otherwise}
    \end{cases}
    % \bigvee_{\act \in A\cap B}\hSuf{\act}(\hV_\act \wedge \hVV_\act) \text{, if $A\cap B \neq \emptyset$, $\hFls$ otherwise}
    \\
    \label{eq:diamond-and-box}
    \bigwedge_{\act \in A}\hNec{\act}\hV_\act \land \bigvee_{\act \in B}\hSuf{\act}\hVV_\act  &\rewriteL \bigvee_{\act \in A \cap B} \hSuf{\act} (\hV_\acta \land \hVV_\acta) \lor \bigvee_{\act \in  B \setminus A} \hSuf{\act} \hVV_\acta
    \\
    \label{eq:diamond-or-box}
    \bigwedge_{\act \in A}\hNec{\act}\hV_\act \lor \bigvee_{\act \in B}\hSuf{\act}\hVV_\act  &\rewriteL \bigwedge_{\act \in A \cap B}\hNec{\act} (\hV_\act \land \hVV_\act) \wedge  \bigwedge_{\act \in A \setminus B}\hNec{\act} \hV_\act
\end{align}
\endgroup
\caption{\HML rewrite rules where $A,B \subseteq \Act$.
% , and $\hV \rewriteL \hVV$ stands for $\hmeaning{\hV}_L = \hmeaning{\hVV}_L$
$\rewriteL$ is the smallest binary relation on \HML that satisfies the rules above and
is closed with respect to HML contexts.
%whenever $\hVV_1 \rewriteL \hVV_2$, it is also the case that $\hV_1 \rewriteL \hV_2$, where $\hV_2$ results by replacing in $\hV_1$ an occurrence of $\hVV_1$ by $\hVV_2$.
}
\label{fig:fun-excercise}
\end{figure}

\begin{lem}\label{lem:slim_violations}
If $\hV \in \HML$ is slim and
$\hSemL{\hV} = \emptyset$
(\resp $\hSemL{\hV} = \Trc$), then $\hV = \hFls$ (\resp $\hV = \hTru$). \qed
\end{lem}

%In the following, we use the notations $\bigodot_{\act \in A} \mV_\act$ and $\mV_1 \paralG \cdots \paralG \mV_i  \paralG \cdots \paralG \mV_k$  to denote a combination of
%monitors
%using the parallel operator $\paralG$
%%
%% In the context that the notation is used,
%since the particular way the monitors are combined does not matter.
%Furthermore, since we are dealing with reactive monitors---and, as a consequence of \Cref{prop:monitor2automaton}, the parallel operators are associative with respect to verdict-equivalence---any way we combine the monitors with $\paralG$ will reach the same verdict for the same (finite) trace.

\begin{lem} \label{lem:slim-implies-tight}
If $\hV$ is a slim HML formula, then $m(\hV)$ is tight. % monitor for $\hV$.
\end{lem}
\begin{proof}
  By \Cref{prop:hml-monitorable},  $\tV {\notin} \hmeaning{\hV}_L$ implies that there is a finite prefix $\ftV$  of $\tV$ such that
  $m(\hV)\wtra{\ftV}\no$.
  We  prove by induction on \ftV that
  % if $\mV \not \wtraS{\ftr} \vV$, then there is some $\ftrr$, such that $\mV \wtraS{\ftr\ftrr} \vV' \neq \vV$.
  % This is proven for the case of $\vV = \no$, as the case of $\yes$ is symmetric.
  % The contrapositive statement, is that if
  % for no finite trace $\ftrr$, $\mV \wtraS{\ftr\ftrr} \yes$, then $\mV \wtraS{\ftr} \vV$. This, in turn, due to Proposition \ref{prop:hml-monitorable}, is equivalent to saying that
  if $\forall\tV.~
  % $, $
  % \ftr \tr \notin \hmeaning{\varphi}_L$,
  \rej{\mV,\ftV\tV}$,
  then $\mV \wtraS{\ftr} \no$ (the case for acceptance is symmetric).
  See \Cref{sec:complete-mon-app-tight} for details.
\end{proof}

We can transform every \HML\ formula into an equivalent slim formula.
This transformation is based on a set of rewrite rules of the form $\hV \rewriteL \hVV$, given in \Cref{fig:fun-excercise}, that allows us to iteratively replace the formula on the left-hand side with that on the right-hand side.
% equivalences between \HML\ formulae and it allows
% us to recursively turn an \HML\ formula that appears on the left-hand-side of an equivalence into a slim formula, by replacing it with the corresponding equivalent formula that appears on the right-hand side of that equation.
% These equivalences can be found in \Cref{fig:fun-excercise}.
% We observe that in the following equivalences, the right-hand-side formulae are of smaller size than their left-hand-side corresponding formulae.
% For convenience, in the following $\hV \equiv_L \hVV$ will mean that $\hmeaning{\hV}_L = \hmeaning{\hVV}_L$.
% , and $\acta \neq \actb$.
% we can transform an \HML\ formula to another, equivalent \HML\ formula, for which the monitor synthesis function will give a tight monitor.
%
% Let $\hV$ be a box-block if it is of the form $\bigwedge_{\act \in A}\hNec{\act}\hV_\act$; it is a diamond-block if it is of the form $\bigvee_{\act \in A}\hSuf{\act}\hV_\act$.
% We show how to eliminate conjunctions and disjunctions until all subformulae are either diamond-blocks or box-blocks.
% For linear time:

\begin{lem}
  \label{lem:rewrite-semantics-predicates}
  $\hV \rewriteL \hVV$ implies $\hSemL{\hV} = \hSemL{\hVV}$ and $l(\hV) > l(\hVV)$  \qed
\end{lem}

\begin{prop}[\HML normalisation]\label{prop:substitute-to-slim}
For every formula $\hV \in \HML$, there exists $k \leq l(\hV)$ such that $\hV = \hV_0 \rewriteL \hV_1 \rewriteL \ldots \rewriteL \hV_k=\hVV$ where $\hVV$ is slim and $\hSemL{\hV} = \hSemL{\hVV}$. \qed
\end{prop}

\begin{exmp} \label{exm:non-tightness-3}
Assume $\Act=\sset{a,b}$ and  consider the \emph{non}-slim \HML formula $\hV = \hAnd{\hSuf{a}{\hSuf{a}{\hFls}}}{\hNec{b}{\hFls}}$.
The synthesised monitor $\hSyn{\hV} = \ch{(\prf{a}{(\ch{\prf{a}{\no}}{\prf{b}{\no}})}}{\prf{b}{\no})} \paralC (\ch{\prf{b}{\no}}{\prf{a}{\no}})$ is \emph{not} tight. However, we can apply the transformations based on the given equivalences to obtain an equivalent slim formula
%(left to right)
thus:
\begin{math}
{\hSuf{a}{\hSuf{a}\hFls}} \land {\hNec{b}\hFls} \rewriteL {\hSuf{a}{\hSuf{a}\hFls}}
\rewriteL {\hSuf{a}{\hFls}} \rewriteL \hFls.
% \tag*{\qedd}
\end{math} \qedd
\end{exmp}

%% file: monitorability-pt3.tex
% !TEX root = main.tex

% In \Cref{sec:complete-monitorability} we showed that,

As opposed to the branching-time semantics of \UHML, where only properties that are semantically equivalent to $\hTru$ and $\hFls$ have sound and complete monitors~\cite{FraAI:17:FMSD}, the linear-time semantics permits a far richer class of complete-monitorable properties, namely \HML.
By some measures, however, this monitorable fragment is still quite restrictive.
For example, whereas the property ``\textit{initialise} occurs within the first ten actions'' can be expressed in terms of \HML, the property ``\textit{initialise} eventually occurs''---which can be expressed using least fixpoints---\emph{cannot}.  In fact, although the latter property \emph{cannot} be monitored for in a complete manner, it can be monitored completely for satisfaction.
%
% So far, we have shown that in contrast to branching-time, where only $\hTru$ and $\hFls$ have sound violation- \textit{and} satisfaction-complete monitors~\cite{FraAI:17:FMSD}, on linear time $\HML$ is monitorable according to this strong notion of complete monitorability.
%
In this section, we relax the notion of monitorability to \textit{partial-completeness}, which only requires a monitor to be either violation- or satisfaction-complete.

% This brings recursion back into our monitorable fragment.

\begin{defn}
	A formula $\hV \in \UHML$ is \emph{monitorable for satisfaction} (\resp \emph{for violation}) iff there exists a monitor \mV that is a sound and satisfaction-complete (\resp and violation-complete) monitor for \hV. It is \emph{partially-monitorable} when it is monitorable for satisfaction or for violation.
	\qedd
\end{defn}

We can extend these definitions to fragments of \UHML in a similar way to that in \Cref{def:complete-monitorability}.
Here, the trade-off between the guarantees we expect from monitors and the monitorable specifications is clear: for the linear-time interpretation, recursion can be traded for partial-completeness, while no such option exists for branching-time.
We can extend the observations of \Cref{sec:complete-monitorability} to the context of
partial
%complete
monitorability.

\begin{prop}\label{prop:vedict-equiv-implies-part-complete-formula}\label{prop:part-complete-semantic-equivalence}
	If \mV is sound and satisfaction-complete (\resp violation-complete) for \hV, then
	\begin{enumerate}
		\item $\mV \mveq \mVV$ implies \mVV is sound and satisfaction-complete (\resp violation-complete) for \hV.
		\item If for all  $\vV \in \{\yes,\no \}$, $\mVV \wtraS{\ftr} \vV$ implies $\mV \wtraS{\ftr} \vV$, then \mVV is sound for \hV.
		\item $\mV \maeq \mVV$ (\resp $\mV \mreq \mVV$) implies \mVV is satisfaction-complete (\resp violation-complete) for \hV.
		\item \mV is sound and satisfaction-complete (\resp violation-complete)
		%		monitor %%commented out for space
		for $\hV'$ implies $\hSemL{\hV}=\hSemL{\hV'}$.
		\qed
	\end{enumerate}
\end{prop}
% \begin{proof}
% 	Similar to the proof of \Cref{prop:complete-semantic-equivalence}.
% \end{proof}

\begin{exmp}
	\label{ex:not-partially-monitorable}
	Let $\Act = \{a,b,c\}$ and
	$\hV = \max X.(\hNec{b} \hFls \land \hNec{\{a,c\}} X) \lor \min Y.(\hSuf{c}\hTru \lor \hNec{\{a,b\}} Y )$,
	which is satisfied by traces
	of the form $(a+c)^\omega+\bigl((a+b)^* c (a+b+c)^\omega\bigr)$, \ie traces
	where either $b$ does \emph{not} appear, or $c$ \emph{does} appear.
	We show that \hV is not partially-monitorable.
	For if there was some \mV that is sound and satisfaction-complete for \hV, it should accept $a^\omega$;
	this means that \mV must reach \yes\ after analysing $a^k$ for some $k \geq 0$.
	In this case, the trace $a^kb^\omega$, which does not satisfy \hV, must also be accepted by \mV, resulting in a contradiction.
	If, on the other hand, there was some \mV that is sound and violation-complete for \hV, then it should reject $b^\omega$.
	%
	% Similarly to the above,
	Again, $\mV$ must reach \no\ after $b^k$ for some $k \geq 0$, but $b^kc^\omega$ satisfies \hV.
	Therefore, \hV \emph{cannot}  be partially monitorable.
	\qedd
\end{exmp}

%not sure what the structure of this section is, so I comment out:
%The structure of this section is simple: we define two fragments of $\UHML$, for which we define a monitor synthesis function. We then show that the synthesis produces sound partially complete monitors, and, finally, that any formula with a sound partially complete monitor is equivalent to a formula in one of these fragments.

For partial monitorability, we can identify two fragments of \UHML, namely $\ltmuC$, which is monitorable for satisfaction, and $\ltmuS$, which is monitorable for violation.

%these were called safety and cosafety fragmennts before, but I'm not sure it fits.
\begin{defn}[MAX and MIN Fragments of \ltmu] \label{def:minHML-maxHML}
	The greatest-fixed-point and least-fixed point fragments of \UHML are, respectively, defined as:
  \begin{align*}
      \hV,\hVV \in \ltmuS
      &\bnfdef  \hTru &
      &\bnfsepp  \hFls &
      &\bnfsepp \hOr{\hV\,}{\,\hVV}  &
      &\bnfsepp \hAnd{\hV\,}{\,\hVV} &
      &\bnfsepp \hSuf{A}{\hV} &
      &\bnfsepp \hNec{A}{\hV} &
      & \bnfsepp \hMaxXF
      \\
      \hV,\hVV \in \ltmuC
      &\bnfdef  \hTru &
      &\bnfsepp  \hFls &
      &\bnfsepp \hOr{\hV\,}{\,\hVV}  &
      &\bnfsepp \hAnd{\hV\,}{\,\hVV} &
      &\bnfsepp \hSuf{A}{\hV} &
      &\bnfsepp \hNec{A}{\hV} &
      & \bnfsepp \hMinXF \tag*{\qedd}
\end{align*}
\end{defn}

Both $\ltmuS$ and $\ltmuC$ are extensions of \HML.
We can extend the monitor synthesis from \Cref{def:mon-synt-complete} to these fragments by using the recursion that is available for monitors.

\begin{defn}[Monitor Synthesis]\label{def:formula-to-monitor-part}
	The monitor synthesis for $\ltmuS$ and $\ltmuC$ results by simply extending the definition of $\hSyn{-}$ from \Cref{def:mon-synt-complete} with the cases for the respective fixed-point of each fragment:
	$\hSyn{\hMaxXF} =
	\hSyn{\hMinXF} = \rec{x}{\hSyn{\hV}}$ and $\hSyn{\hVarX} = x$.
%\begin{equation*}
%\begin{split}
%  \hSyn{\hFls} &= \no \\
%  \hSyn{\hTru} &= \yes \\
%  \hSyn{\hVarX} &= x \\
%  \hSyn{\hAndF} &= \hSyn{\hV_1} \paralC \hSyn{\hV_2} \\
%  \hSyn{\hOrF} &= \hSyn{\hV_1} \paralD \hSyn{\hV_2} \\
%  \hSyn{\hNec{\ASet}{\hV}} &= \ch{\prf{\ASet}{\hSyn{\hV}}}{\uprf{\ASet}{\yes}} \\
%  \hSyn{\hSuf{\ASet} \hV} &= \ch{\prf{\ASet}{\hSyn{\hV}}}{\uprf{\ASet}{\no}} \\
%  \hSyn{\hMaxXF} &= \recX{\hSyn{\hV}} \\
%  \hSyn{\hMinXF} &= \recX{\hSyn{\hV}}
%\end{split}
%\end{equation*}
\qedd
\end{defn}

We observe that the extended monitor synthesis function still produces reactive monitors.
We also show the first important result of this subsection, namely that \Cref{def:formula-to-monitor-part} yields the required witness monitors to prove that the syntactic fragment $\ltmuS \cup \ltmuC$ is partially-monitorable.

\begin{prop}\label{prop:partial-reactive}
	For every $\hV \in \ltmuS \cup \ltmuC$, $\hSyn{\hV}$ is reactive. \qed
\end{prop}

\begin{prop} \label{lem:sound-and-rej-compl}
For every $\hV\in \ltmuS$, $\hSyn{\hV}$ is a sound and violation-complete monitor for $\hV$. For every  $\hV\in \ltmuC$, $\hSyn{\hV}$ is a sound and satisfaction-complete monitor for $\hV$.
%\qed
\end{prop}
\begin{proof}
	This requires us to prove soundness and violation/satisfaction-completeness for every $\hV\in \ltmuS$ and $\hV\in \ltmuC$ \resp as stated in \Cref{def:soundness-n-completeness}.  See \Cref{sec:partial-complete-mon-app}.
\end{proof}

%\ac{I left it here.}
As in the case of \Cref{sec:complete-monitorability}, we now turn our attention to the maximality of the syntactic fragment $\ltmuS \cup \ltmuC$ for partial-monitorability.
Particularly, we can define
two formula synthesis functions that produce partially monitorable formulae from monitors:
one maps
%\emph{consistent}
monitors to formulae in \ltmuS, and the other one to formulae in \ltmuC.
Depending on the fragment, we then show that if \mV is mapped to \hV, then \mV is sound and violation-complete, or satisfaction-complete \resp for \hV.
Here we only present the synthesis function for \ltmuS; the case for \ltmuC is dual.
%for which \mV is sound and violation-complete. Due to dealing with violation-completeness, this synthesis is asymmetric: $\act.\mV$ translates into $\hNec{\act}{\mSyn{\mV}}$, a $+$ translates into a conjunction and $\stp$ translates to $\hTru$. In the dual case of satisfaction-completeness, we would synthesise $\hSuf{\act}{\mSyn{\mV}}$ from $\act.\mV$, a disjunction from $+$ and $\hFls$ from $\stp$ instead.

\begin{defn}[\ltmuS Formula Synthesis] \label{def:regmon-to-formula}
%\begin{equation*}
\begin{align*}
\!\!  \mSyn{\no} &= \hFls %\\
  &\mSyn{\stp} &= \mSyn{\yes} = \hTru
  & \mSyn{x} &= X
%  \\
    &\mSyn{\recX{\mV}} &= \hMaxX{\mSyn{\mV}}\\
\!\!
  \mSyn{\esel{\mV}{\mVV}} &= \hAnd{f(\mV)}{\mSyn{\mVV}} %\\
    &\mSyn{m\paralC n} &= \hAnd{\mSyn{\mV}}{\mSyn{\mVV}} %\\
      &\mSyn{m\paralD n} &= \hOr{\mSyn{\mV}}{\mSyn{\mVV}} %\\
	&\mSyn{\prf{\acta}\mV} &= \hNec{\acta}{\mSyn{\mV}}  %\\
%  \mSyn{x} &= X
\tag*{\qedd}
\end{align*}
%\end{equation*}

\end{defn}

\begin{exmp}
	Let $\mV = \ch{\prf{a}{\prf{b}{\no}}}{\prf{a}{\prf{a}{\yes}}}$.
	Then,
$\mSyn{\mV} = \hAnd{\hNec{a}{\hNec{b}{\hFls}}}{\hNec{a}{\hNec{a}{\hTru}}}$ (which is equivalent to just
$\hNec{a}{\hNec{b}{\hFls}}$).
%
%$\mV$ accepts  traces of the form $aa \tV$ which all satisfy \mSyn{\mV}, and rejects traces
The monitor \mV rejects traces
of the form $ab \tV$ which are \emph{exactly}  all the traces violating \mSyn{\mV}. Thus $\mV$
is
 sound and violation-complete  for \mSyn{\mV}.  \qedd
 % $\hAnd{\hNec{a}{\hNec{b}{\hFls}}}{\hNec{a}{\hNec{a}{\hTru}}}$.
% Similarly,
% {\prf{a}{\prf{b}{\no}}}\paralD{\prf{\acta}{\prf{\acta}{\no}}}
% %on the other hand
% is sound and
% violation-complete
% $\hOr{\hNec{a}{\hNec{b}{\hFls}}}{\hNec{a}{\hNec{a}{\hFls}}}$.
\end{exmp}

Note that  $\mSyn{\mV} \in \ltmuS$, for any \mV.
However, when we apply the formula synthesis function from \Cref{def:regmon-to-formula} to a consistent monitor \mV to generate a formula \hV, and then apply the monitor synthesis from \Cref{def:formula-to-monitor-part} to
\hV, we will generate a monitor that has similar parts to \mV, but it will be somewhat different due to the asymmetry of the \resp syntheses.
For example, for $\Act = \sset{a,b}$, $\mSyn{\ch{\prf{a}{\no}}{\prf{b}{\yes}}} = \hAnd{\hNec{a}\hFls}{\hNec{b}{\hTru}}$,
and $\hSyn{\hAnd{\hNec{a}\hFls}{\hNec{b}{\hTru}}} = (\ch{\prf{a}{\no}}{\prf{b}{\yes}}) \paralC (\ch{\prf{b}{\yes}}{\prf{a}{\yes}})$.
The following lemma allows us to abstract from these discrepancies, thereby enabling the proof of \Cref{prop:form-synth-partial}.

\begin{lem}\label{lem:reject-same}
$\hSyn{\mSyn{\mV}}$
rejects the same traces as $\mV$. \qed
\end{lem}

\begin{prop}\label{prop:form-synth-partial}
	If \mV is consistent, then \mV is a sound and violation-complete monitor for $\mSyn{\mV}$.
\end{prop}

\begin{proof}
	From \Cref{lem:reject-same}, $\hSyn{\mSyn{\mV}}$
	rejects the same traces as $\mV$, and therefore, by \Cref{prop:part-complete-semantic-equivalence,lem:sound-and-rej-compl}, $\mV$ is violation-complete for $\mSyn{\mV}$.
	Since $\mV$ rejects the same traces as $\hSyn{\mSyn{\mV}}$,
	if $\mV$ rejects a trace $\tV$, then $\tV \notin \hSemL{\hV}$.
	Since \mV is consistent, if
	$\mV$ accepts a trace $\tV$, then
	it does not reject $\tV$, and because $\mV$ rejects the same traces as $\hSyn{\mSyn{\mV}}$,
	$\hSyn{\mSyn{\mV}}$ does not reject $\tV$ either.
	Since $\hSyn{\mSyn{\mV}}$ is also violation-complete
	by  \Cref{lem:sound-and-rej-compl}, this yields that
	$\tV \in \hSemL{\hV}$.
	Therefore, \mV is also sound for \hV.
\end{proof}

%The following proposition tells us that \ltmuS is the maximal, up to logical equivalence, fragment of \ltmu, that is monitorable for violation.
%Dually, \ltmuC is the maximal, up to logical equivalence,  fragment of \ltmu, that is monitorable for satisfaction.
The following proposition tells us that, up to logical equivalence, \ltmuS is the largest fragment of \ltmu\ that is monitorable for violation.
Dually, \ltmuC is the largest  fragment of \ltmu\ that is monitorable for satisfaction.

\begin{prop}\label{prop:linear-partial-maximality}
If a formula $\hV \in \ltmu$ has a sound and violation-complete
% parallel
monitor over infinite traces, then it is equivalent to a formula $\hVV \in \ltmuS$ over infinite traces.
\end{prop}
\begin{proof}
Let $\mV$ be a sound and violation-complete monitor for $\hV$ and let $\hVV = \mSyn{\mV}\in \ltmuS$ be the witness formula.
%
% From \Cref{def:regmon-to-formula}, $\hVV \in \ltmuS$.
%
Since $\mV$ is sound for $\hV$, it must be consistent, and by \Cref{prop:form-synth-partial},
\mV is a sound and violation-complete monitor for $\mSyn{\mV}$.
Therefore, by \Cref{prop:part-complete-semantic-equivalence}, $\hSemL{\hV} = \hSemL{\mSyn{\mV}}$.
%
%Let $\mVV$ be $\mV[\stp/\yes]$, which is also sound and violation-complete for $\hV$.
%From \cref{lem:sound-and-rej-compl}, $\hSyn{\mSyn{\mV}}$ is sound and violation-complete for $\mSyn{\mV}$.
%From \cref{lem:reject-same} $\mV$ and $m(f(\mV))$ reject exactly the same traces. $\mV$ is therefore also violation-complete for $\mSyn{\mV}$. Since $\mVV$ rejects the same traces as $\mV$ and accepts nothing, $\mVV$ is both sound and violation-complete for $\mSyn{\mV}$. Then $\hV$ is equivalent to $\mSyn{\mV}$, which is in \ltmuS.
\end{proof}

\begin{rem}
\Cref{thm:stronger-HML-maximality}
demonstrates that \HML can express any property of infinite traces that has a complete monitor in any monitoring system, assuming that verdicts remain irrevocable.
Unfortunately, this result cannot be replicated for partial completeness.
For instance, let $L \subseteq (\Act \setminus \{ c \})^*$ be a non-regular language, where $c \in \Act$ is some
% special symbol,
distinguished action,
 and $L_c = \setof{ \ftV c\tV}{\ftV \in L \text{ and } \tV \in \Act^\omega }$.
If $L_c$ could be expressed in
%$\CHML$
\ltmuC, then there would be a sound and satisfaction-complete monitor for $L_c$, and by a straightforward use of \Cref{prop:monitor2automaton}, we could construct a finite automaton that recognizes $L$, which contradicts the assumption that $L$ is non-regular.
Yet, we could imagine appropriate choices for $L$ and monitoring systems in which $L_c$ is monitorable.
%For instance,
 For instance, suppose that monitors are described using pushdown automata and
let $L$ contain exactly the finite words on $\{0,1\}$ that have the same number of occurrences of $0$ and of $1$.
%%%% the extra bit:
%Then, a monitor that counts the number of these occurrences and reaches \yes\ exactly when it reads $c$ for the first time and it has counted as many occurrences of $1$ as of $0$, is a sound and satisfaction-complete monitor for $L_c$.
%
% These can range from computationally demanding ones, such as an encoding of
% % the Halting Problem and
% Turing Machines that read until $c$, to more reasonable ones, such as a context-free language and monitors with some kind of memory.
%
\qedd
\end{rem}

%% file: monitorability-pt4.tex
% !TEX root = main.tex

% \ac{work in progress}

To synthesise a tight monitor for a formula $\hV$ of \ltmuS (or \ltmuC), one can synthesise a parallel monitor $\mV(\hV)$, then,
using the methods of Subsection \ref{subsec:transformations}, turn $\mV(\hV)$ into a verdict-equivalent deterministic regular monitor, and, finally,
consecutively replace instances of $\sum_{\act \in \Act} \act.\no$ and $\rec x \no$ by $\no$ and instances of $\sum_{\act \in \Act} \act.\yes$ and $\rec x \yes$ by $\yes$. The resulting monitor is tight.

\begin{lem} \label{lem:partially-complete-tight}
Let $\mV$ be a deterministic regular monitor, where $\sum_{\act \in \Act} \act.\no$, $\rec x \no$, $\sum_{\act \in \Act} \act.\yes$, and $\rec x \yes$ do not occur as submonitors. Then, $\mV$ is tight. \qed
\end{lem}

%Proof in \Cref{sec:partial-complete-tight-mon-app}.

We would like to be able to apply a convenient method to process the formula or the monitor, so that right after the monitor synthesis we could produce a tight monitor.
However, as we will see, a more reasonable monitor synthesis function that produces tight monitors is unlikely, as one could use it to solve the satisfiability problem for \ltmuS --- by checking whether a produced monitor for the formula immediately evaluates to $\no$ (or to \yes, for its negation), --- which is PSPACE-complete.

\begin{prop}\label{prop:pspace-hardness}
For $|\Act|\geq 2$, the satisfiability problem for \ltmuS is PSPACE-complete.
\end{prop}

\begin{proof}
Satisfiability for \ltmu (and therefore for \ltmuS as well) is known to be in PSPACE \cite{Vardi88}.
That satisfiability for \ltmuS is PSPACE-hard results from the observation that \ltmuS with at least two actions can encode the $1$-variable, diamond-free fragment of $D \oplus_\subseteq D4$, which is PSPACE-complete \cite{achilleos16}.
%\ac{Right now, I don't know of anything more reasonable to reduce from. I'm open to suggestions. I also could not find this stated/proven anywhere else. Again, open to suggestions.}
The reduction can be found in Appendix \ref{subsec:pspace-hardness}.
\end{proof}

\begin{rem}
For singleton $\Act = \sset{a}$, \ltmu-satisfiability is a lot simpler, as there is only one trace, $a^\omega$.
Therefore, satisfiability for \ltmuS can be reduced to model-checking on $a^\omega$.
A more direct way to solve satisfiability is to reduce the given formula by using the following straightforward rewrite rules:
$\hFls \land \hV \rewriteL \hFls$, $\hFls \lor \hV \rewriteL \hV$, $\hSuf{\act}\hFls \rewriteL \hFls$, $\hNec{\act}\hFls \rewriteL \hFls$, and $\max X.\hFls \rewriteL \hFls$; the cases for $\hTru$ are symmetric.
After applying these formula simplifications, we will either reach one of $\hTru,\hFls$, in which case the answer to satisfiability is obvious, or we will reach a formula $\hV$ without these constants.
In the latter case, we can easily see that $\hSyn{\hV}$ can never reach a verdict, and therefore it will never reject a trace, which, from \Cref{lem:sound-and-rej-compl}, implies that $\hSemL{\hV} = \Trc = \{ \act^\omega \}$. \qedd
\end{rem}

%% file: branching.tex
% !TEX root = main.tex

% \ac{I got this part}

% We branch and we collapse, yo!

\input{branching-intro}

\subsection{The Finfinite Domain}\label{sec:finfinite}
\input{finfinite}

\subsection{Monitorability over Finfinite Traces}\label{sec:mon-finfinite}
\input{finfinite-monitorability}

\subsection{Monitorable Formulae Across Semantics}\label{sec:compare}
%\subsection{A unifying monitorable fragment}
\input{collapse}

%\subsection{Comparing linear- and branching -time monitorability}
\input{compare}

%% file: branching-intro.tex
% !TEX root = main.tex

Monitorability over branching-time semantics has been examined in \citeMac{AceAFI:17:FSTTCS,AceAFI:18:FOSSACS} and \citeMac{FraAI:17:FMSD}  for various frameworks. In this section we compare the results of \citeMac{FraAI:17:FMSD}, the closest to our setting, with those of \Cref{sec:ltmu-monitorability}.
We begin by revisiting the basic definitions and results for branching-time monitorability. Then, in \Cref{sec:finfinite} and \Cref{sec:mon-finfinite}, we
%then
extend the study of monitorability to a domain that allows \emph{both} finite and infinite traces, and conclude, in \Cref{sec:compare}, by comparing the monitorable fragments in this domain to those in the branching-time
setting.
%domain.
All omitted proofs are in \Cref{sec:branching-appendix}.

%

%Furthermore, formulae with diamond modalities or disjunctions are not, in general, monitorable for violation, and
%formulae with box modalities or conjunctions are not, in general, monitorable for satisfaction.
%Therefore, parallel monitors do not offer any advantages over regular monitors, which are thus
%our monitoring system of choice.

\begin{defn}[Branching-time Monitor Soundness and Completeness] \label{def:sound-and-complete-branching}
\quad
\begin{itemize}
  \item A monitor \mV is \emph{sound} for a (closed) formula \hV over processes if, \emph{for all} $p\in \Proc$ of every LTS, \ie a triple
  $\langle \Proc,(\Act\cup\sset{\tau}),\reduc\rangle$:
  \begin{itemize}
  \item  \rej{\mV,p} implies $p\not \in \hSemB{\hV}$;
  \item  \acc{\mV,p} implies $p\in \hSemB{\hV}$.
  \end{itemize}
  \item A monitor \mV is \emph{violation-complete} for a
  % (closed)
  formula \hV over processes if \emph{for all} $p\in\Proc$ of every LTS, $p\notin\hSemB{\hV}$ implies  \rej{\mV,p}.
  It is \emph{satisfaction-complete} if $p\in\hSemB{\hV}$ implies \acc{\mV,p}. \qedd
  %
  % \item A monitor \mV is \emph{complete} for a (closed) formula \hV over processes if it is \emph{both} violation- and satisfaction complete for it. \qedd
\end{itemize}
\end{defn}

\begin{rem}
  The LTS is often omitted when it is clear from the context.
  As before, a monitor \mV is \emph{complete} for  \hV  if it is  violation- and satisfaction-complete for it.
  A rejection monitor is a monitor without the verdict \yes; an acceptance monitor is one without the verdict \no.  \qedd
\end{rem}

In the branching-time setting, monitors with both {$\yes$} and {$\no$} verdicts are unsound for any formula,
%\cite{FraAI:17:FMSD},
as whenever one trace leads to an acceptance and another to a rejection, one can easily construct a process that can emit both traces.
As a single-verdict (uni-verdict \cite{FraAI:17:FMSD}) monitor can only be either satisfaction- or violation-complete for a formula (except monitors for $\hTru$ and $\hFls$ which can be both), one cannot hope for complete monitors for $\UHML$, and therefore the best one can do is to identify its fragments for which partially complete monitors exist.
These are \SHML and \CHML, defined by the following grammars:

 \begin{defn}[Safety and Cosafety Fragments for Branching-time \recHML]\label{def:shml-chml}
\begin{align*}
\hV,\hVV \in \SHML &::= \hTru ~\mid~ \hFls ~\mid~ [\ASet] \hV ~\mid~ \hV \land \hVV ~\mid~ \max X.\hV ~\mid~ X \\
\hV,\hVV \in \CHML &::= \hTru ~\mid~ \hFls ~\mid~ \hSuf{\ASet} \hV ~\mid~ \hV \lor \hVV ~\mid~ \min X.\hV ~\mid~ X.
\tag*{{\qedd}}
\end{align*}
\end{defn}

% I would prefer to skip the monitor synthesis function:
%
% \begin{align*}
% \hSyn{\hTru}& = \stp
% & \hSyn{\hFls}& = \no
% & \hSyn{X} =  x
% \end{align*}\vspace{-4ex}
% \begin{align*}
% \hSyn{[\mu]\hVV}& =
% \begin{cases}
% \stp   &\text{if $\hSyn{\hVV}=\stp$} \\
% \mu.\hSyn{\hVV}        &\text{otherwise}                       \\
% \end{cases}
% &
% %\!\!\!\!\!\!\!\!\!\!
% \hSyn{\max X.\hVV}& =
% \begin{cases}
% \stp   &\text{if $\hSyn{\hVV}=\stp$} \\
% \rec\ x.\hSyn{\hVV}                &\text{otherwise}                       \\
% \end{cases}
% %\\ \\
% \end{align*}\vspace{-3ex}
% \begin{align*}
% \hSyn{\hVV_1\land\hVV_2}& =
% \begin{cases}
% \hSyn{\hVV_{1}}                &\text{if $\hSyn{\hVV_{2}}=\stp$} \\ %, for any $i \in \{1,2\}$}  \\
% \hSyn{\hVV_{2}}                &\text{if $\hSyn{\hVV_{1}}=\stp$} \\ %, for any $i \in \{1,2\}$}  \\
% %\hSyn{\hVV_2}                &\text{if $\hSyn{\hVV_1}=\stp$}  \\
% \hSyn{\hVV_1}+\hSyn{\hVV_2}   &\text{otherwise}                       \\
% \end{cases}
% \end{align*}

\begin{thm}[Branching-time Monitorability \cite{FraAI:17:FMSD}]\label{thm:branching-monitorability}
For every $\hV \in \SHML$, there is a regular rejection monitor $\mV$ that is sound and violation-complete for $\hV$.
For every $\hV \in \CHML$, there is a regular acceptance monitor $\mV$ that is sound and satisfaction-complete for $\hV$.
% over any LTS.
% For every regular rejection-monitor (\resp acceptance-monitor) $\mV$, there is a formula $\hV \in \SHML$ (\resp $\hV \in \CHML$), such that $\mV$ is sound and violation-complete (\resp and satisfaction-complete) for $\hV$ over any LTS.
\qed
\end{thm}

\begin{thm}[Maximality of \SHML and \CHML \cite{FraAI:17:FMSD}]\label{thm:branching-max}
% For every $\hV \in \SHML$ (\resp $\hV \in \CHML$), there is a regular rejection-monitor $\mV$ that is sound and violation-complete (\resp and satisfaction-complete) for $\hV$ over any LTS.
For every regular rejection monitor $\mV$, there is a formula $\hV \in \SHML$, such that $\mV$ is sound and violation-complete  for $\hV$.
% over any LTS.
For every regular acceptance monitor $\mV$, there is a formula $\hV \in \CHML$, such that $\mV$ is sound and satisfaction-complete for $\hV$.
% over any LTS.
\qed
\end{thm}

% \ac{some more comments here, possibly an example}

One can identify two key differences between the linear-time and the branching-time semantics introduced in \Cref{sec:preliminaries}.
The first and most characteristic difference is that for branching-time semantics, where formulae are interpreted over processes,
a process is allowed to emit more than one trace.
In other words, a process may exhibit different behaviour each time it runs, and therefore, a trace does not give the whole picture of its possible executions.
By contrast, for linear-time semantics, if one observes an action or a finite trace, then there is no possibility that another one could have been exhibited instead.
This allows for constructs such as parallel monitors to monitor for conjunctions and disjunctions at the same time:
simply decompose the formula as the monitor synthesis function directs in \Cref{def:mon-synt-complete,def:formula-to-monitor-part}, and let each monitor component examine the trace until a conclusion is reached.
For branching-time semantics, this method does not help to monitor a conjunction for satisfaction or a disjunction for rejection, as \citeMac{FraAI:17:FMSD} demonstrates.

\begin{exmp}
Consider $\hV = \hNec{a}\hFls \lor \hNec{b}\hFls \notin \SHML$. In contrast to the linear-time setting, $\hV$ is \emph{not} monitorable for violation under a branching-time interpretation.
For assume, towards a contradiction, that there
is a rejection monitor for $\hV$.
Assume an LTS with a process $p$ that has two transitions, $p \traS{a} \nil$ and $p \traS{b} \nil$.
Then, $p \notin \hSemB{\hV}$ and $p$ can produce three possible traces: $\varepsilon, a, b$.
If a monitor rejected one of these, say $a$, then it would reject $p$, but also
 process $q_a$ that has exactly one transition, $q_a \traS{a} \nil$.
 But we observe that
 $q_a\in\hSemB{\hV}$,
meaning that the monitor would \emph{not} be sound for \hV.
%contradicting soundness.
%
The formula $\hOr{\hNec{a}\hFls}{\hNec{b}\hFls}$ is however monitorable in a linear-time setting (\Cref{def:complete-fragment-HML,def:minHML-maxHML}).
\qedd
\end{exmp}

% For example, consider $\hV= \hSuf{\acta}\hTru \land \hSuf{\actb}\hTru$, for which there is no single trace that gives enough information to guarantee that $\hV$ is satisfied by the process. In contrast, in our linear-time semantics, $\hV$ is equivalent to $\hFls$.

The second difference
% between the two semantics
is that, in the linear-time semantics, formulae are only interpreted over \emph{infinite} traces while, in branching-time semantics, a trace is allowed to end.
Unlike the first difference, this one is not inherent to the linear- versus branching-time distinction, but it is one we have lifted from standard LTL-style semantics~\cite{Vardi88,BRADFIELDS:01:HandbookMu}.
Therefore, as a first step to reconcile the two semantics, we focus on this less essential difference
% between the two semantics
for our logic.

%% file: finfinite.tex
We introduce an alternative linear-time semantics for our logic, where formulae are interpreted over traces that are allowed to be either finite or infinite.
For convenience, we call these kinds of traces \emph{finfinite} and the resulting semantics \emph{finfinite} linear-time semantics, or just finfinite semantics.
(A semantics akin to ours for a linear-time temporal logic may be found in, for instance,~\citeMac{Schneider1997}. 
%The paper~
\citeMac{Falcone2012} define linear-time properties over finite and infinite traces, but do not consider a specific logic.)
The finfinite semantics, \hSemF{-}, is presented in  \Cref{fig:finfinite-recHML}.
The set of finfinite traces is $\fTrc = \Trc \cup \Act^*$ and we use $\fftV,\fftVV \in \fTrc$ (\resp $\FSet \subseteq \fTrc$) to range over (\resp sets of) finfinite traces.
% and to range over sets of finfinite traces.
% \ac{must probably add more here}

\begin{figure}[!h]
  % \textbf{Finfinite Linear-Time Semantics}
     \[\begin{array}{rlrl}
      \hSemF{\hTru,\sigma}  & \!\!\!\deftxt   \fTrc
      &
      % \\
      \hSemF{\hFls,\sigma}  & \!\!\!\deftxt   \emptyset
      \\
      \hSemF{\hOr{\hV_1}{\hV_2},\sigma} & \!\!\!\deftxt   \hSemF{\hV_1,\sigma} \cup \hSemF{\hV_2,\sigma}
      \qquad\qquad
       &
      % \\
      \hSemF{\hAnd{\hV_1}{\hV_2},\sigma} & \!\!\!\deftxt   \hSemF{\hV_1,\sigma} \cap \hSemF{\hV_2,\sigma}
      \\
      \hSemF{\hSuf{\ASet}{\hV},\sigma}  &
      \multicolumn{3}{l}{
           \!\!\!\deftxt \sset{\fftV \;|\; \exists \fftVV \cdot \exists \acta\in \ASet \cdot\fftV=\act\fftVV \;\text{ and }\; \fftVV \in \hSemF{\hV,\sigma}
          % suggestion:
          % \!\!\!\deftxt \sset{\act\fftV \;|\; \act\in \ASet  \;\text{ and }\; \fftV \in \hSemF{\hV,\sigma}
          }
      }
    \\
    \hSemF{\hNec{\ASet}{\hV},\sigma}  &
    \multicolumn{3}{l}{
       \!\!\!\deftxt \sset{\fftV \;|\;\forall \fftVV \cdot \forall \acta \in \ASet \cdot  \tr=\acta\fftVV \;\text{ implies }\;\fftVV \in \hSemF{\hV,\sigma}
      % suggestion:
      % \!\!\!\deftxt \sset{\fftV \;|\; \fftV = \act\fftVV \text{ and } \act \in \ASet \;\text{ implies }\;\fftVV \in \hSemF{\hV,\sigma}
       }
     }

    \\
    \hSemF{\hMin{\!\hVarX}{\hV},\sigma} & \!\!\!\deftxt
    \bigcap \sset{\FSet \;|\;  \hSemF{\hV,\sigma[\hVarX\mapsto \FSet]} \subseteq \FSet\ }

    \\
    \hSemF{\hMax{\!\hVarX}{\hV},\sigma} & \!\!\!\deftxt
    \bigcup \sset{\FSet \;|\;  \FSet \subseteq \hSemF{\hV,\sigma[\hVarX\mapsto \FSet]}\ }
    \quad\;
    &
    % \\
    \hSemF{\hVarX,\sigma} & \!\!\!\deftxt   \sigma(\hVarX)
    % \\
    % \\
    \end{array}
    \]
  \caption{Finfinite Linear-Time Semantics}
  \label{fig:finfinite-recHML}
\end{figure}

%% OK, why is the environment given as a plain \sigma without macros???

\begin{rem}
% As we have seen in  \Cref{sec:preliminaries},
  For \ltmu, \hSuf{\Act}{\hV} and \hNec{\Act}{\hV} can be seen as the strong and weak next operators, $X \hV$ and $\overline{X}\hV$ from LTL~\cite{Clarke1999Book}.
  In this same setting,
  % the linear-time setting over infinite traces,
  \hNec{\ASet}{\hV} may be seen as shorthand for \hOr{ \hSuf{\coASet}{\hTru}}{\hSuf{\ASet}{\hV}}.
  However the encoding does \emph{not} work for the finfinite interpretation of \Cref{fig:finfinite-recHML}.
  % infinite \Act or
  % linear-time models defined over finfinite traces.
  \qedd
\end{rem}

The two linear-time semantics for \recHML still correspond in some sense; see \Cref{lem:infinite-to-finfinite-semantics}.
In particular, formula equivalence over finfinite traces implies equivalence over infinite traces.

\begin{lem}\label{lem:infinite-to-finfinite-semantics}
For all $ \hV\in \UHML$, $\hSemF{\hV}\cap \Trc = \hSemL{\hV}$.
\qed
%\hSemL{\hV} \subseteq \hSemF{\hV}$ and $\Trc \setminus \hSemL{\hV} \subseteq \fTrc \setminus \hSemF{\hV}$.
\end{lem}

% \ac{we should perhaps just have the following corollary as a comment/inline/remark/.... I'm not sure if we have defined formula equivalence, so it is a bit informal. I leave it to F to decide.}
%
% \begin{cor}\label{cor:finfinite-infinite-eq}
% Formula equivalence over finfinite traces implies equivalence over infinite traces.
% \qed
% \end{cor}

% \ac{here we need to talk about the monitors in this setting}

We consider the same monitoring systems of regular and parallel monitors that were introduced in  \Cref{sec:monitors}.
However, what it means for $\mV$ to monitor for $\hV$ depends on the semantics that we use for the formulae: the definition used in  \Cref{sec:ltmu-monitorability} is therefore not sufficient for the finfinite domain.

\begin{defn}[Finfinite Linear-time Monitor Soundness and Completeness] \label{def:sound-and-complete-finfinite}
\quad
\begin{itemize}
  \item A monitor \mV is \emph{sound} for a (closed) formula \hV over finfinite traces if, \emph{for all} $\fftV\in \fTrc$:
  \begin{itemize}
  \item  \rej{\mV,\fftV} implies $\fftV\not \in \hSemF{\hV}$;
  \item  \acc{\mV,\fftV} implies $\fftV\in \hSemF{\hV}$.
  \end{itemize}
  \item A monitor \mV is \emph{violation-complete} for a formula \hV over finfinite traces if \emph{for all} $\fftV\in \fTrc$, $\fftV\not\in\hSemF{\hV}$ implies \rej{\mV,\fftV}.
  It is \emph{satisfaction-complete} if $\fftV\in \hSemF{\hV}$ implies \acc{\mV,\fftV}.
  %
  % \item A monitor \mV
  It is \emph{complete} for a formula \hV over finfinite traces if it is \emph{both} violation- and satisfaction complete for it. \qedd
\end{itemize}
\end{defn}

% \ac{I think we must start specifying under what semantics a monitor is sound/complete/... for a formula}

Monitorability of formulae and logics can be adjusted to finfinite traces analogously.

%% file: finfinite-monitorability.tex
% !TEX root = main.tex

We now identify the complete- and partial-monitorable fragments of $\UHML$ over finfinite traces. Our first observation is that under finfinite semantics, there are no complete-monitorable formulae, except the ones equivalent to $\hTru$ or $\hFls$.

\begin{lem}\label{lem:finfinite-trivial-mon}
 If $\mV$ is sound and complete for $\hV$ over finfinite traces, then $\hSemF{\hV} = \fTrc$ or $\hSemF{\hV} = \emptyset$.
% is equivalent to $\hTru$ or to $\hFls$.
 \qed
\end{lem}

\begin{rem}
% We note that the claim of
\Cref{lem:finfinite-trivial-mon} holds regardless of the considered logic: due to verdict-persistence (\Cref{lem:ver-persistence}), a logical fragment that is complete-monitorable over finfinite traces must be trivial for any logic interpreted over finfinite traces. \qedd
\end{rem}

% \begin{rem}
The concept of tightness, as defined in \Cref{def:tight}, does not apply for the finfinite interpretation since there is no guarantee that a finfinite trace will have a continuation.
A definition of tightness
% that relates to a formula
might stipulate that a rejection-monitor is tight for a formula when it is guaranteed to reject any finite trace as long as the trace and
all of its (finfinite) continuations violate the formula (\ie bad prefixes).
However, this notion of tightness is implied by partial completeness.
% :
% %
%  it is not hard to see that any violation-complete rejection-monitor for $\hV$ must also be tight for $\hV$.
% \end{rem}

% \begin{prop}
% If $\mV$ is sound and satisfaction- or violation-complete for a formula, then it must also be tight.
% \end{prop}

% \begin{proof}
% I'm not sure what tightness is yet.
% \end{proof}

\begin{exmp}\label{ex:one-mod}
In contrast to the infinite trace semantics,
% the formula
$\hSuf{a}\hTru$ is not monitorable for violation under finfinite semantics.
%
% To see this,
For assume towards a contradiction that $\mV$ is a monitor that is sound and violation-complete for $\hSuf{a}\hTru$.
Then, $\mV$ must reject the empty trace, $\varepsilon$, and thus all of its extensions, including $a \in \hSemF{\hSuf{a}\hTru}$, making $\mV$ unsound.
Similarly, $[\act]\hFls$ is not monitorable for satisfaction.
\qedd
\end{exmp}

Our next goal is to characterize the expressive power of monitors in finfinite semantics. To this end, we identify the following fragments of \UHML.
%We identify the following syntactic fragments of $\UHML$.
Only one type of modality is kept in each of these fragments.
This is because, as observed in \Cref{ex:one-mod}, the two modalities are not mutually expressive and even simple formulae using them are not monitorable for violation or satisfaction.
%As \Cref{ex:one-mod} illustrates, unlike in $\maxHML$ and $\minHML$, each fragment only has one of the two modalities.
%
\begin{defn}
\begin{align*}
\hV,\hVV \in \ftmuS &::= \hTru &&\mid &&\hFls &&\mid &&[\ASet] \hV &&\mid &&\hV \lor \hVV &&\mid &&\hV \land \hVV &&\mid &&\max X.\hV &&\mid &&X, \text{ and} \\
\hV,\hVV \in \ftmuC &::= \hTru &&\mid &&\hFls &&\mid &&\hSuf{\ASet} \hV &&\mid &&\hV \lor \hVV &&\mid &&\hV \land \hVV &&\mid &&\min X.\hV &&\mid &&X.
\tag*{\qedd}
\end{align*}
\end{defn}

The next lemma formalises the property that formulae in $\ftmuS$ denote prefix-closed sets of (finfinite) traces whereas formulae  in $\ftmuC$ denote suffix-closed sets of traces.

%The next lemma
%formalises the property
%that $\ftmuS$ is prefix-closed and $\ftmuC$ is suffix-closed.
\begin{lem}\label{lem:finfinite-suffix-closed}
For all $\ftV \in \Act^*$ and $\fftV \in \fTrc$, $(i)$ if $\hV \in \ftmuS$ and $\ftV\fftV \in \hSemF{\hV}$, then $\ftV \in \hSemF{\hV}$;
$(ii)$ if $\hV \in \ftmuC$ and $\ftV \in \hSemF{\hV}$, then $\ftV\fftV \in \hSemF{\hV}$.
\qed
\end{lem}

% \begin{proof}
% 	To proof uses induction on $\ftV$ and the depth of the top-level modalities in $\hV$, and can be found in \Cref{sec:branching-appendix}.
% \end{proof}

Interestingly, for \ftmuS\ and \ftmuC\ over finfinite traces, we can use the same monitor synthesis function that we used to generate monitors for \maxHML and \minHML  over infinite traces.

\begin{prop} \label[lem]{lem:sound-and-compl-finfinite}
For every $\hV\in \ftmuS$, $\hSyn{\hV}$ is sound and violation-complete for $\hV$ over finfinite traces.
For every $\hV\in \ftmuC$, $\hSyn{\hV}$ is sound and satisfaction-complete for $\hV$ over finfinite traces. \qed
\end{prop}

% \begin{proof}
% 	To proof uses \Cref{lem:sound-and-rej-compl} and also uses induction on the depth of the top-level modalities in $\hV$. The full proof  can be found in \Cref{sec:branching-appendix}.
% \end{proof}

% We can assume an LTS where for every $\fftV$, there is some process
% $p_\fftV$, such that $p_\fftV \wtraS{\ftV}$ if and only if $\ftV$ is a prefix of $\fftV$.
% Then, $p_\fftV$ is called a trace-process.

To facilitate our comparisons between the finfinite and the branching-time interpretations of \recHML, we define the notion of trace-processes.

\begin{defn}
Process $p$ is a \emph{trace-process} when $p \traS{\actu} q$ and $p \traS{\actu'} q'$ implies $\actu=\actu'$, $q = q'$ and $q$ is a trace-process.
% as well.
 %there exists a process $p_\fftV$, such that
A (trace) process $p$ \emph{represents} a finfinite $\fftV$ when
$p \wtraS{\ftV}$ iff $\ftV$ is a prefix of $\fftV$.
% , and $p \traS{\actu} q$ and $p \traS{\actu} q$
% for each $k \geq 0$, $p (\reduc)^k q$ and $p (\reduc)^k q'$ implies that $q = q'$.
% .
\qedd
\end{defn}

For a trace $\fftV$, we can assume the existence of a trace-process $p_\fftV$ that represents $\fftV$:
 one can construct such a trace-process $p_\fftV$ whereby its states are
 all the prefixes of $\fftV$ and its transitions are
 those
% all
 of the form $\ftV \traS{\act} \ftV\act$, where  $\ftV$ and $\ftV\act$ are prefixes of $\fftV$.

 \begin{rem}
 	We note that, unlike for monitors, we have \emph{not} assumed any specific syntax for processes, which
% 	we allow to
 	can
 	come from an arbitrary LTS.
 	This makes it possible to represent every finfinite trace, even one without a finite representation, by a process.
 	\qedd
 \end{rem}

\begin{exmp}
A process representing $a b$ is the three-state process $p$, with \emph{just} the transitions  $p \traS{a} p'$ and $p' \traS{b} \nil$.
A process representing $a^\omega$ is $q$ that has exactly one transition, $q \traS{a} q$.
\qedd
\end{exmp}

\Cref{lem:branching-finf-match} shows that, for \UHML, (finfinite) traces and trace-processes are different descriptions of the same model.

%\begin{lem}\label{lem:trace-processes-are-deterministic}
%If $p$ represents $\act \fftV$ and  $p \reduc q$, then $q$ represents $\fftV$.
%\end{lem}

\begin{lem}\label{lem:branching-finf-match}
If $p$ represents $\fftV$, then
$\fftV \in \hSemF{\hV}$\ iff\ $p \in \hSemB{\hV}$.
\qed
\end{lem}

% The following proposition shows that
Coincidentally, all formulae that are monitorable for violation or satisfaction over a finfinite semantics are equivalent to $\SHML$ or $\CHML$ formulae \resp from \Cref{def:shml-chml}.
Since $\ftmuS$ and $\ftmuC$  syntactically subsume $\SHML$ and $\CHML$ \resp they are maximally monitorable fragments of $\UHML$ when interpreted over finfinite traces.

%As a corollary, we obtain that $\SHML$ and $\CHML$ are as expressive as $\ftmuS$ and $\ftmuC$ respectively on finfinite traces.

%demonstrates that $\ftmuS$ and $\ftmuC$ are semantically maximal fragments of $\UHML$ with
%respect to violation- and satisfaction-completeness, respectively, over finfinite traces.
% Therefore, $\ftmuS \cup \ftmuC$ is semantically maximal with respect to partial completeness over finfinite traces.
%However, so are $\SHML$ and $\CHML$, respectively, implying that
%$\ftmuS$ and $\SHML$ have the same expressive power over $\fTrc$.

\begin{prop}\label{prop:finfinite-maximality}
If $\hV \in \UHML$ has a
% if $\mV$ is a
sound and violation-complete (\resp satisfaction-complete)
% regular or
reactive parallel
monitor over finfinite traces, then there is some
$\hVV \in \SHML$ (\resp $\hVV \in \CHML$) that is equivalent to $\hV$ over finfinite traces.
% $\fftV \in \hSemF{\hV}$ if and only if $p_\fftV \in \hSemB{\hV}$.
\end{prop}

\begin{proof}
% We observe that $\fftV \mapsto p_\fftV$ is a bijection from finfinite traces to trace-processes.
% Consider the LTS of trace-processes, that is, the processes $p_\fftV$ for finfinite $\fftV$ such that $p_\fftV \wtraS{\ftV}$ if and only if $\ftV$ is a prefix of $\fftV$.
Let $\mV$ be a sound and violation-complete
% regular or
reactive parallel monitor for $\hV$ over finfinite traces.
By  \Cref{prop:extended-mon-to-reg-mon}, there is a regular monitor $\mVV$ that is verdict-equivalent to $\mV$,
so it is also sound and violation-complete for $\hV$ over finfinite traces.
We can then obtain a single-verdict monitor $\mVV'$ from $\mVV$ that is rejection equivalent to it by swapping any \yes\ with \stp.
$\mVV'$ is thus still sound and violation-complete for $\hV$ over finfinite traces.
From \Cref{thm:branching-max}
%It has been shown in \cite{FraAI:17:FMSD} (see also \cite{AceAFI:17:FSTTCS}) that
there is a formula $\hVV \in \SHML$, such that $\mVV$ is sound and violation-complete for
$\hVV$ over all LTSs, including the LTS of trace-processes.
Since $\mVV$ is sound and violation-complete for $\hVV$ on trace processes,
$p_\fftV \in \hSemB{\hVV}$ is equivalent to claiming that $\mVV$ does not reject any trace that $p_\fftV$ can produce.
However, this is equivalent to saying that $\mVV$ does not reject $\fftV$ which, by violation-completeness, is equivalent to $\fftV \in \hSemF{\hV}$.
% Therefore, to complete this proof it suffices to prove that
By \Cref{lem:branching-finf-match},
$\fftV \in \hSemF{\hVV}$ iff $p_\fftV \in \hSemB{\hVV}$,
and the proof is complete.
% which can be done by straightforward induction on $\hVV$.
The case for a satisfaction-complete monitor is similar.
\end{proof}

%% file: collapse.tex
% !TEX root = main.tex

So far, we have identified a different pair of partial-monitorable syntactic fragments for each of the three semantics that we have presented in this paper.
However, as the reader may suspect from \Cref{prop:finfinite-maximality}, we may be able to further restrict the syntax that we allow for our formulae, and still be able to express all monitorable formulae, and therefore, an identified maximally monitorable fragment of \UHML may be equally expressive as a syntactic fragment of its own.

Here we show that \emph{for each of the semantics that we have presented,}
\ie over infinite traces, finfinite traces, and processes,
%in each case
$\SHML$ and $\CHML$ are equally expressive
%\emph{with respect to the semantic domain}
as the corresponding identified partially monitorable fragment. That is to say, $\SHML$ is as expressive as $\ftmuS$ over finfinite traces and as expressive as $\maxHML$ over infinite traces --- and
dually, $\CHML$ is as expressive as $\ftmuC$ over finfinite traces and as expressive as $\minHML$ over infinite traces.
%trivially, as expressive as \SHML over processes.

\begin{prop}\label{prop:finfinite-max}
If $\hV \in \ftmuS$ (\resp $\hV \in \ftmuC$), %has a
% if $\mV$ is a
% sound and violation-complete (\resp satisfaction-complete) regular or reactive parallel monitor,
then there is some
$\hVV \in \SHML$ (\resp $\hVV \in \CHML$) that is equivalent to $\hV$ over finfinite traces.
% $\fftV \in \hSemF{\hV}$ if and only if $p_\fftV \in \hSemB{\hV}$.
\qed
\end{prop}

\begin{prop}\label{prop:infinite-max}
%For finite \Act, i
If $\hV \in \maxHML$ (\resp $\hV \in \minHML$), %has a
% if $\mV$ is a
% sound and violation-complete (\resp satisfaction-complete) regular or reactive parallel monitor,
then there is some
$\hVV \in \SHML$ (\resp $\hVV \in \CHML$) that is equivalent to $\hV$ over infinite traces.
\qed
\end{prop}

The proofs of both of these propositions proceed by considering a sound and partially complete monitor for a formula in \ftmuS, \maxHML or their duals, and using the formula synthesis to find an \SHML formula that is equivalent to the original formula
on finfinite and infinite traces respectively.
The full proofs can be found in \Cref{sec:branching-appendix}.

%However, the parallel monitors seem likely to be more concise.

%% file: compare.tex
% !TEX root = main.tex

\smallskip

The import of \Cref{prop:finfinite-max,prop:infinite-max} is that, in settings where \Act is finite, logical fragment $\SHML\cup\CHML$ can be used to \emph{syntactically} characterise the class of monitorable properties (for sound and partial-completeness) for all three interpretations (\ie traces, finfinite traces and processes).
In spite of this felicitous (and somewhat surprising) result, one should nevertheless stress that their interpretation is still \emph{semantically different}.
In fact, the synthesised monitors presented here in \Cref{def:mon-synt-complete,def:formula-to-monitor-part} yield \emph{behaviourally different} monitors to those obtained by the synthesis in \citeMac{FraAI:17:FMSD}.
Moreover, they can \emph{not} be used interchangeably: \Cref{def:mon-synt-complete,def:formula-to-monitor-part} produce multi-verdict monitors, even when applied to the syntactic fragment $\SHML\cup\CHML$, which makes them immediately unsound for a branching-time interpretation.
In \Cref{prop:shml-and-forall-traces}, we can however show that the monitors synthesised by the procedure of \citeMac{FraAI:17:FMSD} for the \SHML fragment qualify also as correct monitors for the finfinite interpretation of the logic.
This means that the tools developed in \citeMac{Attard:17:Book} and \citeMac{Attard:16:RV}, which are based on the branching-time synthesis of \citeMac{FraAI:17:FMSD}, can be used out of the box to monitor for finfinite properties.

% We conclude the section with an observation that links the interpretation of \SHML formulae over processes
% to their interpretation over finfinite traces.

\begin{prop}\label{prop:shml-and-forall-traces}
For a process $p$ and a formula $\hV\in\SHML$, the following are equivalent:
% \begin{enumerate}
% \item
$(i)$ $p\in \hSemB{\hV}$  % \label{one}
and
% \item
$(ii)$ If $p$ produces a finfinite trace $\fftV$, then $\fftV \in \hSemF{\hV}$ %\label{two}
.
\qed
% \end{enumerate}
\end{prop}

%% file: conclusion.tex
% !TEX root = main.tex

%Points to mention:
%\begin{itemize}
%	\item Compare with \cite{AceAFI:18:FOSSACS} as a possible extension.
%	\item Observe the application of the theory as the tool detectEr.
%	Mentioned in sec 5 (A).
%	\item Maximality results - we are the only ones to do this, but I couldn't flaunt it in the introduction.
%	I tried to (A).
%\end{itemize}

%\subsection{Our Work}

We have presented a systematic study of the monitorability of \UHML, a highly expressive specification logic: we have developed results relating to its linear-time interpretation and established correspondences with previous monitorability results for the branching-time interpretation of the logic.
This allows us to use existing RV tools (developed for branching-time) to monitor linear-time \UHML properties.
To our knowledge, this is the first study of monitorability that spans across the  linear-time/branching-time spectrum.
Moreover, although monitorability has been studied extensively for linear-time specifications, we are unaware of any maximality results such as those presented in
 \Cref{prop:hml-maximal,thm:stronger-HML-maximality,prop:linear-partial-maximality,prop:finfinite-maximality,prop:finfinite-max,prop:infinite-max}.

Concretely, in \Cref{sec:monitors}, we introduce parallel monitors and we gave a way to construct a deterministic regular monitor (introduced in \citeMac{AceAFI:2017:CIAA} and \citeMac{FraAI:17:FMSD}) from a parallel one, establishing that the two monitoring frameworks are equivalent with respect to the properties they can monitor.
%
% Regular monitors were introduced in \cite{FraAI:17:FMSD}, yet the monitoring system that is used in that paper is not exactly the same as the one that we use in this paper.
%
% The difference is that as Francalanza \etal establish in  \cite{FraAI:17:FMSD},
% to ensure that monitors are sound for any formula in the branching-time setting, one must only consider single-verdict monitors, but for the linear-time framework that we consider, monitors can potentially return both verdicts, but still be sound.
%
In \Cref{sec:ltmu-monitorability}, we give a natural monitor synthesis from
%\HML, \ltmuS, and \ltmuC,
three fragments of \UHML
to parallel monitors, and establish that the resulting monitors satisfy the requirement of
soundness and  a version of the requirement for completeness.
%sound and complete, violation-complete, and satisfaction-complete respectively.
%correct for a corresponding notion of correctness that includes soundness and a version.
For complete monitors, we identify the requirement of tightness and show how one can satisfy it.
In \Cref{sec:branchingtime}, we see how these findings apply in the intermediate finfinite setting, and
we establish that \SHML has the same expressive power as the respective maximal monitorable fragments of \UHML
in the finfinite and infinite-trace settings.

%tentative title...
% \subsection
\paragraph{Multiple Ways to Monitor}
These results show that there is more than one way to
monitor for a property \hV that is monitorable for violation.
If \hV is already in $\SHML$, or if we want to make the effort to write the property as an \SHML formula, we can use the monitor synthesis in \citeMac{FraAI:17:FMSD} to synthesise a single-verdict, sound and violation-complete regular monitor for \hV that will work in all (infinite-trace, finfinite, and branching-time) semantics.
Alternatively, if we are interested in the linear-time domain (for either infinite or finfinite traces), we can synthesise a parallel monitor with the synthesis function from \Cref{def:formula-to-monitor-part}, hoping that the possibly dual-verdict monitor may occasionally report the satisfaction of the formula, providing us with more information.
In the latter case, we may choose to deploy the parallel monitor as is, or use the construction from \Cref{prop:extended-mon-to-reg-mon} to obtain a verdict-equivalent regular monitor.
An advantage of using the parallel monitor is that
% we conjecture that parallel monitors
it can be  significantly more concise than a regular monitor, at least at the early stages of the computation.
An advantage of using a regular monitor is that it is guaranteed to be finite state (\Cref{prop:reg-mon-fin-state}).
% , while parallel monitors are not.
Furthermore, regular monitors can be determinized and then minimized (see \Cref{prop:determinization} and \citeMac{determinization}), making their implementation more straightforward.
Therefore, one can think of \ltmuS as a high-level specification language for properties that are monitorable for violation in the linear-time setting.
From \ltmuS, we can generate parallel monitors that can then be compiled into (deterministic, minimized) regular monitors that can be implemented and deployed to monitor the system.
On the other hand, \SHML can be thought of as a lower-level language that is closer to regular monitors and can allow for better fine-tuning of the monitor's behaviour, and avoids the cost of constructing a regular monitor.
%
%
%
%As a consequence of \Cref{prop:finfinite-max,prop:infinite-max}, for any monitorable formula $\hV$ in \UHML, the monitor synthesis function from \cite{FraAI:17:FMSD}, when applied to the equivalent \SHML formula, gives a monitor $\mV$ that is sound and violation-complete for $\hV$ over processes.
%But then, \Cref{lem:branching-finf-match} implies that $\mV$ is sound and violation-complete for $\hV$ over (infinite or finfinite, respectively) traces.
%Therefore, the techniques from \cite{FraAI:17:FMSD}, which do not use parallel monitors, can be used as an alternative tool for building monitors for \UHML in both the linear- and branching-time setting.
%
%Nonetheless, given the relationship between regular monitors and NFAs, and parallel monitors and alternating finite automata (see \Cref{subsec:transformations}), and the exponential explosion when one transforms an alternating automaton to an NFA \cite{FJY90},
%we conjecture that
%parallel monitor are, in general, significantly more succinct than
%their verdict-equivalent regular monitors.
%Furthermore, with parallel monitors we can achieve a natural monitor synthesis function (\Cref{def:formula-to-monitor-part}).
%Then,  one can choose to either use the resulting parallel monitor, or transform it into a regular one, using the construction from \Cref{prop:extended-mon-to-reg-mon}.

% \subsection
\paragraph{Future Work}
% As mentioned above,
We are interested in a detailed taxonomy and comparison of different notions of monitorability, and this work is a first step in that direction. Additionally, in \citeMac{AceAFI:18:FOSSACS}, the authors examine how the set of monitorable properties can be extended by encoding additional information into the trace that describes a system execution. Noticeably, their framework allows for the interaction of multiple verification methods, and this is an approach we would like to explore for our own framework.

% \subsection{Related Work}
\label{sec:related}
\input{related.tex}

% Consistently detecting monitors \cite{Fra:17:CONCUR} as a

%% file: related.tex
% !TEX root = main.tex

\paragraph{Related Work on Runtime Verification}

RV has been
%used
applied
in the computer-aided
verification of complex programs and models written in a variety of
high-level languages. For example, RV has been
%applied
used
in the
verification of properties written in an extension of PSL and SVA over
SystemC models in~\citeMac{TabakovRV12}
(but see \citeMac{PnueliZaks:06:FM} and references in~\citeMac{TabakovRV12} for earlier work on monitor synthesis for PSL).
% (SystemC is an IEEE standard
% based on
% %an extension of
% C++.)
Like we do in this paper, Tabakov \etal
% Rozier and Vardi
argue for the algorithmic generation of ``correct''
monitors from properties. However, their focus
% in the aforementioned paper
is on an experimental study of
%the
monitor-generation
%of monitors
procedures
that
offer the best performance in terms of runtime overhead at simulation
time. In order to do so, they employ the CHIMP tool~\cite{DuttaVT13} to
generate monitors (represented as DFAs) from LTL properties using a
number of workflows that take into account various options regarding
state minimization, alphabet representation, alphabet minimization and
the representation of the transition function of the monitor.

\paragraph{Diagnosability}

It is worth mentioning here work on \emph{diagnosability}, \eg  \citeMac{sampath1995diagnosability,bertrand:hal-01088117}.
Diagnosability is a similar notion to the one of monitorability.
What is different is that, for diagnosability one knows a model of the system,
%to be diagnosed,
and then, by observing the visible events of a system run,
infers whether an unobservable fault event has occurred during this run.
A further goal is to diagnose the kind of fault event that has occurred.
Typically, the detection and diagnosis of fault events is performed by a diagnoser, which is synthesised from the model of the system.
Although RV and diagnosability appear, at first glance, to work in different ways,  one can view diagnosability as the runtime monitoring of a set of trace-properties (the occurrence of different types of fault events), using
information about the system's branching structure, in a framework that considers unobservable events --- as in \citeMac{FraAI:17:FMSD,AceAFI:17:FSTTCS}.
We feel that there is significant potential in addressing the two areas in a more unified manner. This is an interesting avenue for future research.

\paragraph{Related Work on Specification logics}

$\UHML$ is a
multi-modal
%the process-algebra flavoured
variant of the $\mu$-calculus that
is interpreted over
%operates on
edge-labelled LTSs rather than node-labelled ones. The distinction is mainly a question of presentation; how to go between the two types of models is discussed by
%De Nicola and Vaandrager~
\citeMac{DeNicola90}.
The $\mu$-calculus itself is a logic which subsumes CTL, CTL*, LTL, as well as more exotic variations thereof. Its links to automata theory are well established~\cite{wilke2001} and can be used in the implementation of verification tools.  This makes the $\mu$-calculus well suited for foundational research on verification, even though logics with more intuitive syntax may appeal to practitioners.
$\UHML$ over traces
is similar to the
%comes close to
linear-time $\mu$-calculus.
The main difference is that in the linear-time $\mu$-calculus, which is usually interpreted over infinite traces, it is common to have only one successor-modality: the difference between $\hNec{\act}$ and $\hSuf{\act}$ only manifests itself over finite traces. Here we have chosen to keep the two modalities, to enable the syntactic comparison between branching-time and linear-time monitorability. From an implementation point of view, $\UHML$ formulae, like those in the linear time $\mu$-calculus, can be represented by weak automata~\cite{lange05weak}, which benefit from lower-complexity decision procedures than the more general parity automata, which are necessary to capture the expressiveness of the $\mu$-calculus in a branching-time setting.
Note, however, that, as shown in~\citeMac{MarkeyS2006}, the $\mu$-calculus model-checking problem over paths of the form $s^\omega$ is, surprisingly, as hard as the general model-checking problem for that logic.

In the context of RV, many-valued logics
%three or more valued logics
\cite{barringer2004rule,d2005lola,drusinsky2003monitoring,BauerLeuckerSchallhart:10:LandC} have also emerged as a way to reconcile the infinitary semantics of, for example, LTL specifications with the finite observations of a monitor. Our concept of monitor can itself also be understood as a logic with three-valued semantics, consisting of accepted traces, rejected traces and traces on which the monitor remains indecisive. Conversely, these many-valued logics can also be seen as describing monitor behaviour, albeit without an operational semantics as in our case.
 Our parallel monitors are reminiscent of alternating automata. The use of alternating automata for RV is not new:
% Finkbeiner and Sipma~
 \citeMac{finkbeiner2004checking} propose this for the verification of their finite-trace semantics for LTL. The main difference in their approach is that their semantics is not suffix closed: for whether ``infinitely often $a$'' holds in a finite trace according to their semantics will depend on whether $a$ holds in the last position. In contrast, our verdicts are irrevocable, so a sound monitor for ``infinitely often $a$'' in our setting
 will never reach a verdict.

\paragraph{Related Work on Monitorability}

The question of exactly which specifications can be verified at runtime is very natural in the RV context. It is perhaps surprising that there is no consensus on what exactly it means for a specification to be monitorable.

The class $\Pi^0_1$ of the arithmetic hierarchy --- the class of co-recursively enumerable safety properties --- was proposed as the set of monitorable properties by
%Viswanathan and Kim~
\citeMac{viswanathan2004foundations}. It seems that our notion of partial monitorability matches well with this classical definition. In this sense, partial monitorability could be seen as an operational account of Viswanathan and Kim's monitorability.
On the other hand,
%Pnueli and Zaks~
\citeMac{PnueliZaks:06:FM} and
%Bauer, Leuker and Schallhart~
\citeMac{bauer2011runtime}
%on the other hand,
 propose a definition of monitorability that includes more properties: roughly, they call a property monitorable if every prefix has a finite continuation of which either all infinite continuations are in the property, or none is.
This means that a monitor, although it does not necessarily ever reach a verdict, can never give up hope of reaching a verdict.
%Diekert and Leuckert
%%, and then Diekert, Muscholl and Walukiewicz,
%bring  a topological perspective to the study of monitorability in \cite{Diekert2014}.
%
Our definitions of monitorability appear to be stronger. For example, specifications such as ``never $\mathtt{error}$ and eventually $\mathtt{success}$'' is monitorable according to
%Pnueli and Zaks~%
%\cite{pnueli2006psl}
\citeMac{PnueliZaks:06:FM} and
%Bauer, Leuker and Schallhart~
\citeMac{bauer2011runtime} but not according to our notions of monitorability
% nor even our notion of
and partial monitorability.

Diekert and Leucker have studied monitorability in a topological setting  in~\citeMac{Diekert2014}, where they show that all $\omega$-regular languages that  are deterministic and co-deterministic are monitorable. Using their topological framework, they also establish that some deterministic liveness properties, such as  ``infinitely many a's'', cannot be written as a countable union of monitorable languages.
%In~
\citeMac{DiekertMW15}
%, Dieker, Muscholl and Walukiewicz
discuss monitor constructions for deterministic $\omega$-regular languages. They isolate a collection of  deterministic $\omega$-regular languages that properly includes all the languages that are deterministic and codeterministic, and for which one can construct accepting monitors.
These classical definitions of monitorability are independent of how a monitor might be implemented. Conversely, implementations of LTL monitors~\cite{giannakopoulou2001runtime,havelund2002synthesizing} do not seem to refer to the concept of monitorability at all.
%The novelty of our operational approach, i
In line with previous work~\cite{FraAI:17:FMSD}, our operational approach bridges this gap by defining \textit{what} can be monitored explicitly in terms of \textit{how} specifications are monitored.
%
% We leave a more detailed taxonomy of different notions of monitorability for future work.

%% file: appendix.tex
% !TEX root = main.tex

% In this appendix, we give the full proofs for the results that appear in the main body of this paper.

\section{Properties of Monitors}
We first present the omitted proof from \Cref{sec:monitors}

\subsection{Regular Monitor Properties}
\label{sec:monitors-appendix}
\input{monitors-app.tex}

\subsection{Reactive Parallel Monitors}
\label{sec:parallel-properties-appendix}
\input{parallel-properties-appendix.tex}

\subsection{An Equivalence of Two Monitoring Systems}
\label{sec:appendix-equivalent-rules}
\input{appendix_equivalent_rules.tex}

\subsection{Monitor Transformations}
\label{sec:transformations-appendix}

\input{transformations-appendix.tex}

\section{Linear-Time Monitorability}
\label{sec:monitorability-appendix}
We now present the omitted proofs of \Cref{sec:ltmu-monitorability}.

\subsection{Complete Monitoring}
\label{sec:complete-mon-app}
\input{complete-monitorability-appendix.tex}

\subsection{Tight Complete Monitoring}
\label{sec:complete-mon-app-tight}
\input{tight-complete-monitorability-appendix.tex}

\subsection{Partially-Complete Monitoring}
\label{sec:partial-complete-mon-app}
\input{partial-complete-monitorability-appendix.tex}

\subsection{Tight Partially-Complete Monitoring}
\label{sec:partial-complete-tight-mon-app}
\input{tight-partial-complete-monitorability-appendix.tex}

\section{Monitorability Across Semantics}
\label{sec:branching-appendix}
We now present the omitted proofs of \Cref{sec:branchingtime}.
%just one subsection for now...
\input{branching-appendix.tex}

%% file: monitors-app.tex
% !TEX root = main.tex

\Cref{def:mon-state} attempts to characterise the set of reachable states for a monitor.
\Cref{def:mon-measure} maps every monitor to a finite positive integer, which can be used to put an upperbound on the size of our state-space approximation of \Cref{def:mon-state}; see \Cref{lem:state-space-charact-bounded}.

\begin{defn}[Monitor State Space Characterisation] \label{def:mon-state} \quad
  \begin{align*}
    \states{\mV} & \deftxt
    \begin{cases}
      \sset{\mV} & \text{ if }\mV = \vV \text{ or }\mV = x\\
      \sset{\mV} \cup \states{\mVV} & \text{ if } \mV = \prf{\acta}{\mVV}\\
      \sset{\mV} \cup  \statesSkip{\mV_1} \cup \statesSkip{\mV_2} & \text{ if } \mV = \ch{\mV_1}{\mV_2}\\
      \states{\mVV}\subS{\rec{x}{\mVV}}{x} & \text{ if } \mV = \rec{x}{\mVV}
    \end{cases}\\
    \statesSkip{\mV} & \deftxt
    \begin{cases}
      \sset{\mV} & \text{ if }\mV = \vV \text{ or }\mV = x\\
       \states{\mVV} & \text{ if } \mV = \prf{\acta}{\mVV}\\
        \statesSkip{\mV_1} \cup \statesSkip{\mV_2} & \text{ if } \mV = \ch{\mV_1}{\mV_2}\\
      \states{\mVV}\subS{\rec{x}{\mVV}}{x} & \text{ if } \mV = \rec{x}{\mVV}
    \end{cases}\\
  \end{align*}
\end{defn}

\begin{defn}[Skip Reachability] \label{def:mon-skip-reach}
  % \quad
\begin{math}
  \reachSkip{\mV} \deftxt \setof{\reach{\mVV}}{ \mV \traS{\actu} \mVV}
\end{math}
\end{defn}

\begin{defn}[Monitor Measure] \label{def:mon-measure} \quad
  \begin{align*}
    \size{\mV} & \deftxt
    \begin{cases}
      1 & \text{ if }\mV = \vV \text{ or }\mV = x\\
      1 + \size{\mVV} & \text{ if } \mV = \prf{\acta}{\mVV}\\
      1 + \size{\mV_1} + \size{\mV_2} & \text{ if } \mV = \ch{\mV_1}{\mV_2}\\
      \size{\mVV} & \text{ if } \mV = \rec{x}{\mVV}
    \end{cases}
  \end{align*}
\end{defn}

\begin{lem} \label{lem:state-space-charact-bounded}
  \begin{math}
    \forall \mV\in\REMon \cdot |\states{\mV}| \leq \size{\mV}
  \end{math}
\end{lem}
\begin{proof}
  By structural induction on \mV.
\end{proof}

\Cref{lem:state-space-charact-reach} shows that the state-space approximation of \Cref{def:mon-state} characterises precisely the actual state-space of a monitor.  It however relies on a few technical lemmata.

% A few technical lemmata required by \Cref{lem:state-space-charact-reach}.

\begin{lem} \label{lem:open-vs-closed-mon-single-reduction}
   For all (possibly open) $\mV \in \REMon$:
  \begin{enumerate}
    \item $\mV \subS{\mVV}{x} \traS{\actu} \mV'\ $ implies $\ (\exists \mV'' \cdot \mV \traS{\actu} \mV'' \text{ and } \mV''\subS{\mVV}{x} = \mV')$ or  $(x \text{ is a summand of } \mV \text { and } \mVV \traS{\actu} \mV')$.
    \item $\mV \traS{\actu} \mV'\ $ implies $\ (\exists \mV'' \cdot \mV \subS{\mVV}{x} \traS{\actu} \mV'' \text{ and } \mV'\subS{\mVV}{x} = \mV'')$
  \end{enumerate}
\end{lem}
\begin{proof}
  Both clauses are proved by structural induction on \mV. The main cases for the first clause are:
  \begin{description}
    \item[Case \mV = y:] Since  $y \subS{\mVV}{x} \traS{\actu} \mV'$, it must be the case that $y=x$ (hence a summand of \mV) and $y \subS{\mVV}{x} = \mVV$ from which we obtain $\mVV \traS{\actu} \mV'$ as required.
    \item[Case \mV = \ch{\mV_1}{\mV_2}:] We have $\mV \subS{\mVV}{x} = \ch{(\mV_1 \subS{\mVV}{x})}{(\mV_2 \subS{\mVV}{x})}$, meaning that the transition was inferred using \rtit{eSel}.
    Without loss of generality, assume that $\mV_1 \subS{\mVV}{x} \traS{\actu} \mV'$.
    By the I.H. we obtain the following subcases:
    \begin{itemize}
      \item Either $(\exists \mV'' \cdot \mV_1 \traS{\actu} \mV'' \text{ and } \mV''\subS{\mVV}{x} = \mV')$.
      By \rtit{eSel} we deduce $\ch{\mV_1}{\mV_2} \traS{\actu} \mV''$ as required.
      \item Or $x$ is a summand of $\mV_1$ and $\mVV \traS{\actu} \mV'$, which is precisely the required result since $x$ is then also a summand of $\ch{\mV_1}{\mV_2}$.
    \end{itemize}
    \item[Case \mV = \rec{y}{\mV_1}:]  By \rtit{mRec}, we have $\mV \subS{\mVV}{x} = \rec{y}{(\mV_1\subS{\mVV}{x})} \traS{\actt} \mV_1\subS{\mVV}{x}\subS{\rec{y}{(\mV_1\subS{\mVV}{x})}}{y}$.
    Again, by \rtit{mRec}, the required transition is
    $\rec{y}{\mV_1} \traS{\actt} \mV_1\subS{\rec{y}{\mV_1}}{y}$, since we can infer the equality
    $\mV_1\subS{\rec{y}{\mV_1}}{y}\subS{\mVV}{x} = \mV_1\subS{\mVV}{x}\subS{\rec{y}{(\mV_1\subS{\mVV}{x})}}{y}$.
  \end{description}
  For the second clause, the main cases are:
  \begin{description}
    \item[Case \mV= y:] The implication holds trivially since variables do not transition.
    \item[Case \mV = \ch{\mV_1}{\mV_2}:] By \rtit{mSel}, we have either $\mV_1 \traS{\actu} \mV'$ or $\mV_2 \traS{\actu} \mV'$.
    Without loss of generality, pick $\mV_1 \traS{\actu} \mV'$.
    By the I.H. we have $\mV_1\subS{\mVV}{x} \traS{\actu} \mV''$ for some $\mV''$ where
    $\mV'\subS{\mVV}{x} = \mV''$.
    Again, using \rtit{mSel} we deduce $\mV\subS{\mVV}{x} = \ch{(\mV_1 \subS{\mVV}{x})}{(\mV_2 \subS{\mVV}{x})} \traS{\actu} \mV''$ as required.
    \item[Case \mV = \rec{y}{\mV_1}:]
    By \rtit{mRec}, we have $\rec{y}{\mV_1} \traS{\actt} \mV_1\subS{\rec{y}{\mV_1}}{y}$.
     The required transition for $\mV\subS{\mVV}{x} = \rec{y}{(\mV_1\subS{\mVV}{x})}$ is again obtained by the rule \rtit{mRec} in the form of
      $\rec{y}{(\mV_1\subS{\mVV}{x})} \traS{\actt} (\mV_1\subS{\mVV}{x}) \subS{\rec{y}{(\mV_1\subS{\mVV}{x})}}{y}$, since
     $(\mV_1\subS{\mVV}{x}) \subS{\rec{y}{(\mV_1\subS{\mVV}{x})}}{y} = (\mV_1\subS{\rec{y}{\mV_1}}{y})\subS{\mVV}{x}$. \qedhere
  \end{description}
\end{proof}

\begin{lem}
  \label{lem:open-vs-closed-mon-multi-reduction}
 %  For all (possibly open) $\mV \in \REMon$ and $k\geq 0$:
 % \begin{enumerate}
 %   \item $\mV \subS{\mVV}{x} (\reduc)^k \mV'$ implies $(\exists \mV'' \cdot \mV (\reduc)^k \mV'' \text{ and } \mV''\subS{\mVV}{x} = \mV')$ or $\mVV (\reduc)^k \mV'$
 %   \item $\mV (\reduc)^k \mV'$ implies $\exists \mV'' \cdot (\mV \subS{\mVV}{x} (\reduc)^k \mV'' \text{ and } \mV'\subS{\mVV}{x} = \mV'')$
 % \end{enumerate}
 For all (possibly open) $\mV \in \REMon$:
\begin{enumerate}
  \item $\mV \subS{\mVV}{x} \etraS{\etV} \mV'$ implies that
  \begin{itemize}
    \item Either $\exists \mV'' \cdot \mV \etraS{\etV} \mV'' \text{ and } \mV''\subS{\mVV}{x} = \mV'$
    \item Or $\exists \etV_1,\etV_2,\mV'' \cdot \etV=\etV_{1}\etV_{2}$   and $\etV_2 \neq \epsilon$  and  $\mV\etraS{\etV}\mV''$ where $x$  is a summand of \mV  and $\mVV\etraS{\etV_2} \mV'$.
  \end{itemize}
  \item $\mV \etraS{\etV} \mV'$ implies $\exists \mV'' \cdot (\mV \subS{\mVV}{x} \etraS{\etV} \mV'' \text{ and } \mV'\subS{\mVV}{x} = \mV'')$
\end{enumerate}
\end{lem}
\begin{proof}
  The proof for the second clause is by a straightforward induction on the structure of \etV, where the inductive step relies  \Cref{lem:open-vs-closed-mon-single-reduction}(2).

  The proof for the first clause is also by induction on \etV, but it is slightly more involved.
  \begin{description}
    \item[Case $\etV=\epsilon$:]
    Immediate, since $\mV'=\mV\subS{\mVV}{x}$ and  $\mV \etraS{\epsilon} \mV$.
    \item[Case $\etV=\actu\etVV$:]
    We thus have
    \begin{equation}
      \label{eq:open-vs-closed-mon-multi-reduction}
      \mV \subS{\mVV}{x} \traS{\actu} \mV'' \etraS{\etVV} \mV' \text{ for some intermediary monitor } \mV''.
    \end{equation}
    By $\mV \subS{\mVV}{x} \traS{\actu} \mV''$ and \Cref{lem:open-vs-closed-mon-single-reduction}(1) we have to consider either of two cases:
    \begin{enumerate}
      \item Either there exists some $\mV'''$ such that $\mV \traS{\actu} \mV''' $ and $\mV'''\subS{\mVV}{x} = \mV''$.
      By $\mV''=\mV'''\subS{\mVV}{x} \etraS{\etVV} \mV'$ from \Cref{eq:open-vs-closed-mon-multi-reduction} and the I.H. we have two possibilities:
      \begin{enumerate}
        \item Either there exists some $\mV''''$ such that $\mV''' \etraS{\etVV} \mV'''' $ where $\mV''''\subS{\mVV}{x} = \mV'$.  By prefixing this with $\mV \traS{\actu} \mV''' $  gives us $\mV \etraS{\etV}\mV''''$ as required.
        \item Or $\etVV=\etVV_1 \etVV_2$ for some $\etVV_1 $ and $\etVV_2\neq \epsilon$ where $\mV''' \etraS{\etVV_1}\mV''''$ for some $\mV''''$ with a summand $x$ and $\mVV \etraS{\etVV_2} \mV'$.
        Again, by contacting $\mV \traS{\actu} \mV''' $ with  $\mV''' \etraS{\etVV_1}\mV''''$  as $\mV' \etraS{\actu\etVV_1}\mV''''$ gives us the result required.
      \end{enumerate}
      \item Or $x$ is a summand of $\mV$ and  $\mVV \traS{\actu} \mV''$.  Using $\mV'' \etraS{\etVV} \mV'$ of \Cref{eq:open-vs-closed-mon-multi-reduction}, this would satisfy the second clause with $\etV_1=\epsilon$ and $\etV_2=\actu\etVV$  since $\mV \etraS{\epsilon} \mV$. \qedhere
    \end{enumerate}
  \end{description}
\end{proof}

We prove the required property for \emph{closed} monitors.  Note that closed monitors are closed \wrt transitions.

\begin{lem}  \label{lem:state-space-charact-reach}
  \begin{math}
    \forall \mV \in \REMon \cdot \fv(\mV){=}\emptyset \text{ implies } \reach{\mV} = \states{\mV} \text{ and } \reachSkip{\mV} = \statesSkip{\mV}
  \end{math}
\end{lem}
\begin{proof}
  By structural induction on \mV:
  \begin{description}
    \item[Case \mV = \vV:]
      It follows from \Cref{lem:ver-persistence}.
    \item[Case \mV = x:]
      % Since $x \reducNot$, \reach{x} = \sset{x} = \states{x}.
      Immediate since $\fv(\mV)\neq\emptyset$.
    \item[Case \mV = \prf{\acta}{\mVV}:]
      By \Cref{def:mon-reach} and \rtit{mAct} from \Cref{fig:monit-instr}, $\reach{\prf{\acta}{\mVV}} = \sset{\prf{\acta}{\mVV}} {\cup} \reach{\mVV}$.
      By the I.H. we have $\reach{\mVV} = \states{\mVV}$, and thus
      $\sset{\prf{\acta}{\mVV}} {\cup}\, \reach{\mVV} = \sset{\prf{\acta}{\mVV}} {\cup}\, \states{\mVV} = \states{\prf{\acta}{\mVV}}$ (by \Cref{def:mon-state}).
      For the second property, we know that $\reachSkip{\prf{\acta}{\mVV}} = \reach{\mVV}$ by \Cref{def:mon-skip-reach}, and that
      $\statesSkip{\prf{\acta}{\mVV}} = \states{\mVV}$ by \Cref{def:mon-state}.
      By the I.H. we then obtain $\reach{\mVV} = \states{\mVV}$ as required.
    \item[Case \mV = \ch{\mV_1}{\mV_2}:] For the first property,  we
    deduce that
    \begin{math}
      \reach{\ch{\mV_1}{\mV_2}} = \sset{\ch{\mV_1}{\mV_2}} \cup \reachSkip{\mV_1} \cup  \reachSkip{\mV_2}
    \end{math}
     by \Cref{def:mon-reach} and \rtit{mSel} from \Cref{fig:monit-instr}.
    By \Cref{def:mon-state}, we akso know that $\states{\ch{\mV_1}{\mV_2}} = \statesSkip{\mV_1} \cup \statesSkip{\mV_2}$.
    The rest of the proof uses the I.H. to obtain the required result, in analogous fashion to the previous case.
    For the second property we need to show that $\reachSkip{\ch{\mV_1}{\mV_2}} = \statesSkip{\ch{\mV_1}{\mV_2}}$.
    By \rtit{mSel} from \Cref{fig:monit-instr} we know that $\reachSkip{\ch{\mV_1}{\mV_2}} = \reachSkip{\mV_1} \cup \reachSkip{\mV_1}$.
    By the I.H., we also know that for $i \in 1..2$ $\reachSkip{\mV_i} = \statesSkip{\mV_i}$.
    The required result thus follows since, by  \Cref{def:mon-state} we have $\statesSkip{\ch{\mV_1}{\mV_2}} = \statesSkip{\mV_1} \cup \statesSkip{\mV_1}$.
    \item[Case \mV = \rec{x}{\mVV}:]
    By \Cref{def:mon-reach} and \rtit{mRec} from \Cref{fig:monit-instr} we have
    $\reach{\rec{x}{\mVV}} = \sset{\rec{x}{\mVV}} \cup \reach{\mVV\subS{\rec{x}{\mVV}}{x}}$.
    By the I.H. we also know that $\reach{\mVV} = \states{\mVV}$ which means that $(\reach{\mVV}\subS{\rec{x}{\mVV}}{x}) = (\states{\mVV}\subS{\rec{x}{\mVV}}{x})$ .
    From \Cref{lem:open-vs-closed-mon-multi-reduction} we deduce
    $\reach{\mVV\subS{\rec{x}{\mVV}}{x}} = \reach{\mVV}\subS{\rec{x}{\mVV}}{x}$ thus
    $\sset{\rec{x}{\mVV}} \cup \reach{\mVV\subS{\rec{x}{\mVV}}{x}} = \sset{\rec{x}{\mVV}} \cup \states{\mVV}\subS{\rec{x}{\mVV}}{x} = \states{\rec{x}{\mVV}}$ by \Cref{def:mon-state}, as required.
    For the second property, we know, by \rtit{mRec} and \Cref{def:mon-skip-reach}  that $\reachSkip{\rec{x}{\mVV}} = \reach{\mVV\subS{\rec{x}{\mVV}}{x}}$.
    By \Cref{def:mon-state}, we also know that  $\statesSkip{\rec{x}{\mVV}} = \states{\mVV}\subS{\rec{x}{\mVV}}{x}$.
    Recall that by \Cref{lem:open-vs-closed-mon-multi-reduction} we already established that $\reach{\mVV\subS{\rec{x}{\mVV}}{x}} = \reach{\mVV}\subS{\rec{x}{\mVV}}{x}$.
    Now by the I.H. we know that $\reach{\mVV} = \states{\mVV}$, from which we obtain $\reach{\mVV}\subS{\rec{x}{\mVV}}{x} = \states{\mVV}\subS{\rec{x}{\mVV}}{x}$, which is the result required.
    \qedhere
   \end{description}
\end{proof}

%% file: parallel-properties-appendix.tex
% !TEX root = main.tex

\begin{lem}[Monitor Combinators] \label{lem:mon-combinators} \  %If $\mV_1$ and $\mV_2$ are reactive, then:
\begin{enumerate}
  \item If $\mV_1$ and $\mV_2$ are reactive, then $\mV_1 \paralC \mV_2 \wtraS{\ftr} \no$ iff $\mV_1 \wtraS{\ftr} \no$ or  $\mV_2 \wtraS{\ftr} \no$.
  % \af{(if direction does not work)}
  \item %If $\mV_1$ and $\mV_2$ are reactive, then
  $\mV_1 \paralC \mV_2 \wtraS{\ftr} \yes$ iff $\mV_1 \wtraS{\ftr} \yes$ and  $\mV_2 \wtraS{\ftr} \yes$.
  \item If $\mV_1$ and $\mV_2$ are reactive, then $\mV_1 \paralD \mV_2 \wtraS{\ftr} \yes$ iff $\mV_1 \wtraS{\ftr} \yes$ or  $\mV_2 \wtraS{\ftr} \yes$.
  % \af{(if direction does not work)}
  \item %If $\mV_1$ and $\mV_2$ are reactive, then
  $\mV_1 \paralD \mV_2 \wtraS{\ftr} \no$ iff $\mV_1 \wtraS{\ftr} \no$ and  $\mV_2 \wtraS{\ftr} \no$.
  \item If $\ftr \neq \varepsilon$ or $m_1,m_2 \neq v$, then $\esel{\mV_1}{\mV_2} \wtraS{\ftr} \vV$ iff $\mV_1 \wtraS{\ftr} \vV$ or  $\mV_2 \wtraS{\ftr} \vV$.
\end{enumerate}
\end{lem}
\begin{proof}
	% By induction on the length of the transition sequence $\mV_1 \paralG \mV_2 \wtraS{\ftr} \vV$ in one direction, and  $\mV_i \wtraS{\ftr} \vV$ for $i\in\sset{1,2}$ in the other direction.
	%
	%
	We first prove the ``only if'' directions for the five statements.
	If $\mV_1 \paralC \mV_2 \wtraS{\ftr} \no$, then there is an explicit trace $\etr$ that agrees with $\ftr$ on the external actions, such that $\mV_1 \paralC \mV_2 \traS{\etr} \no$.
	We prove by induction on $\etr$ that
	% there are $s_1,s_2$ general traces that agree with $s$ on the external actions and have fewer silent actions (maybe later) --- there is a definition for this (appendix?), calling these traces contractions of s
	$\mV_1 \wtraS{\ftr} \no$ or $\mV_2 \wtraS{\ftr} \no$.

	The base case is $\etr = \varepsilon$, which is a contradiction, because it implies that $\mV_1 \paralC \mV_2 = \no$.
	For the inductive step, let $\etr = \mu \etrr$ and $\mV_1 \paralC \mV_2 \traS{\mu} \mV \traS{\etrr} \no$.
	We distinguish the following cases:
	\begin{description}
		\item[Case  $\mu \in \Act$:] Then, $\mV = \mV_1' \paralC \mV_2'$, where $\mV_1 \traS{\mu} \mV_1'$ and $\mV_2 \traS{\mu} \mV_2'$, and rule \rtit{mPar} was used, so
		%  and $\mu \in \Act$) then,
		the
		claim follows %from
		%  induction is complete
		by
		the inductive hypothesis applied to $\mV \traS{\etrr}\no$.
		\item[Case  $\mu = \tau$ and $\mV = \mV_1' \paralC \mV_2'$, where $\mV_1 \traS{\mu} \mV_1'$ or $\mV_2 \traS{\mu} \mV_2'$]
		%  , but not both]
		(that is,
		%  one of the
		rule \rtit{mTauL} or rule \rtit{mTauR} was used). Then, again,
		the
		%  induction is complete by
		claim follows by
		the inductive hypothesis.
		\item[Otherwise:] The only possibilities are that $\mu = \tau$ and rule \rtit{mVrC1} or \rtit{mVrC2} were used, so  respectively,
		$\mV = \mV_1$ or $\mV = \mV_2$, so
		$\mV_1 \traS{\etrr} \no$ or $\mV_2 \traS{\etrr} \no$;
		% one of $\mV_1$ or $\mV_2$ rejects the trace,
		or $\mV = \no$ and either $\mV_1 = \no$ or $\mV_2 = \no$. In all cases, we have that $\mV_1 \wtraS{\ftr} \no$ or $\mV_1 \wtraS{\ftr} \no$.
		% one of them has rejected the trace.
		% \qedhere
	\end{description}
	The ``only if'' directions for the other cases are proven in a similar way.

	We now prove the ``if'' directions of the statements of the lemma, and specifically we prove that if $\mV_1 \wtraS{\ftr} \no$ and $\mV_2$ is reactive, then $\mV_1 \paralC \mV_2 \wtraS{\ftr} \no$; the remaining cases are analogous.
	Assume that $\mV_1 \wtraS{\ftr} \no$.
	Then, there is an explicit trace $\etr$ that agrees with $\ftr$ on the external actions, such that $\mV_1 \traS{\etr} \no$ --- fix $\etr$ to be the shortest such explicit trace for $\mV_1$ and $\ftr$.
	We proceed by induction on the length of $\etr$.
	\begin{description}
		\item[Case  $\etr = \varepsilon$:] Then $\mV_1 = \no$ and $\no \paralC \mV_2 \traS{\tau} \no$.
		\item[Case  $\etr = \tau \etrr$:] Then, $\mV_1 \traS{\tau} \mV_1' \traS{\etrr} \no$. By rule \rtit{mTauL},
		$\mV_1 \paralC \mV_2 \traS{\tau} \mV_1' \paralC \mV_2$, and by the inductive hypothesis, $\mV_1' \paralC \mV_2 \traS{\ftr} \no$.
		\item[Case  $\etr = \act \etrr$:] Then, $\mV_1 \traS{\act} \mV_1' \traS{\ftr'} \no$. Since $\mV_2$ is reactive, there is a reactive $\mV_2'$, such that $\mV_2 \wtraS{\act} \mV_2'$, and by successive applications of rule \rtit{mTauR} and then rule \rtit{mPar}, and the inductive hypothesis,
		$\mV_1 \paralC \mV_2 \wtraS{\act} \mV_1' \paralC \mV_2' \wtraS{\ftr} \no$.
		% , and by the inductive hypothesis, $\mV_1' \paralC \mV_2 \traS{\ftr} \no$.
		% \qedhere
	\end{description}

	Statements (2), (3), and (4) are proven similarly. For (5), we observe that if $\mV_1 + \mV_2 \wtra{\ftr} v$, then there is an explicit trace $\etr$ that agrees with $\ftr$ on the external actions, such that $\mV_1 + \mV_2 \traS{\etr} v$, and since $\mV_1+\mV_2 \neq v$, $\etr = \mu\etrr$; therefore, there is some $\mV$ such that $\mV_1 + \mV_2 \traS{\mu} \mV \traS{\etrr} v$. According to the monitor rules, for some  $i \in \{1,2\}$, $\mV_i \traS{\mu} \mV$, and thus, $\mV_i \wtraS{\ftr} v$.
	For the other direction, if, say, $\mV_1 \wtraS{\ftr} v$, then there is an explicit trace $\etr$, such that $\mV_1 \traS{\etr} v$; since $m_1 \neq v$ or $\ftr \neq \varepsilon$, we see that $\etr = \mu \etrr$, and the remaining argument is as above.
\end{proof}

\begin{rem}
	\Cref{lem:mon-combinators} indirectly describes three different kinds of non-determinism for reactive parallel monitors.
	Operator $\paralD$ can be thought of as an existential monitor choice, as
	$\mV_1 \paralD \mV_2$ will accept (\resp reject) iff either (\resp both) of its components accepts (\resp reject). Dually, $\paralC$ can be thought of as a universal choice. The operator $+$ is a different choice that favours neither acceptance nor rejection, but generates either verdict, as long as one of its component monitors can reach it. \qedd
\end{rem}

\begin{rem}
	\Cref{ex:reactiveness} indicates that the assumption that $\mV_1$ and $\mV_2$ are reactive are needed in statements (1) and (3) of the above lemma.
	However,
	as the following lemma
	%\ref{lem:mon-combinators2}
	demonstrates, that assumption
	%that the monitors are reactive in Lemma \ref{lem:mon-combinators}
	is only necessary to prove one (the ``if'') direction of
	%the lemma for cases
	statements
	(1) and (3) in \Cref{lem:mon-combinators}.
\end{rem}

\begin{lem}\label{lem:mon-combinators2} \ %If $\mV_1$ and $\mV_2$ are reactive, then:
	\begin{enumerate}
		\item If $\mV_1 \paralC \mV_2 \traS{\etr} \no$, then $\mV_1 \traS{\etrr} \no$ or  $\mV_2 \traS{\etrr} \no$, where $\etrr$ agrees with $\etr$ on the external actions and is strictly shorter than $\etr$.
		% \af{(if direction does not work)}
		\item If $\mV_1 \paralC \mV_2 \traS{\etr} \yes$, then $\mV_1 \traS{\etrr} \yes$ and  $\mV_2 \traS{\etrr'} \yes$, where $\etrr,\etrr'$ agree with $\etr$ on the external actions and are strictly shorter than $\etr$.
		\item If $\mV_1 \paralD \mV_2 \traS{\etr} \yes$, then $\mV_1 \traS{\etrr} \yes$ or  $\mV_2 \traS{\etrr} \yes$, where $\etrr$ agrees with $\etr$ on the external actions and is strictly shorter than $\etr$.
		% \af{(if direction does not work)}
		\item If $\mV_1 \paralD \mV_2 \traS{\etr} \no$, then $\mV_1 \traS{\etrr} \no$ and  $\mV_2 \traS{\etrr'} \no$, where $\etrr,\etrr'$ agree with $\etr$ on the external actions and are strictly shorter than $\etr$.
	\end{enumerate}
\end{lem}

\begin{proof}
	We can use the same induction as for the ``only if'' direction of the proof of Lemma \ref{lem:mon-combinators}, noticing that the explicit traces for the submonitors are shorter than $\etr$.
\end{proof}

%\ac{NEW LEMMA!!}

In the technical developments that follow, we will require the following version of statements (1) and (3) of \Cref{lem:mon-combinators,lem:mon-combinators2}.

\begin{lem}\label{lem:mon-combinators-follow-along} %If $\mV_1$ and $\mV_2$ are reactive, then:
	\
	\begin{enumerate}
		\item If $\ftr$ is minimal such that $\mV \paralC \mVV \wtraS{\ftr} \no$, then
		there are $q_1,q_2$, such that $\mV \wtraS{\ftr}q_1$ and $\mVV \wtraS{\ftr} q_2$, and $q_1 = \no$ or $q_2 = \no$.
		%and $\mV \traS{\etrr} \no$, where $\etrr$ agrees with $\etr$ on the external actions, then $\exists q.\mVV \wtraS{\ftr} q$, where $\ftr$ agrees with $\etr$ on the external actions.
		\item If $\ftr$ is minimal such that $\mV \paralD \mVV \wtraS{\ftr} \yes$, then
		there are $q_1,q_2$, such that $\mV \wtraS{\ftr}q_1$ and $\mVV \wtraS{\ftr} q_2$, and $q_1 = \yes$ or $q_2 = \yes$.
		%  \item If $\etr$ is minimal such that $\mV \paralD \mVV \traS{\etr} \yes$ and $\mV \traS{\etrr} \yes$, where $\etrr$ agrees with $\etr$ on the external actions, then $\exists q.\mVV \wtraS{\ftr} q$, where $\ftr$ agrees with $\etr$ on the external actions.
	\end{enumerate}
\end{lem}

\begin{proof}
	We prove the first part of the lemma, as the second one is similar.
	Since $\mV \paralC \mVV \wtraS{\ftr} \no$, there must be an external trace $\etr$ that agrees with $\ftr$ on the external actions, so that
	$\mV \paralC \mVV \traS{\etr} \no$.
	We can use a similar induction on $\etr$ as for the first direction of the proof for Lemma \ref{lem:mon-combinators}.
	The base case is $\etr = \tau^k$, which is
	immediate,
	%a contradiction,
	because, by Lemma \ref{lem:mon-combinators2},
	it implies that
	$\mV \wtraS{} \no$ or $\mVV \wtraS{} \no$.
	% , while $\mV \wtraS{} \mV$ and $\mVV \wtraS{} \mVV$.
	For the inductive step, let $\ftr \neq \varepsilon$, $\etr = \mu \etr'$, and $\mV \paralC \mVV \traS{\mu} q \traS{\etr'} \no$.
	We distinguish the following cases:
	\begin{description}
		\item[Case $\mu \in \Act$] (that is, rule \rtit{mPar} was used): Then,
		$q = \mV' \paralC \mVV'$, where $\mV \traS{\mu} \mV'$ and $\mVV \traS{\mu} \mVV'$, and
		%$\mu \in \Act$ and $\ftrr = \tau^k \mu \ftrr'$ for some $k \geq 0$, where $\ftrr'$ agrees with $\ftr'$ on the external actions.
		%Therefore,
		%$\mV \paralC \mVV \wtraS{\mu} \mV'' \paralC \mVV'$, where $\mV''$
		the induction is complete by the inductive hypothesis.
		\item[Case $\mu = \tau$ and $q = \mV' \paralC \mVV'$, where $\mV \traS{\mu} \mV'$ or $\mVV \traS{\mu} \mVV'$] (that is, one of the rules \rtit{mTauL} and \rtit{mTauR} was used): Then, again, the induction is complete by the inductive hypothesis.
		\item[Case $\mu = \tau$ and rule \rtit{mVrC1}
		was used:]
		Then, without loss of generality,
		$\mV \paralC \mVV \traS{\tau} \mV \traS{\etr'} \no$ and $\mVV = \yes$,
		in which case we have that $\mV \wtraS{\ftr}\no$ and $\mVV = \yes \traS{\ftr} \yes$.

		%respectively,
		%$q = \mV$ or $q = \mVV$, so
		%$\mV \traS{\etrr} \no$ or $\mVV \traS{\etrr} \no$;
		% one of $\mV$ or $\mVV$ rejects the trace,
		%or $q = \no$ and either $\mV = \no$ or $\mVV = \no$, but $\no \wtra{} \no$.
		\item[Case $\mu = \tau$ and  rule
		\rtit{mVrC2} was used:]
		Then, without loss of generality,
		$\mV \paralC \mVV \traS{\tau} \mV = \no$, which is a contradiction, because either $\ftr = \varepsilon$ or it is not minimal, which violates our assumptions.
		%
		% and, respectively,
		%$q = \mV$ or $q = \mVV$, so
		%$\mV \traS{\etrr} \no$ or $\mVV \traS{\etrr} \no$;
		% one of $\mV$ or $\mVV$ rejects the trace,
		%or $q = \no$ and either $\mV = \no$ or $\mVV = \no$, but $\no \wtra{} \no$.
		\qedhere
	\end{description}
\end{proof}

\begin{rem}
	We remark that although \Cref{lem:mon-combinators-follow-along} seems to be an immediate consequence of \Cref{lem:mon-combinators,lem:mon-combinators2}, this is not the case.
	Notice that \Cref{lem:mon-combinators-follow-along} asserts that \emph{both} components are able to follow the finite trace $\ftV$, and this is the reason the minimality of $\ftr$ is important. Otherwise, a counterexample would be $\no \paralC \acta.\yes \wtraS{\actb} \no$, as $\acta.\yes$ cannot transition with a $\actb$.
\end{rem}

%% file: appendix_equivalent_rules.tex
% !TEX root = main.tex

% \subsubsection{An Equivalence of Two Monitoring Systems}
% \label{sec:appendix-equivalent-rules}

% \ac{
% must rework on it
% }

To prove \Cref{prop:monitor2automaton}, we assume a different set of rules for parallel and regular monitors. These rules are the ones that result by replacing \rtit{mRec} with  the following rules:
\begin{mathpar}
	\inference[\rtit{mRecF}]{}{\rec{x}{\mV_x} \traSS{\tau} \mV_x}
	\and
	\inference[\rtit{mRecB}]{}{x \traSS{\tau} \mV_x}
\end{mathpar}
Here, we assume that for every monitor variable $x$, there is a unique monitor $p_x = \rec x \mV_x$ of $\mV$ such that $x$ appears in $\mV_x$. Therefore, the rules above are well-defined.
By substituting rule \rtit{mRec} by \rtit{mRecF} and \rtit{mRecB}, we get an equivalent monitoring system, where reactive monitors remain reactive. This is partly shown in \citeMac{determinization} for regular monitors and 
%in Appendix \ref{sec:appendix-equivalent-rules} 
here
we prove these claims in the context of parallel monitors. Thus, in the rest of this section we assume that the rules above are used.

%Here we prove the equivalence of the monitor LTS rules given in Section \ref{sec:monitors} with the one we use in Section \ref{subsec:transformations}.
%
We call System O the system of rules given in Table \ref{fig:monit-instr}, while
System N
is
the result of replacing rule \rtit{mRec} by the rules \rtit{mRecF} and \rtit{mRecB}. 
%given in Section \ref{subsec:transformations}. 
The reader is encouraged to read \citeMac{determinization} for a discussion of the two systems.

% We remind the reader that
For System N, we assume
the fixed mappings $x \mapsto p_x$ and $x \mapsto m_x$, such that $p_x = \rec x m_x$ and
$p_x$ is the only monitor of the form $\rec x m$ that we allow.
% fixed monitor that we call $m_g$ in this section, such that every other monitor
%
% We call $\mV$ a sum of $\mVV$ when: $\mV = \mVV$; or $\mV = \mVV'+\mVV''$ or $\mV = \mVV' \paralG \mVV'$ and $\mVV'$ or $\mVV''$ is a sum of $\mVV$.
% We call $\mV$ a simple sum of $\mVV$ when: $\mV = \mVV$; or $\mV = \mVV'+\mVV''$ and $\mVV'$ or $\mVV''$ is a sum of $\mVV$.
% We call $\mVV$ a component of $\mV$ when: $\mV = \mVV$; or $\mV = \mVV' \paralG \mVV''$ and $\mVV$ is a component of $\mVV'$ or $\mVV''$.
% We call a monitor simple if it is the only component of itself.
% \ac{
% it seems that the definitions of this paragraph (among other things) may be unnecessary; will keep for a bit and see
% }
%
Derivations $\traS{}$ and $\wtraS{}$ are defined as before,
% \ac{
% 	whatever that means (must check later to ensure definitions match)
% }
but the resulting relations are called $\traS{}_O$ and $\wtraS{}_O$, and $\traS{}_N$ and $\wtraS{}_N$, respectively for systems O and N.
We prove that systems O and N are equivalent. That is, for any monitor $\mV$, finite trace $\ftr$, and verdict $\vV$,
$$\mV \wtraS{\ftr}_O \vV \ \text{ if and only if }\ \mV \wtraS{\ftr}_N \vV.$$

% \ac{
% new trial
% }

\begin{lem}
	For every $x$, $p_x$ is simple.
\end{lem}

\begin{proof}
	Immediate from the definition.
\end{proof}

\begin{lem}
	If $x$ is a free variable in $p_y$, then $p_y$ is inside the scope of $p_x$.
\end{lem}

\begin{proof}
	An immediate observation.
\end{proof}

% Therefore, since we assume the mapping $x \mapsto p_x$ there is a well-ordering
There is an ordering
$\leq$ of monitor variables: $x \leq y$ iff $p_y$ is in the scope of $p_x$.
We note that if we only consider a finite number of variables (say, the ones that appear in a specific monitor), then $\leq$ is  a well-order.
% Using this ordering, we can define $\mV[\vec{p_x}/\vec{x}]$, where $\vce{x} = x$
This ordering allows us to define when a submonitor can substitute a variable for its corresponding recursive formula.
We recursively define when $\mVV$ is an unfolding of $r$: $\mVV = r$; or $\mVV = \mVV'[p_x/x]$, where $\mVV'$ is an unfolding of $r$ and $x$ is $\leq$-minimal among the variables that occur free in $\mVV'$.
% For general traces $s_1,s_2$, we recursively refine when $s_1$ is a contraction of $s_2$: $s_1 = s_2$; or $s_1 = \act s'_1$ and $s_2 = \act s'_2$ and $s'_1$ is a contraction of $s'_2$; or $s_1 = \tau s'_1$ and $s'_1$ is a contraction of $s_2$.

\begin{lem}\label{lem:unfolding-is-bisim}
If $\mVV$ is an unfolding of $r$,
	then
	\begin{enumerate}
		\item $\mVV = \vV$ if and only if $r = \vV$;
		\item
		for every $\act \in \Act$ and $\mVV'$,
			if $\mVV \traS{\act}_N \mVV'$, then there some $r'$, such that $r \wtraS{\act}_N r'$ and $\mVV'$ is an unfolding of $r'$;
		\item
		for every $\act \in \Act$ and $\mVV'$,
			if $r \traS{\act}_N r'$, then there some $\mVV'$, such that $\mVV \traS{\act}_N \mVV'$ and $\mVV'$ is an unfolding of $r'$;
		\item
		% for every $\act \in \Act$ and $\mVV'$,
			if $\mVV \wtraS{}_N \mVV'$, then there some $r'$, such that $r \wtraS{}_N r'$ and $\mVV'$ is an unfolding of $r'$;
		\item
		% for every $\tau \in \tau$ and $\mVV'$,
			if $r \wtraS{}_N r'$, then there some $\mVV'$, such that $\mVV \wtraS{}_N \mVV'$ and $\mVV'$ is an unfolding of $r'$.
	\end{enumerate}
\end{lem}

\begin{proof}
The proof is by induction on the number of substitutions required to construct $\mVV$ from $r$.
	The base case of $\mVV = r$ is trivial.
	To complete the inductive step, it suffices to prove
	the lemma for the case of
	$\mVV = r[p_x/x]$.
	This is done by induction on the structure of $r$.
	\begin{itemize}
	\item
	If $r$ is a verdict $v$ or a variable $y \neq x$, then $\mVV = r$ and we are done. %, as $v$ can only transition to itself.
	\item
	If $r = x$, then
	% either $\mVV = r = x$ or
	$\mVV = p_x$.
		% In the first case, we are done again, as $\mVV = r$; for the second case, s
		Since $x$ can transition exactly to $p_x$ with a $\tau$,
		$r \centernot{\traS{\act}} r'$
		and
		if $\mVV \wtraS{\act} \mVV'$, then $r= x \traS{\tau} p_x \traS{\act} \mVV'$.
		Similarly,
		if
		$r \wtraS{} r'$, then either $r = r'$, so we can have $\mVV = \mVV'$, or $r= x \traS{\tau} p_x \wtraS{} r'$, in which case $\mVV \wtraS{} r'$;
		if $\mVV \wtraS{} \mVV'$, then $r = x \traS{\tau} p_x = \mVV \wtraS{} \mVV'$.
		% $p_x \xRightarrow{t} v$ iff $x \xrightarrow{\tau} p_x \xRightarrow{t} v$. Furthermore, if $r \xrightarrow{s} v$, then $s = \tau s'$ and $s'$ is a contraction of $s$ and $r = x \xrightarrow{\tau} n \xrightarrow{s'} v$; therefore, this case is complete.
	\item
	If $r = \act.r''$, then $\mVV = \act.\mVV''$ and $\mVV'' = r''[p_x/x]$.
	The only possible transitions for $\mVV$ and $r$ are then $\mVV \traS{\act} \mVV''$ and $r \traS{\act} r''$, respectively.
	Therefore, $\mVV \traS{\act} \mVV'$ iff $\mVV'' = \mVV'$ and $r \traS{\act} r'$ iff $r'' = r'$;
	$\mVV \wtraS{} \mVV'$ iff $\mVV = \mVV'$ and $r \wtraS{} r'$ iff $r = r'$.
	% By the inductive hypothesis, $\mVV' \xRightarrow{t} v$ if and only if $r' \xRightarrow{t} v$, and therefore $\mVV \xRightarrow{t} v$ if and only if $t = \act s$ and $r \xRightarrow{t} v$; similarly for the second part of the lemma.
	\item
	If $r = \rec x r''$, then $\mVV = r[r/x] = r$, as $r = p_x$ and $x$ is bound in $r$; therefore, this case is complete.
	% and $\mVV' = r'[p_x/x]$.
	% The only possible transitions for $\mVV$ and $r$ are then $\mVV \xrightarrow{\tau} \mVV'$ and $r \xrightarrow{\tau} r'$, respectively.
	% Then, $\mVV \xRightarrow{t} v$ if and only if $\mVV' \xRightarrow{t} v$ if and only if $r' \xRightarrow{t} v$ if and only if $r \xRightarrow{t} v$.
	\item
	If $r = \rec y r''$ for $y \neq x$, then $\mVV = \rec y \mVV''$ and $\mVV'' = r''[p_x/x]$.
	The only possible strong transitions for $\mVV$ and $r$ are then $\mVV \traS{\tau} \mVV''$ and $r \traS{\tau} r''$, respectively.
	Then, $\mVV \centernot{\traS{\act}} \mVV'$ and $r'' \centernot{\traS{\act}} r'$.
	 % (by the inductive hypothesis) if and only if $r \wtraS{\act} r'$.
	% Similarly,
	If $\mVV \wtraS{} \mVV'$, then either $\mVV = \mVV'$ and we are done, or $\mVV'' \wtraS{} \mVV'$, so $r \traS{\tau} r'' \wtraS{} r'$; the case for $r \wtraS{} r'$ is symmetric.
	\item
	If $r = r_1 + r_2$, then $\mVV = n_1 + n_2$, where $\mVV_i = r_i[p_x/x]$ for $i \in \{1,2\}$.
	% $\mVV_1$ is an unfolding of $r_1$ and $\mVV_2$ is an unfolding of $r_2$.
	Then, if $\mVV \traS{\act} \mVV'$, then $\mVV_1 \traS{\act} \mVV'$ or $\mVV_2 \traS{\act} \mVV'$, so  $r_1 \wtraS{\act} r'$ or  $r_2 \wtraS{\act} r'$, implying  $r \wtraS{\act} r'$ and we are done by the inductive hypothesis;
	the remaining cases are similar.
	% similarly for $\mVV \wtraS{} \mVV'$ if and only if $r \wtraS{} r'$.
	\item
	If $r = r_1 \paralC r_2$, then  $\mVV = n_1 \paralC n_2$, where
	$\mVV_i = r_i[p_x/x]$, for $i \in \{1,2\}$.
	% $\mVV_1$ is an unfolding of $r_1$ and $\mVV_2$ is an unfolding of $r_2$.
	The remaining argument is similar to the above.
	\item
	If $r = r_1 \paralD r_2$, then  $\mVV = n_1 \paralD n_2$, where
	$\mVV_i = r_i[p_x/x]$, for $i \in \{1,2\}$.
	% $\mVV_1$ is an unfolding of $r_1$ and $\mVV_2$ is an unfolding of $r_2$.
	 The remaining argument is similar to the above.
	 \qedhere
	\end{itemize}
\end{proof}

As Lemma \ref{lem:unfolding-is-bisim} demonstrates, the unfolding relation is a kind of bisimulation for System N --- although we do not define such a notion here.
It is not hard to see that this relation is reflexive and transitive.

\begin{cor}\label{cor:unfolding_bisim_ind}
	If $\mVV$ is an unfolding of $r$,
	% and $r$ is a submonitor of $\mV$ (?),
	then for every finite trace $\ftr$,
	\begin{enumerate}
		\item
			if $\mVV \wtraS{\ftr}_N \mVV'$, then $r \wtraS{\ftr}_N r'$, where $\mVV'$ is an unfolding of $r'$.
		\item
		Furthermore,
			if $r \wtraS{\ftr}_N r'$, then $\mVV \wtraS{\ftr}_N \mVV'$, where $\mVV'$ is an unfolding of $r'$.
	\end{enumerate}
\end{cor}

\begin{proof}
By straightforward induction on \ftr.
\end{proof}

\begin{lem}\label{lem:NOclosed}
	For every closed monitor $\mVV$,
	% where $\mVV$ is closed and an unfolding of $r$,
	\begin{enumerate}
	\item
	% if
	$\mVV \traS{\act}_O \mVV'$ if and only if $\mVV \traS{\act}_N \mVV'$;
	% , where $\mVV'$ is an unfolding of $r'$;
	\item if
	$\mVV \traS{\tau}_O \mVV'$, then $\mVV \traS{\tau}_N \mVV''$, where $\mVV'$ is an unfolding of $\mVV''$;
	\item if
	$\mVV \traS{\tau}_N \mVV''$, then $\mVV \traS{\tau}_O \mVV'$, where $\mVV'$ is an unfolding of $\mVV''$.
	\end{enumerate}
\end{lem}

\begin{proof}
Immediate from the rules.
\end{proof}

\begin{cor}\label{cor:NOmatching-steps}
	For every  $\mVV, r$,
	where $\mVV$ is closed and an unfolding of $r$,
	\begin{enumerate}
	\item
	if
	$\mVV \traS{\act}_O \mVV'$, then $r \wtraS{\act}_N r'$, where $\mVV'$ is an unfolding of $r'$;
	\item
	if
	$r \traS{\act}_N r'$, then $\mVV \traS{\act}_O \mVV'$ where $\mVV'$ is an unfolding of $r'$;
	\item if
	$\mVV \traS{\tau}_O \mVV'$, then $r \wtraS{}_N r'$, where $\mVV'$ is an unfolding of $r'$;
	\item if
	$r \traS{\tau}_N r'$, then $\mVV \wtraS{}_O \mVV'$, where $\mVV'$ is an unfolding of $\mVV''$.
	\end{enumerate}
\end{cor}

\begin{proof}
A consequence of Lemmata \ref{lem:unfolding-is-bisim} and \ref{lem:NOclosed}.
\end{proof}

\begin{lem}
	For every closed monitor $\mV$, $\mV \wtraS{\ftr}_O v$ if and only if $\mV \wtraS{\ftr}_N v$.
\end{lem}

\begin{proof}
Specifically, we prove that the more general claim that if $\mV$ is an unfolding of $r$, then $\mV \wtraS{\ftr}_O v$ if and only if $r \wtraS{\ftr}_N v$.
We prove each direction separately.
If $\mV \wtraS{\ftr}_O v$, then there is an explicit trace $\etr$ that agrees with $\ftr$ on the external actions, such that
$\mV \traS{\etr}_O v$. Using induction on $\etr$, the first part of Lemma \ref{lem:unfolding-is-bisim}, and Corollary \ref{cor:NOmatching-steps}, it is not hard to prove that
$r \wtraS{\ftr} v$.
The other direction is similar.
\end{proof}

\begin{cor}
If $\mV$ is reactive for System O, then it is also reactive for System N.
\end{cor}

\begin{proof}
From Corollary \ref{cor:NOmatching-steps} and straightforward induction on the number of transitions to reach a monitor in $\reach{\mV}$.
\end{proof}

%% file: transformations-appendix.tex
To prove Proposition \ref{prop:monitor2automaton}, we use the following lemmata.

\begin{lem}{\label{lem:paral-preserves-reactivicity}} \
	\begin{itemize}
		\item
		If $\mV = \mVV \paralD \mVV'$ and $\mV \wtraS{\act}$, then either $\mV \wtraS{} \no$, or $\mVV \wtraS{\act}$ and $\mVV' \wtraS{\act}$.
		\item
		If $\mV = \mVV \paralC \mVV'$ and $\mV \wtraS{\act}$, then either $\mV \wtraS{} \yes$, or  $\mVV \wtraS{\act}$ and $\mVV' \wtraS{\act}$.
	\end{itemize}
\end{lem}
\begin{proof}
	We prove the first case, as the second one is similar.
	If $\mV \wtraS{\act} q$, then we can assume that $\mV (\traS{\tau})^k\traS{\act} q$. We prove the claim by induction on $k$.
	The case for $k = 0$ is immediate from rule \rtit{mPar}.
	If $\mV \traS{\tau} \mV' (\traS{\tau})^k\traS{\act} q$, then one of the following rules was used:
	\begin{description}
		\item[\rtit{mTauL} or \rtit{mTauR}] In this case, we are done by the inductive hypothesis.
		\item[\rtit{mVrE}] In this case,
		$\mVV = \mVV' = \stp \wtraS{\act} \stp$.
		\item[\rtit{mVrC1}] In this case, without loss of generality, $\mVV = \yes \wtraS{\act} \yes$ and $\mVV' = \mV' \wtraS{\act} q$.
		\item[\rtit{mVrC2}] In this case, without loss of generality, $\mV' = \no$ and therefore, $\mV \wtraS{} \no$.
		\qedhere
	\end{description}
\end{proof}

\begin{defn}
	We can define that $\mVV$ is an \emph{immediate submonitor} of $\mV$ recursively: $\mV$ is a immediate submonitor of $\mV$, and
	% , with the expected cases, except:
	the immediate submonitors of $\mV$ are also immediate submonitors of $\act.\mV$, $\rec x \mV$, $\mV + \mVV$, and $\mVV + \mV$.
\end{defn}

\begin{lem}{\label{lem:transition-to-reactive}}
	Let $\mVV$ be an immediate submonitor of a reactive monitor $\mV$.
	Then, for every $\mVV'$ for which $\mVV \reduc \mVV'$, $\mVV'$ is reactive.
\end{lem}

\begin{proof}
	% From Lemma \ref{lem:paral-preserves-reactivicity}, it suffices to prove that
	% if $\mVV$ is an immediate submonitor of some (not necessarily reactive) relevant submonitor $\mV'$ of $\mV$, then
	% $\mVV'$ appears as a paraller component of some $r \in \reach{\mV}$.
	% It suffices to prove that if $\mVV \traS{\act} \mVV'$, then $\mVV' \in \reach{\mV}$.
	The proof is by induction on $l(\mV) - l(\mVV)$ and the base case is $\mV = \mVV$, which is trivial.
	If $\mV' = \acta.\mVV$ is an immediate submonitor of $\mV$,
	% and $\mV'$ is an immediate submonitor of $\mV$,
	then
	$\mV' \traS{\acta} \mVV$, so
	% it is also reactive and
	% we are done
	by the inductive hypothesis, $\mVV$ is reactive and $\mVV' \in \reach{\mVV}$, so $\mVV'$ is also reactive.
	The case is similar for $\mV' = \rec x \mVV \traS{\tau} \mVV$: by the inductive hypothesis, $\mVV$ is reactive 
	and $\mVV' \in \reach{\mVV}$, so $\mVV'$ is also reactive.
	If $\mV' = \mVV + \mV''$ or $\mV' = \mV'' + \mVV$, and $\mV'$ is an immediate submonitor of $\mV$, then,
	$\mVV \reduc \mVV'$ implies that $\mV' \reduc \mVV'$
	% again,
	and
	we are done by the inductive hypothesis.
	% If $\mV' = \mV_1 \paralG \mV_2$ and (without loss of generality) $\mV_1 = \mVV + \mV_1'$
	% % is an immediate submonitor of $\mV_1$
	% and for every $\act$, if $\exists q.\mVV \wtraS{\act} r$, then
	% $\exists q.\mV_2 \wtraS{\act} r$, then,
	% $\mV' \reduc \mVV' \paralG r$ for some $r$; by the inductive hypothesis,
	% $\mVV' \paralG r$ is reactive, so by Lemma \ref{lem:paral-preserves-reactivicity}, so is $\mVV'$.
	%
	%
	% by the inductive hypothesis, $\mV_1 \reduc^k \mVV'$.
	% \ac{here it fails} -- or not!
\end{proof}

\begin{propbis}{\ref{prop:monitor2automaton}}
	For every reactive parallel monitor $\mV$, there is an alternating automaton that
	accepts $L_a(\mV)$
	% the set of finite traces that $\mV$ accepts
	and one that
	accepts $L_r(\mV)$.
	% the set of finite traces that $\mV$ rejects.
	% Furthermore,  has at most $|m|$ states.
\end{propbis}
\begin{proof}
	For completeness of exposition, we describe here, as well, the
	% recursive
	process of constructing an alternating automaton that accepts $L_a(\mV)$
	% exactly the finite traces that $\mV$ accepts
	--- the case for $L_r(\mV)$ is similar.
	% automaton that accepts exactly the traces that $\mV$ rejects is constructed in a similar way.
	% that may have free variables.
	% For every variable $x$ of the monitor, there is a designated state $q(x)$ of the resulting automaton.
	%
	The automaton for $\mV$ is $A_\mV = (Q,\Act,\mV,\delta,F)$, where
	\begin{itemize}
		\item $Q$ is the set of submonitors of $\mV$;
		% \item $q_0 = \mV$;
		\item $F = \{ \mVV \in Q \mid n \text{ accepts } \varepsilon \} $;
		\item Let for every $S \subseteq Q$, $\delta_0 (q,\act) (S) = 1$ iff $q \in F$;  $\delta$ is the closure of $\delta_0$ under the following conditions.
		% (we can view $\delta$ as the set of triples $(\mVV,\act,S)$ for which $\delta(\mVV,\act,S)=1$).
		For every $S \subseteq Q$:
		\begin{itemize}
			\item if $\mVV \in S$, then $\delta (\act.n,\act) (S) = 1$; % iff $n \in S$;
			\item if $\delta (\mVV,\act) (S) = 1$ or $\delta (\mVV',\act) (S) = 1$, then $\delta (\mVV + \mVV',\act) (S) = 1$;
			\item if $\delta (\mVV,\act) (S) = 1$ or $\delta (\mVV',\act) (S) = 1$,
			and  
			%              $\exists q.\mVV \wtraS{\act} q$ 
			%              and 
			%              $\exists q.\mVV' \wtraS{\act} q$,
			$\mVV \wtraS{\act}$ 
			and 
			$\mVV' \wtraS{\act}$,
			then $\delta (\mVV \paralD \mVV',\act) (S) = 1$;
			\item if $\delta (\mVV,\act) (S) = 1$ and $\delta (\mVV',\act) (S) = 1$, then $\delta (\mVV \paralC \mVV',\act) (S) = 1$;
			\item if $\delta (\mV_x,\act) (S) = 1$, then $\delta (p_x,\act) (S) = \delta (x,\act) (S) = 1$.
		\end{itemize}
	\end{itemize}
	% We observe that $\delta^*(\mVV,\ftr)(F)$ also satisfies these closure conditions --- by induction on $\ftr$:
	% if $\ftr = \varepsilon$, then $\delta^*(\mVV,\ftr)(F) = 1$ iff $\mVV \in F$
	
	We consider the parallel extension of a set $S$ of monitors in $Q$, which is the smallest set $S^+$ such that $S \subseteq S^+$, and if $\mVV, \mVV'\in S^+$, then $\mVV \paralC \mVV' \in S^+$, and if $\mVV \in S^+$, then $\mVV \paralD \mVV', \mVV' \paralD \mVV \in S^+$.
	% $$S^+ = S \cup \{ \mVV \paralC \mVV' \mid \mVV, \mVV'\in S  \} \cup \{ \mVV \paralD \mVV' \mid \mVV\in S \text{ or } \mVV'\in S
	% , \text{ and } \exists q.\mVV \wtraS{\act} r \text{ and } \exists q.\mVV' \wtraS{\act} r
	% \}. $$
	We prove the following claims:
	
	\paragraph{Claim 1:}
	If $\delta(\mVV,\act)(S) = 1$, then $\exists q \in (S \cup F)^+. \mVV \wtraS{\act} q$.
	By induction on the closure conditions for $\delta$.
	The base case is that $\delta_0(\mVV,\act) = 1$, which implies that $\mVV \wtraS{} \yes \traS{\act} \yes \in F$.
	The remaining cases are straightforward.

	\paragraph{Claim 2:} If $\delta^*(\mVV,\ftr)(F) = 1$, then $\delta^*(\mVV,\ftr \ftrr)(F) = 1$.
	We observe that $\delta^*(\mVV,\ftr)(S)$ is monotone with respect to $S$
	(\ie if $S \subseteq S'$ and $\delta^*(\mVV,\ftr)(S) = 1$, then  $\delta^*(\mVV,\ftr)(S') = 1$).
	%    We can continue with 
	The claim follows by a
	straightforward induction on $\ftr$.
	
	\paragraph{Claim 3:}
	If $\mVV_1 \wtraS{\ftr} \yes$ and $\mVV_1 \paralD \mVV_2$ is reactive, then $\mVV_1 \paralD \mVV_2 \wtraS{\ftr} \yes$. The proof is by induction on $\ftr$.
	%, and this suffices to complete the proof of the claim.
	If $\ftr = \varepsilon$, then $\mVV_1 \wtraS{} \yes$, which implies that $\mVV_1 \paralD \mVV_2 \wtraS{\ftr} \yes \paralD \mVV_2 \traS{\tau} \yes$.
	If $\ftr = \act \ftrr$, then $\mVV_1 \wtraS{\act} \mVV' \wtraS{\ftrr} \yes$.
	Since $\mVV_1 \paralD \mVV_2$ is reactive,
	there is some $q$ such that $\mVV_1 \paralD \mVV_2 \wtraS{\act} q$.
	Therefore, by Lemma \ref{lem:paral-preserves-reactivicity}, either  $\mVV_1 \paralD \mVV_2 \wtraS{} \yes$, or there is some $q'$, such that $\mVV_2 \wtraS{\act} q'$. Therefore,  $\mVV_1 \paralD \mVV_2 \wtraS{\act}  \mVV_1' \paralD q'$, and we are done by the inductive hypothesis.
	
	\paragraph{Claim 4:} If every $\mVV \in S$ accepts $\ftr$, then every reactive $\mVV \in S^+$ accepts $\ftr$.
	The proof is by induction on the construction of $\mVV$ from monitors in $S$.
	If $\mVV \in S$, then by our assumptions, $\mVV$ accepts $\ftr$.
	If $\mVV = \mVV_1 \paralC \mVV_2$ where $\mVV_1, \mVV_2 \in S^+$, then by the inductive hypothesis, $\mVV_1,\mVV_2 \wtraS{\ftr} \yes$;
	using the rules for parallel monitors and induction on $\ftr$, we can complete the proof.
	If $\mVV = \mVV_1 \paralD \mVV_2$ where $\mVV_1 \in S^+$, then by the inductive hypothesis, $\mVV_1 \wtraS{\ftr} \yes$. Then, the proof is complete by Claim 3.

	Using the above claims, we now prove that
	% for every
	% $\mVV$ relevant submonitor of $\mV$,
	$\mV$ accepts $\ftr$ if and only if
	% $\exists q. \mVV_1 \wtraS{\ftr} q$ and
	$\delta^*(\mV,\ftr) (F) = 1$.
	% , and, since $\mV$ is rective, this is enough to complete the proof.
	%
	We prove each implication separately
	for the submonitors of $\mV$.
	
	\paragraph{We first prove that if $\mVV$ accepts $\ftr$, then $\delta^*(\mVV,\ftr) (F) = 1$.}
	By Claim 2, we can assume that $\ftr$ is minimal such that $\mVV \wtraS{\ftr} \yes$.
	If $\mVV$ accepts $\ftr$, then there is an explicit trace $\etr$ that agrees with $\ftr$ on the external actions, such that $\mVV \traS{\etr} \yes$.
	Thus, we prove that for every explicit trace $\etr$,  if $\ftr$ is a finite trace that agrees with $\etr$ on the external actions, $\mVV \traS{\etr} \yes$, and $\ftr$ is minimal such that $\mVV \wtraS{\ftr} \yes$, then $\delta^*(\mVV,\ftr) (F) = 1$.
	We prove this claim by induction on $\etr$:
	\begin{description}
		\item[Case $\etr \in \{\tau\}^*$:] Then $\mVV$ accepts $\varepsilon$, so by the definition of $F$, $\delta^*(\mVV,\ftr) (F) = 1$.
		\item[Case $\etr \notin \{\tau\}^*$ and $\mVV = \mVV_1 \paralC \mVV_2$:]
		% we assume that $\paralG = \paralC$, as the case for $\paralD$ is similar. L
		Let $\ftr = \act \ftrr$.
		Since $\mVV\traS{\etr} \yes$, by Lemma \ref{lem:mon-combinators2}, there are explicit traces $\etr_1,\etr_2$ that agree with $\etr$ (and with $\ftr$) on the external actions and are strictly shorter than $\etr$, such that $\mVV_1 \traS{\etr_1} \yes$ and $\mVV_2 \traS{\etr_2} \yes$.
		By the inductive hypothesis, $\delta^*(\mVV_1,\ftr)(F) = \delta^*(\mVV_2,\ftr)(F) = 1$.
		Let $S = \{ \mVV'' \mid \delta^*(\mVV'',\ftrr)(F) = 1 \}$; 
		by the definition of $\delta^*$, 
		\
		%      we have that 
		$\delta(\mVV_1,\act)(S) = \delta(\mVV_2,\act)(S) = \delta^*(\mVV_1,\ftr)(F) = 1$,
		and thus,
		% by the claim, $\exists q. \mVV_1 \wtraS{\act} q$ and $\exists q. \mVV_2 \wtraS{\act} q$.
		%
		% Thus,
		by the closure properties of $\delta$, 
		we have that 
		$\delta^*(\mVV,\ftr)(F) = \delta(\mVV,\act)(S) = 1$.
		\item[Case $\etr \notin \{\tau\}^*$ and $\mVV = \mVV_1 \paralD \mVV_2$, ]
		let $\ftr = \act \ftrr$.
		Since $\mVV\traS{\etr} \yes$, by Lemma \ref{lem:mon-combinators2}, (without loss of generality) there is an explicit trace $\etr_1$ that agrees with $\etr$ (and with $\ftr$) on the external actions and is strictly shorter than $\etr$, such that $\mVV_1 \traS{\etr_1} \yes$.
		Therefore, 
		%     $\exists q.\mVV_1 \wtraS{\act} q$
		$\mVV_1 \wtraS{\act}$, and by the minimality of $\ftr$ and Lemma \ref{lem:mon-combinators-follow-along},
		%     $\exists q.\mVV_2 \wtraS{\act} q$
		$\mVV_2 \wtraS{\act}$.
		By the inductive hypothesis, $\delta^*(\mVV_1,\ftr)(F) = 1$.
		Let $S = \{ \mVV'' \mid \delta^*(\mVV'',\ftrr)(F) = 1 \}$; by the definition of $\delta^*$, we have that $\delta(\mVV_1,\act)(S) =  \delta^*(\mVV_1,\ftr)(F) = 1$.
		% and thus, by the claim, $\exists q. \mVV_1 \wtraS{\act} q$ and $\exists q. \mVV_2 \wtraS{\act} q$.
		%
		Thus, by the closure properties of $\delta$, we can conclude that $\delta^*(\mVV,\ftr)(F) = \delta(\mVV,\act)(S) = 1$.
		% , and similarly to the previous case,
		% $\delta^*(\mVV,\ftr)(F) = 1$.
		\item[Case $\etr = \tau \etrr $ and $\ftr = \act \ftrr$:] Then
		% we can prove each direction of the biimplication separately.
		% If $\mVV$ accepts $\act \ftrr$, then
		$\mVV \traS{\tau}\mVV'$ for some $\mVV'\traS{\etrr} \yes$ that agrees with $\ftr$ on the external actions.
		By the inductive hypothesis, $\delta^*(\mVV',\ftr) (F) = 1$.
		We prove by induction on the derivation of $\mVV \traS{\tau}\mVV'$ that $\delta^*(\mVV,\ftr) (F) = 1$:
		\begin{description}
			\item[Case $\mVV \traS{\tau}\mVV'$ was derived by rule \rtit{mRecF} or \rtit{mRecB}:]
			Then either $\mVV = p_x$ and $\mVV = \mV_x$, or $\mVV = x$ and $\mVV' = p_x$. Let $S = \{ \mVV'' \mid \delta^*(\mVV'',\ftrr)(F) = 1 \}$. Since $\delta^*(\mVV',\ftr) (F) = 1$, by the definition of $\delta^*$,
			we have that 
			$\delta(\mVV',\act)(S) = \delta^*(\mVV',\ftr)(F) = 1,  $
			and therefore, by the closure conditions of $\delta$,
			we can conclude that
			$\delta^*(\mVV,\ftr)(F) = \delta(\mVV,\act)(S) = 1  .$
			\item[Case $\mVV \traS{\tau}\mVV'$ was derived by  \rtit{mTauL}, \rtit{mVrE}, \rtit{mVrC1}, \rtit{mVrC2}, \rtit{mVrD1}, or \rtit{mVrD2}:] Then we are in the case of $\mVV = \mVV_1 \paralG \mVV_2$, which was handled above.
			% then $\mVV = \mVV_1 \paralG \mVV_2$; we assume that $\paralG = \paralC$, as the case for $\paralD$ is similar.
			% % , and we assume that $\mVV' = \mVV_1' \paralC \mVV_2$.
			% Since $\mVV\traS{\etr} \yes$, by Lemma \ref{lem:mon-combinators}, there are explicit traces $\etr_1,\etr_2$ that agree with $\etr$ (and with $\ftr$) on the external actions and are strictly smaller than $\etr$, such that $\mVV_1 \traS{\etr_1} \yes$ and $\mVV_2 \traS{\etr_2} \yes$.
			% By the inductive hypothesis (for the traces), $\delta^*(\mVV_1,\ftr)(F) = \delta^*(\mVV_2,\ftr)(F) = 1$, and similarly to the previous case,
			% $\delta^*(\mVV,\ftr)(F) = 1$.
			\item[Case $\mVV \traS{\tau}\mVV'$ was derived by rule \rtit{mSelL} or \rtit{mSelR}:] Then
			$\mVV = \mVV_1 + \mVV_2$ and $\mVV_1 \traS{\tau} \mVV'$ or $\mVV_2 \traS{\tau} \mVV'$. By the inductive hypothesis  (for the derivation of $\mVV \traS{\tau} \mVV'$), $\delta^*(\mVV_1,\ftr)(F) = 1$ or $\delta^*(\mVV_2,\ftr)(F) = 1$, and similarly to the previous cases, from the closure conditions of $\delta$, $\delta^*(\mVV,\ftr)(F) = 1$.
			% \item[If $\mVV \traS{\tau}\mVV'$ was derived by rule \rtit{mVrE}:] Then $\mVV' = \stp$ and we have a contradiction, because $\stp$ can only transition to itself, so it accepts no verdicts.
			% \item[If $\mVV \traS{\tau}\mVV'$ was derived by rule \rtit{mVrC1},] then
			% \item[If $\mVV \traS{\tau}\mVV'$ was derived by rule \rtit{mVrC2},] then
			% \item[If $\mVV \traS{\tau}\mVV'$ was derived by rule \rtit{mVrD1},] then
			% \item[If $\mVV \traS{\tau}\mVV'$ was derived by rule \rtit{mVrD2},] then
		\end{description}
		% Furthermore, as we saw, $\mVV \traS{\tau}\mVV'$, which can result from rules \rtit{mRecF}, \rtit{mRecB},
		\item[Final case $\etr = \act \etrr$ and $\ftr = \act \ftrr$,] where $\ftrr$ agrees with $\etrr$ on the external actions: Then
		% we can prove each direction of the biimplication separately.
		% If $\mVV$ accepts $\act \ftrr$, then
		$\mVV \traS{\act}\mVV'$ for some $\mVV' \traS{\ftrr} \yes$.
		By the inductive hypothesis, $\delta^*(\mVV',\ftrr)(F) = 1$.
		By the definition of $\delta^*$, we have that $\delta^*(\mVV,\act \ftrr)(F) = \delta(\mVV,\act) (S)$, where $S = \{ q\in Q \mid \delta^*(q,\ftrr)(F) = 1 \}$.
		We observe that $\mVV' \in S$.
		We now prove, by induction on the derivation of $\mVV \traS{\act} \mVV'$ from the rules of \Cref{fig:monit-instr,fig:monit-parallel}, that $\delta^*(\mVV,\act \ftrr)(F) = 1$, or, equivalently,  that $\delta(\mVV,\act)(S) = 1$.
		\begin{description}
			\item[The base case is that $\mVV \traS{\act} \mVV'$ is produced by rule \rtit{mAct} or \rtit{mVerd}:] Then, $\mVV = \mVV' = \yes$ or $\mVV = \act.\mVV'$.
			If  $\mVV = \yes$, then $\mVV \in F$, and by the definition of $\delta_0$, 
			we have that 
			$\delta(\mVV,\act)(S) = 1$. If $\mVV = \act. \mVV'$, then, since $\mVV' \in S$, by the first closure condition, 
			we infer that 
			$\delta(\mVV,\act)(S) = 1$.
			\item[Case $\mVV \traS{\act} \mVV'$ is produced by rule \rtit{mSeL} or \rtit{mSeR}:] Then the argument is similar to the analogous case for $\mVV \traS{\tau} \mVV'$ above.
			\item[Case $\mVV \traS{\act} \mVV'$ is produced by rule \rtit{mPar}:] Then $\mVV = \mVV_1 \paralG \mVV_2$, which has been handled above.
		\end{description}
	\end{description}
	
	\paragraph{For the other direction,}
	we prove that for every
	immediate submonitor $\mVV$ of a reactive submonitor $\mV'$,
	if
	% $\exists q. \mVV_1 \wtraS{\ftr} r$ and
	$\delta^*(\mVV,\ftr) (F) = 1$, then $\mVV$ accepts $\ftr$.
	Since $\mV$ is reactive, this is enough to complete the proof.
	%
	%
	% We assume that $\delta^*(\mVV,\ftr)(F) = 1$ and that $\exists q. \mVV \wtraS{\act} q$; we will
	We
	prove that $\mVV$ accepts $\ftr$,
	by induction on $\ftr$.
	% Let $\gamma$ be an alternative transition function,
	% such that $\gamma(\mVV',\act)(S) = 1$ iff
	% $\mVV' \wtra{} \yes$ or there is some $\mVV'' \in S$, such that $\mVV' \wtraS{\act} \mVV''$.
	% By simple induction on $\ftr'$, we can see that $\gamma^*(\mVV',\ftr')(F)$ if and only if $\mVV' \wtraS{\ftr'} \yes$.
	% As $\delta$ is minimal, to complete the proof it suffices to demonstrate that $\gamma$ satisfies the same closure conditions as $\delta$, and that if $\delta_0(q,\act)(S) = 1$, then $\gamma(q,\act)(S) = 1$. Yep, it's not true...
	\begin{description}
		\item[Case $\delta^*(\mVV,\varepsilon)(F) = 1$:] Then $\mVV \in F$, and thus, $\mVV$ accepts $\varepsilon$.
		\item[Case $\delta^*(\mVV,\act \ftrr)(F) = 1$:] Then $\delta(\mVV,\act)(S) = 1$, where $S = \{ \mVV' \mid \delta^*(\mVV',\ftrr)(F) = 1 \}$.
		Therefore, either $\delta_0(\mVV,\act)(S) = 1$, or $\delta(\mVV,\act)(S) = 1$ can be derived from the closure conditions for $\delta$; therefore, we can use induction on this derivation of $\delta(\mVV,\act)(S) = 1$. We observe that, from the inductive hypothesis, for every $\mVV' \in S$, $\mVV'$ accepts $\ftrr$.
		By Claim 4, for every reactive $\mVV' \in S^+$, $\mVV'$ accepts $\ftrr$.
		\begin{description}
			\item[The base case is $\delta_0(\mVV,\act)(S) = 1$:] In this case $\mVV \in F$, and thus, $\mVV$ accepts $\varepsilon$ and all its extensions, including $\ftr$.
			\item[Case $\mVV = \act.\mVV'$, where $\mVV' \in S$:] Then $\mVV \traS{\act} \mVV'$ and $\mVV'$ accepts $\ftrr$;
			by Lemma \ref{lem:transition-to-reactive}, $\mVV'$ is reactive, %so $\exists q. \mVV' \wtraS{\ftrr} q$,
			therefore, by the inductive hypothesis, $\mVV'$ accepts $\ftrr$, and so $\mVV$ accepts $\ftr$.
			\item[Case $\mVV = \mVV_1 + \mVV_2$, where $\delta(\mVV_1,\act)(S) = 1$ or $\delta(\mVV_2,\act)(S) = 1$,]
			then
			$\mVV_1,\mVV_2$ are also immediate submonitors of
			% $\mVV$, so also of
			$\mV'$, so by the inductive hypothesis,
			$\mVV_1 \wtraS{\act\ftrr} \yes$ or $\mVV_2 \wtraS{\act\ftrr} \yes$, and by Lemma \ref{lem:mon-combinators}, $\mVV \wtraS{\act\ftrr} \yes$.
			% If $\mVV_1 = \no$ or $\mVV_2 = \no$, then it is not hard to see that
			
			\item[Case $\mVV = x$ or $\mVV = p_x$ and $\delta(\mV_x,\act)(S) = 1$:]
			Then in either case,
			since $x \traS{\tau} p_x \traS{\tau} \mV_x $, by Lemma \ref{lem:transition-to-reactive},
			$\mV_x$ is reactive, and therefore
			% a immediate submonitor of $\mV'$, and
			by the inductive hypothesis, $\mV_x \wtraS{\act\ftrr} \yes$, but
			$x \traS{\tau} p_x \traS{\tau} \mV_x \wtraS{\act\ftrr} \yes$,
			and the proof is thus complete.
			
			\item[Case $\mVV = \mV_1\paralD \mV_2$, where $\delta(\mVV_1,\act)(S) = 1$ or $\delta(\mVV_2,\act)(S) = 1$,
			and 
			%    $\exists q. \mV_1 \wtraS{\act} q$ and $\exists q. \mV_2 \wtraS{\act} q$
			$\mV_1 \wtraS{\act}$ and $\mV_2 \wtraS{\act}$:] Then
			by Claim 1 and rule \rtit{mPar}, there is some $\mV_1\paralD \mV_2 \wtraS{\act} \mV_1'\paralD \mV_2'$, where $\mV_1' \in S^+$ or $\mV_2' \in S^+$ --- therefore, also $\mV_1'\paralD \mV_2' \in S^+$.
			By Lemma \ref{lem:transition-to-reactive}, $\mV_1'\paralD \mV_2'$ is reactive.
			Hence, from the observation above about $S$ and $S^+$, $\mV_1'\paralD \mV_2'$ accepts $\ftrr$, and thus $\mV_1\paralD \mV_2$ accepts $\ftr$.
			% by Lemma \ref{lem:paral-preserves-reactivicity}, so are $\mV_1'$ and $\mV_2'$.
			% Since $\mV_1$
			
			\item[The case for $\mVV = \mV_1\paralD \mV_2$] is similar.
			\qedhere
		\end{description}
	\end{description}
	%
	%
	% if $\delta^*(\mVV,\act s)(F) = 1$, then for $D = \{ \mVV' \in Q \mid \delta^*(\mVV',s)(F) = 1 \}$,
	% $\delta(\mVV,\act)(D) = 1, $ so
	% we can use induction on the
	% closure conditions for $\delta$ to prove that $\mVV$ accepts $t$.
	% The base case is that $\delta_0 (\mVV,\act)(D) = 1$, which implies that $\mVV \in F$, and by definition (of $F$), $n$ accepts $\varepsilon$, and therefore also $t$.
	% For each remaining closure condition, there should  be a corresponding rule/lemma/observation.
\end{proof}

%% file: complete-monitorability-appendix.tex
% !TEX root = main.tex

\begin{propbis}{\ref{prop:vedict-equiv-implies-complete-formula}}
  % \label{prop:complete-semantic-equivalence}
  If \mV is sound and complete for \hV then
  \begin{enumerate}
    \item $\mV \mveq \mVV$ implies \mVV is sound and complete for \hV.
    \item \mV is a sound and complete monitor for $\hV'$ implies $\hSemL{\hV}=\hSemL{\hV'}$
  \end{enumerate}
\end{propbis}
\begin{proof}
  For the first clause, we need to prove soundness and completeness for \mVV.

  For \emph{soundness}, \Cref{def:soundness-n-completeness}, assume \rej{\mVV,\tV}, \ie $\exists\; p,\ftV \cdot \rej{\mVV,p,\ftV}$ and \ftV\ is a prefix of \tV.
  By \Cref{def:acc-n-rej} and \Cref{lem:unzipping},
  it follows
%  we know
  that $\mVV \wtraS{\ftV} \no$.
  By $\mV \mveq \mVV$ we know that $\mV \wtraS{\ftV} \no$ which, in turn, implies that  \rej{\mV,\tV}.
  Since \mV is sound (and complete) for \hV, it must be the case that $\tV \not\in\hSem{\hV}$, which is the result we want.
  The case for \acc{\mVV,\tV} is analogous.

  The argument for \emph{completeness}, \Cref{def:soundness-n-completeness}, is similar.
  Pick a trace $\tV \in \hSem{\hV}$.
  Since \mV is complete for \hV,
  we prove
%  it must be
  that $\acc{\mV,\tV}$.
  Using the fact that $\mV \mveq \mVV$, \Cref{def:acc-n-rej} and \Cref{lem:unzipping}, we can then deduce that $\acc{\mVV,\tV}$ which is the required result.
  The case for  $\tV \not\in \hSem{\hV}$ is analogous.

  For the second clause,  pick a $\tV\in\hSemL{\hV}$ without loss of generality.
  By completeness, \Cref{def:soundness-n-completeness}, $\acc{\mV,\tV}$, and by soundness, \Cref{def:soundness-n-completeness}, $\tV\in\hSemL{\hV'}$.
\end{proof}

We now present the proof demonstrating the complete-monitorability of the syntactic fragment \HML from \Cref{def:complete-fragment-HML}.
To prove \Cref{prop:hml-monitorable}, we
first show that all the synthesised monitors \hSyn{\hV} are \emph{reactive}, as defined in \Cref{def:reactive-mon}.

\begin{lemmabiss}{\ref{lem:synt-complete-mon-reactive}}
  For all $\hV \in \HML$, \hSyn{\hV} is reactive.
\end{lemmabiss}
\begin{proof}
The proof proceeds by structural induction on \hV.
The cases for \hTru,\hFls,\hSuf{\ASet}{\hVV} and  \hNec{\ASet}{\hVV} are immediate.
For the case of \hAnd{\hVV_1}{\hVV_2}, we know from \Cref{def:mon-synt-complete} that  $\hSyn{\hAnd{\hVV_1}{\hVV_2}} = \hSyn{\hVV_1} \paralC \hSyn{\hVV_2}$.
By the inductive hypothesis we know that both \hSyn{\hVV_1} and \hSyn{\hVV_2} are reactive and, by \rtit{mPar} of \Cref{fig:monit-parallel}, it follows that $\hSyn{\hVV_1} \paralC \hSyn{\hVV_2}$ is reactive as well.
The case for  \hOr{\hVV_1}{\hVV_2} is analogous.
\end{proof}

\begin{propbis}{\ref{prop:hml-monitorable}}
For all $\hV \in \HML$, \hSyn{\hV} is a sound and complete monitor for \hV.
\end{propbis}
\begin{proof}
For \emph{soundness}, \Cref{def:soundness-n-completeness}, we need to show that
$(i)$ \rej{\hSyn{\hV},\tV}   implies $\tV\not \in \hSem{\hV}$ and
$(ii)$ \acc{\hSyn{\hV},\tV} implies $\tV\in \hSem{\hV}$.
We proceed by structural induction on \hV, and the main cases are:
\begin{description}
  \item[Case \hAndF and \hOrF:]
  By \Cref{def:mon-synt-complete} we know that $\hSyn{\hAndF} = \hSyn{\hV_1} \paralC \hSyn{\hV_2}$.
  If  \rej{\hSyn{\hV_1} \paralC \hSyn{\hV_2},\tV}, by \Cref{def:acc-n-rej},  there exist $p,\ftV$ such that \ftV\ is a prefix of \tV\ and \rej{\hSyn{\hV_1} \paralC \hSyn{\hV_2},p,\ftV}.
  Using \Cref{def:acc-n-rej} and \Cref{lem:unzipping} we know that $\sys{(\hSyn{\hV_1} \paralC \hSyn{\hV_2})}{p} \wtraS{\ftV} \sys{\no}{p'}$ for some $p'$.
  By \Cref{lem:mon-combinators}, this implies that either $\sys{\hSyn{\hV_1}}{p} \wtraS{\ftV} \sys{\no}{p'}$   or $\sys{\hSyn{\hV_2}}{p} \wtraS{\ftV} \sys{\no}{p'}$, which means that
  either \rej{\hSyn{\hV_1},\tV} or \rej{\hSyn{\hV_2},\tV}.
  By the I.H., we deduce that either $\tV \not \in \hSem{\hV_1}$ or $\tV \not \in \hSem{\hV_2}$ which is enough to conclude that $\tr \not \in \hSem{\hAndF}$.
  If \acc{\hSyn{\hV_1} \paralC \hSyn{\hV_2},\tV}, then by \Cref{def:acc-n-rej}, \Cref{lem:unzipping,lem:mon-combinators} and the I.H. we obtain
  $\tV \in \hSem{\hV_1} $ and  $\tV \in \hSem{\hV_2}$, and therefore we conclude $\tr \in \hSem{\hAndF}$.
  The case for \hOrF is analogous.
  \item[Case \hNec{\ASet}{\hV} and \hSuf{\ASet}{\hV}:]  In the case of \hNec{\ASet}{\hV}, by \Cref{def:mon-synt-complete} we know that
   $\hSyn{\hNec{\ASet}{\hV}}= \ch{\prf{\ASet}{\hSyn{\hV}}}{\uprf{\ASet}{\yes}}$.
  If \rej{\ch{\prf{\ASet}{\hSyn{\hV}}}{\uprf{\ASet}{\yes}},\tV} then there exist $p,\ftV$ such that \ftV\ is a prefix of \tV\ and \rej{\ch{\prf{\ASet}{\hSyn{\hV}}}{\uprf{\ASet}{\yes}},p,\ftV}.
  From \Cref{def:acc-n-rej,lem:unzipping} we know that  $\ch{\prf{\ASet}{\hSyn{\hV}}}{\uprf{\ASet}{\yes}} \wtraS{\ftV} \no$ and,
   from the structure of the monitor and \Cref{lem:mon-combinators}, it must be the case that $\prf{\ASet}{\hSyn{\hV}} \wtraS{\ftV} \no$.
   This means that $\ftV = \acta\ftVV$ where $\acta \in \ASet$ and $\hSyn{\hV} \wtraS{\ftVV} \no$, which in turn implies that $\tV=\acta\tVV$ and \rej{\hSyn{\hV},\tVV}.
   By the I.H., \rej{\hSyn{\hV},\tVV} implies that $\tVV \not\in \hSem{\hV}$ which suffices to conclude that $\tV=\acta\tVV \not\in \hSem{\hNec{\ASet}{\hV}}$.
    If  \acc{\ch{\prf{\ASet}{\hSyn{\hV}}}{\uprf{\ASet}{\yes}},\tV}, then by \Cref{def:acc-n-rej}, \Cref{lem:unzipping,lem:mon-combinators} we know that either
    $\prf{\ASet}{\hSyn{\hV}} \wtraS{\ftV} \yes$  or $\uprf{\ASet}{\yes} \wtraS{\ftV} \yes$ for some \ftV\ is a prefix of \tV.
    In the latter case, we deduce that $\tV=\acta\tVV$ for some $\acta,\tVV$ where $\acta \not\in \ASet$, which trivially implies that $\tV \in \hSem{\hNec{\ASet}{\hV}}$.
    In the former case, we know that  $\tV=\acta\tVV$ for some $\acta,\tVV$ where $\acta \in \ASet$ and \acc{\hSyn{\hV}, \tVV}  which, by the I.H., implies that $\tVV \in \hSem{\hV}$ and hence $\tV \in \hSem{\hNec{\ASet}{\hV}}$.
    The case for $\hSuf{\ASet}{\hV}$ is similar.
\end{description}

For completeness, \Cref{def:soundness-n-completeness}, we need to show that
$(i)$ \emph{violation-completeness}, \ie $\tV\not \in \hSem{\hV}$  implies \rej{\hSyn{\hV},\tV} and
$(ii)$ \emph{satisfaction-completeness}, \ie $\tV\in \hSem{\hV}$  implies \acc{\hSyn{\hV},\tV}.
Again, we proceed by structural induction on \hV, and the main cases are:
\begin{description}
  \item[Case \hAndF and \hOrF:]
  If $\tV\not \in \hSem{\hAndF}$, then $\tV \not \in \hSem{\hV_1}$ or $\tV \not \in \hSem{\hV_2}$.
  Without loss of generality, assume the former, \ie $\tV \not \in \hSem{\hV_1}$.
  By the I.H.,  we have \rej{\hSyn{\hV_1},\tV} which, by \Cref{def:acc-n-rej,lem:unzipping} means that there exists
  % a process $p$ and
  a prefix \ftV of \tV such that $\hSyn{\hV_1} \wtraS{\ftV} \no$.
  Since, $\hSyn{\hAndF} = \hSyn{\hV_1} \paralC \hSyn{\hV_2}$,  by \Cref{lem:mon-combinators}, we conclude that $\rej{\hSyn{\hAndF},\tV}$.
  The proof for $\tV \in \hSem{\hAndF}$ follows a similar structure, using the fact that both $\tV \in \hSem{\hV_1}$ or $\tV \in \hSem{\hV_2}$.
  The case for \hOrF is analogous.
  \item[Case \hNec{\ASet}{\hV} and \hSuf{\ASet}{\hV}:]
    If $\tV\not \in \hSem{\hNec{\ASet}{\hV}}$ then, by \Cref{fig:recHML}, it must be the case that $\tV=\acta\tVV$ for some $\acta \in \ASet$ and $\tVV\not \in \hSem{\hV}$.
    By the I.H., we know that $\rej{\hSyn{\hV},\tVV}$,
    from which one can then conclude that \rej{\ch{\prf{\ASet}{\hSyn{\hV}}}{\uprf{\ASet}{\yes}}, \tV} where
    $\hSyn{\hNec{\ASet}{\hV}}=\ch{\prf{\ASet}{\hSyn{\hV}}}{\uprf{\ASet}{\yes}}$, via \Cref{lem:unzipping,lem:mon-combinators}.
     If $\tV \in \hSem{\hNec{\ASet}{\hV}}$ then, by \Cref{fig:recHML}, it must be one of two cases.
     Either $\tV = \acta \tVV$ where $\acta \not\in \ASet$, which implies that $\acc{\hSyn{\hNec{\ASet}{\hV}},\tV}$, since $\hSyn{\hNec{\ASet}{\hV}} = \ch{\prf{\ASet}{\hSyn{\hV}}}{\uprf{\ASet}{\yes}}$.
    Else $\tV = \acta \tVV$ where $\acta \in \ASet$ and $\tVV \in \hSem{\hV}$.
    By I.H., we deduce that $\acc{\hSyn{\hV},\tVV}$, and by \Cref{lem:unzipping,lem:mon-combinators} we are able to construct an acceptance computation for \ch{\prf{\ASet}{\hSyn{\hV}}}{\uprf{\ASet}{\yes}}, hence \acc{\hSyn{\hNec{\ASet}{\hV}},\tV}.
    The case for $\hSuf{\ASet}{\hV}$ is similar.
    \qedhere
\end{description}
\end{proof}

\medskip

We now procede to give the proof for the maximality of \HML from \Cref{def:complete-fragment-HML}.
The following are technical lemmata leading up to \Cref{lem:remove-rec-complete-monitor}.

\begin{lem} \label{lem:remove-rec-regular-monitor-deterministic} For each $\mV\in\RMon$: \quad
  \begin{enumerate}
    \item If \mV is a syntactically deterministic monitor, then $\noRec{\mV}$ is also
   syntactically
    deterministic.
    \item If \mV is a reactive and syntactically deterministic monitor, then $\noRec{\mV}$ is also
   reactive.
  \end{enumerate}
\end{lem}
\begin{proof}
 By structural induction on \mV.
\end{proof}

\begin{lem}\label{lem:deterministic-sound-complete-formula-rec}
  Suppose that the syntactically deterministic monitor \rec{x}{\mV} is sound and complete for some formula \hV and that $\rec{x}{\mV} \wtraS{\ftV} \vV$ for some finite trace \ftV.   Then $\mV \wtraS{\ftV} \vV$.
\end{lem}
\begin{proof}
  From $\rec{x}{\mV} \traS{\actt} \mV\subS{\rec{x}{\mV}}{x} \wtraS{\ftV} \vV$ and \Cref{lem:open-vs-closed-mon-multi-reduction} we have the following two cases to consider.
  \begin{enumerate}
    \item Assume that there exists some $\mV'$ such that $\mV \wtraS{\ftV} \mV'$ and $ \mV'\subS{\rec{x}{\mV}}{x} = \vV$.
    This immediately
    yields the claim,
%    gives us the result we want
    since  $\mV'\subS{\rec{x}{\mV}}{x} = \vV$ can only hold if $\mV' =\vV$.
    \item Assume that there exist $\ftV_1,\ftV_2$ and $\mV'$ where
    $\ftV=\ftV_{1}\ftV_{2}$  (and $\ftV_2 \neq \epsilon$)  and   $\mV\wtraS{\ftV_1}\mV'=x$ (because \rec{x}{\mV} is syntactically deterministic)
    and $\rec{x}{\mV}\wtraS{\ftV_2} \vV$.
    Stated otherwise, we have
    \begin{equation} \label{eq:deterministic-sound-complete-formula-rec}
      \rec{x}{\mV}\wtraS{\ftV_1} x\subS{\rec{x}{\mV}}{x} = \rec{x}{\mV} \wtraS{\ftV_2} \vV.
    \end{equation}
    %
%    Note also that, since \rec{x}{\mV} is syntactically deterministic, we also have $\ftV_1 \neq \varepsilon$.
%%%%%%%%(I don't see this...)
We show that we can reach a contradiction, and therefore this case cannot occur.
%    We show that this
%%    cannot be the
%    case
%    is impossible,
%    since it would otherwise contradict our assumption that
%		\rec{x}{\mV} is sound and complete for some formula \hV.
    %

    If $\ftV_1 = \varepsilon$, then for some $k>0$
    $\rec{x}{\mV} (\traS{\tau})^k \rec{x}{\mV}$, and therefore for all $k>0$,
    $\rec{x}{\mV} (\traS{\tau})^k $. By \Cref{lem:determinism-and-taus}, we then have that for all $\mVV$, if $\mV \wtraS{} \mVV$ then $\forall \act \cdot \mVV \centernot{\traS{\act}}$, and therefore
    $\ftV_2 = \varepsilon$, which is a contradiction.
    Therefore, $\ftV_1$ must be non-empty.

    Consider the trace $\ftV_1^\omega$.
    We must  have either $\ftV_1^\omega \in \hSemL{\hV}$ or $\ftV_1^\omega \not\in \hSemL{\hV}$; without loss of generality, assume the former.
    %
%    We argue that
%    $\rec{x}{\mV}$ can never reach $\vV'=\yes$, which would contradict the assumption that it is complete for \hV.
%    %
%    If it did, then we would have
%
    Since $\rec{x}{\mV}$ is sound and complete for $\hV$,
    $\rec{x}{\mV} \etraS{\etV} \yes$ for some $\etV\neq \varepsilon$ where $\filter{\etV}$ is a prefix of $\ftV_1$.
    By \Cref{lem:ver-persistence}, this yields that $\rec{x}{\mV}\wtraS{\ftV_1} \yes$, and by \Cref{lem:det_mon_is_det}, it must be the case that $\rec{x}{\mV} \mveq \yes$, and therefore $\rec{x}{\mV} \wtraS{} \yes$,
%
%    Therefore, $\ftV_{1} \neq \varepsilon$,
    and by \Cref{lem:strong-determinism3}, $x = \yes$, which is also a contradiction.
%
%    %
%    By \Cref{lem:strong-determinism}, this also implies that
%     $\rec{x}{\mV}\wtraS{\ftV_1} \rec{x}{\mV}$ from \Cref{eq:deterministic-sound-complete-formula-rec} must be of the form
%    $\rec{x}{\mV} \etraS{\etV} \vV \etraS{\etVV} \rec{x}{\mV}$
%    where $\filter{\etV\etVV} = \ftV_1$.
%    %
%    But by \Cref{lem:ver-persistence}, we can never have $\vV \etraS{\etVV} \rec{x}{\mV}$, which means that we cannot have any $\etV$ where  $\rec{x}{\mV} \etraS{\etV} \vV$.
\qedhere
  \end{enumerate}
\end{proof}

\begin{lemmabiss}{\ref{lem:remove-rec-complete-monitor}}
  If \mV is a syntactically deterministic monitor that is sound and complete \hV, then $\noRec{\mV}$ is also a sound and complete monitor for \hV.
\end{lemmabiss}
\begin{proof}
  Using \Cref{prop:vedict-equiv-implies-complete-formula}, the required result follows if we can show that $\mV\mveq\noRec{\mV}$.

  In one direction, we have to show that, for $\vV \in \sset{\yes,\no}$,  $\noRec{\mV}\wtraS{\ftV}\vV$  implies $\mV\wtraS{\ftV}\vV$.
  We proceed by structural induction on the string \etV\ where $\noRec{\mV}\etraS{\etV}\vV$ and $\filter{\etV}=\ftV$.
  \begin{description}
    \item[Case $\etV=\varepsilon$:]
    $\noRec{\mV} = \vV$ which implies that $\mV=\rec{x_1}\cdots \rec{x_n}{\vV}$ by \Cref{def:no-rec-mon}.
    We therefore obtain $\mV\etraS{{\actt^n}}\vV$, since $\filter{{\actt^n}}=\filter{\varepsilon}=\varepsilon$.
    \item[Case $\etV=\actu\etV'$:]
    By \Cref{lem:remove-rec-regular-monitor-deterministic}(1) and \Cref{def:no-rec-mon}, since \noRec{\mV} is a syntactically deterministic regular monitor and it does not contain any recursion, $\actu=\acta$ for some $\acta\in\Act$.
    Thus we know that $\ftV=\acta\ftVV$ for some \ftVV.
    There are three subcases to consider.
    \begin{description}
      \item[Case $\noRec{\mV}=\no$ or $\noRec{\mV}=\yes$:] The proof is analogous to that of the base case.
      \item[Case $\noRec{\mV}=\stp$:] By \Cref{lem:ver-persistence}, this would contradict $\noRec{\mV}\etraS{\etV}\vV$ for $\vV \in \sset{\yes,\no}$.
      \item[Case $\noRec{\mV}=\sum_{\acta \in \ASet} \prf{\acta}{\mV_\acta}$:]
      From the structure of the monitor \noRec{\mV}, we deduce that
      \begin{align}
        \label{eq:norec-0}
        & \noRec{\mV} \traS{\acta} \mV_\acta \etraS{\etV'}\vV \quad \text{where}\quad \filter{\etV'}=\ftVV.
      \end{align}
      Again, from the structure the monitor \noRec{\mV} and the fact that $\mV$ is syntactically deterministic (\Cref{def:determinism}), we use \Cref{def:no-rec-mon} to conclude that
      \begin{align}
        \label{eq:norec-1}
        & \mV = \textstyle \rec{x_1}\cdots \rec{x_n}{\sum_{\acta \in \ASet} \prf{\acta}{\mVV_\acta}}\\
        \label{eq:norec-2}
        &\text{where, for every }  \acta\in\ASet \text{ we have } \mV_\acta = \noRec{\mVV_\acta}
        .
      \end{align}
      Now, from \Cref{eq:norec-1},  we can derive $\mV \etraS{{\actt^n \acta}} \mVV_\acta$ where, clearly, $\filter{\actt^n \acta} = \acta$.
      By \Cref{eq:norec-2},  $\mV_\acta \etraS{\etV'}\vV$ of \Cref{eq:norec-0} and the inductive hypothesis we obtain that $\mVV_\acta \wtraS{\ftVV}\vV$.
      % for some \ftVV\ where $\filter{\etV'} = \ftVV$.
      %
      Thus, we deduce that $\mV \wtraS{\acta\ftVV} \vV$ as required.
    \end{description}
  \end{description}

    For the other direction, we have to show that, for $\vV \in \sset{\yes,\no}$, $\mV\wtraS{\ftV}\vV$  implies $\noRec{\mV}\wtraS{\ftV}\vV$.
   Again, we proceed by structural induction on \etV\ where $\mV\etraS{\etV}\vV$ and $\filter{\etV}=\ftV$.
   \begin{description}
     \item[Case $\etV=\varepsilon$:]
     Trivially true, since $\mV = \vV$ implies that $\noRec{\mV}=\vV$ by \Cref{def:no-rec-mon}.
     \item[Case $\etV=\actu\etV'$:]
     We have two subcases to consider:
     \begin{description}
       \item[Case $\actu=\actt$:]
       By \Cref{def:determinism} we know that $\mV = \rec{x}{\mVV}$ for some $\mVV$.
       % $\rec{x}{\mVV} \traS{\actt}\mVV\subS{\rec{x}{\mVV}}{x} \etraS{\etV'}\vV$.
       %
       By \Cref{lem:deterministic-sound-complete-formula-rec}, we deduce that  $\mVV\wtraS{\ftV}\vV$.
       By the inductive hypothesis we obtain $\noRec{\mVV} \wtraS{\ftV} \vV$.
       The required results follows by \Cref{def:no-rec-mon}, since $\noRec{\mV}=\noRec{\rec{x}{\mVV}} = \noRec{\mVV}$, and thus
       $\noRec{\mV} \wtraS{\ftV} \vV$.
       \item[Case $\actu=\acta$:]
       By \Cref{def:determinism} we deduce that $\mV = \sum_{\acta \in \ASet} \prf{\acta}{\mV_\acta}$ where
       $\mV_\acta \etraS{\etV'} \vV$.
       By \Cref{def:no-rec-mon} we know that
       $\noRec{\mV} = \sum_{\acta \in \ASet} \prf{\acta}{\noRec{\mV_\acta}}$, where we can derive  $\sum_{\acta \in \ASet} \prf{\acta}{\noRec{\mV_\acta}} \traS{\acta} \noRec{\mV_\acta}$.
         The required result follows from $\mV_\acta \etraS{\etV'} \vV$ and the inductive hypothesis, from which we obtain  $\noRec{\mV_\acta} \wtraS{\ftVV} \vV$ where $\filter{\etV'}=\ftVV$ and $\acta\ftVV=\ftV$.
       \qedhere
     \end{description}
  \end{description}
\end{proof}

The following lemmata relate to properties of the formula synthesis function of \Cref{def:formula-synt-complete}.

% \begin{lem}\label{lem:fin-mon-implies-verdict}
%   For any \emph{reactive} $\mV\in\FMon$ and any $\tV$ either \acc{\mV,\tV} or \rej{\mV,\tV}
% \end{lem}
% \begin{proof}
%   By structural induction on $\mV\in\FMon$.
% \end{proof}

\begin{cor}\label{cor:complete-formula-synt-implies-HML}
  For any $\mV\in\FMon$, $\mSyn{\mV} \in \HML$ \qed
\end{cor}

\begin{lemmabiss}{\ref{lem:consistent-implies-sound-and-complete}}
  Any \emph{reactive} monitor  $\mV\in\FMon$  is a sound and complete monitor for \mSyn{\mV}.
\end{lemmabiss}
\begin{proof}
  We treat soudness and completeness spearately and proceed by induction on the structure of \mV.  The main case is when $\mV=\sum_{\acta\in\Act} \prf{\acta}{\mV_\acta}$ where $\mSyn{\mV}=\bigwedge_{\acta\in \Act} \hNec{\acta}{\mSyn{\mV_\acta}}$.
  \begin{description}
    \item[Soundness:]
    Pick a \tV.
    From the structure of the monitor \mV, \rej{\mV,\tV} implies that $\tV = \acta\tVV$ for some \acta\ and \tVV\ where \rej{\mV_\acta,\tVV}.
    By the inductive hypothesis we know that $\tVV \not\in\hSemL{\mSyn{\mV_\acta}}$ which implies that  \tV\ violates $\mSyn{\mV}$ since $\tV=\acta\tVV \not\in \hSemL{\bigwedge_{\acta\in \Act} \hNec{\acta}{\mSyn{\mV_\acta}}} = \hSemL{\mSyn{\mV}}$.
    The argument for \acc{\mV,\tV} is analogous where we note that, since $\tV = \acta\tVV$, any subformula \hNec{\actb}{\mSyn{\mV_\actb}} where $\actb\neq\acta$ is satisfied trivially by \tV.
    \item[Completeness:]
    Pick a \tV; it must be of the form $\tV=\acta\tVV$
    If $\tV\not\in \hSemL{\bigwedge_{\acta\in \Act} \hNec{\acta}{\mSyn{\mV_\acta}}}$  then it must be because $\tVV\not\in\hSemL{\mSyn{\mV_\acta}}$.
    By the inductive hypothesis, we obtain that \rej{\mV_\acta,\tVV} which in turn implies that
    \rej{\sum_{\acta\in\Act} \prf{\acta}{\mV_\acta}, \acta\tV}.
    The case for
    $\tV\in \hSemL{\bigwedge_{\acta\in \Act} \hNec{\acta}{\mSyn{\mV_\acta}}}$ is analogous.
    \qedhere
  \end{description}
\end{proof}

\Cref{prop:hml-maximal} also makes use of the following technical lemma stating that a deterministic monitor that is complete \wrt some formula must necessarily be reactive.

\begin{lem}\label{lem:complete-implies-reactive}
  If \mV is a
%  syntactically
  deterministic complete monitor for some formula \hV, then it must be reactive.
\end{lem}

\begin{proof}
  % [an alternative suggestion]
	Let $\mVV \in \reach{\mV}$ and let $\act \in \Act$. By \Cref{def:reactive-mon}, it suffices to prove that $\mVV \wtraS{\act}$.
	Since, $\mVV \in \reach{\mV}$, there must be some $\ftV \in \Act^*$, such that
	$\mV \wtraS{\ftV} \mVV$.
	Let $\tV = \ftV \act \tVV$ for some $\tVV$.
	Since $\mV$ is complete for \hV, there is a verdict $\vV$ and a finite prefix $\ftVV$ of $\tV$, such that
	$\mV \wtraS{\ftVV} \vV$.
	If $\ftVV$ is a prefix of $\ftV$, then by \Cref{lem:ver-persistence},
	$\mV \wtraS{\ftV} \vV$, and therefore by \Cref{lem:det_mon_is_det}, $\mVV \wtraS{} \vV \traS{\act} \vV$.
	If $\ftV\act$ is a prefix of $\ftVV$, then $\ftVV = \ftV \act \ftV'$ for some $\ftV'$ and for some $\mV'$,
	$\mV \traS{\ftV} \mV' \traS{\act\ftV'} \vV$. Therefore, by \Cref{lem:det_mon_is_det},
	$\mVV \mveq \mV'$, so $\mVV \traS{\act\ftV'} \vV$, yielding that $\mVV \wtraS{\act}$.
\end{proof}

We are now in a position to give a proof of maximality for \HML. We actually give two proofs: the first one is constructive as reported in \Cref{sec:complete-monitorability}, whereas the other one is stronger (albeit non-constructive) and shows that this expressivity result holds for any logic, not just \recHML.

\begin{propbiss}{\ref{prop:hml-maximal}}{Maximality for \HML}
  For any $\hV\in\UHML$ if \hV is complete-monitorable, then there exists $\hVV\in\HML$ such that $\hSemL{\hV}=\hSemL{\hVV}$.
\end{propbiss}
% \af{Check about reactivity}
\begin{proof}
  % \af{ToDo}
  Pick any  $\hV\in\UHML$ that is complete-monitorable.
  By \Cref{def:complete-monitorability}, there exists a monitor \mV that is sound and complete for \hV.
  By \Cref{prop:extended-mon-to-reg-mon} and \Cref{prop:determinization}, there exists a syntactically deterministic regular monitor $\mV'$ that is verdict-equivalent to \mV.
  By \Cref{prop:vedict-equiv-implies-complete-formula}, monitor  $\mV'$ is also sound and complete for \hV.
  Moreover, by \Cref{lem:complete-implies-reactive}, monitor  $\mV'$ must also be reactive.
  By \Cref{lem:remove-rec-complete-monitor} and \Cref{lem:remove-rec-regular-monitor-deterministic},
  there exists a reactive \emph{recursion-free} deterministic regular monitor $\mV''\in\FMon$ that is verdict-equivalent to $\mV'$.
  Again, by \Cref{prop:vedict-equiv-implies-complete-formula}, monitor  $\mV''$ is also sound and complete for \hV.

  Now, by
  % \Cref{cor:sound-implies-consistent} and
   \Cref{lem:consistent-implies-sound-and-complete}, $\mV''$ is sound and complete for $\mSyn{\mV''}$ as well.
  By \Cref{cor:complete-formula-synt-implies-HML}, we know that $\mSyn{\mV''} \in \HML$.
  Thus, by \Cref{prop:complete-semantic-equivalence} we conclude that $\hSemL{\hV}=\hSemL{\mSyn{\mV''}}$ as required.
\end{proof}

\begin{rem}
	The proof of  \Cref{prop:hml-maximal} is constructive.  We are able to prove (albeit in a non-constructive manner) an even stronger result with respect to complete monitoring for an arbitrary logic that is defined over traces. This increases the importance of the  logic identified in \Cref{def:complete-fragment-HML} with the linear-time interpretation of \Cref{fig:recHML}. \qedd
\end{rem}

\begin{thmbiss}{\ref{thm:stronger-HML-maximality}}  
	Let $\mV$ be a monitor from a monitoring system with the following two properties: 
	\begin{enumerate} 
		\item 
%	(1) 
	verdicts are irrevocable, that is, if $\mV$ accepts (respectively, rejects) a finite trace $\ftV$, then it accepts (respectively, rejects) all its extensions, and 
%	(2) 
	\item 
	$\mV$ accepts (respectively, rejects) a trace $\tV$ if, and only if, it accepts (respectively, rejects) some finite prefix $\ftV$ of $\tV$.
\end{enumerate}
	For any property \hV with a trace interpretation (not necessarily syntactically represented using \UHML), if \mV is sound and complete for \hV then \hV can be expressed via the syntactic fragment \HML of \Cref{def:complete-fragment-HML}.
\end{thmbiss}

\begin{proof}
   A monitor, irrespective of its syntactic structure, is a computational entity that reaches a verdict after a finite sequence of observations/actions: at  this point verdicts are \emph{irrevocable} meaning that further actions observed would not change the status of the monitor.

   Let $\mV$ be such a monitor
   % such
   that
   % $\mV$
   is sound and complete for $\hV$.
   Let $L_a(\mV) = \{ \ftV \in \Act^* \mid \mV$ accepts $\ftV \}$ and $L_r(\mV) = \{ \ftV \in \Act^* \mid \mV$ rejects $\ftV \}$.
   Due to soundness, $L_a(\mV) \cap L_r(\mV) = \emptyset$.
   Now, let
   \begin{align*}
   \min L_a(\mV) &= \{ \ftV \in L_a(\mV) \mid \forall \ftVV,\ftVV' \in \Act^*. (\ftVV\ftVV' = \ftV \Rightarrow \ftVV \notin L_a(\mV) \text{ or } \ftVV = \ftV) \} \text{, and}\\
   \min L_r(\mV) &= \{ \ftV \in L_r(\mV) \mid \forall \ftVV,\ftVV' \in \Act^*. (\ftVV\ftVV' = \ftV \Rightarrow \ftVV \notin L_r(\mV) \text{ or } \ftVV = \ftV) \}.
   \end{align*}
   We observe that both $L_a(\mV)$ and $L_r(\mV)$ are suffix-closed, meaning that if a finite trace is in $L_a(\mV)$ or $L_r(\mV)$, then so are all of its finite extensions.
   Therefore,
   $L_a(\mV) = \{ \ftV\ftVV \in \Act^* \mid \ftV \in L_a(\mV) \}$ and
   $L_r(\mV) = \{ \ftV\ftVV \in \Act^* \mid \ftV \in L_r(\mV) \}$.
   We can define $\ftV.\mVV$ recursively thus: $\varepsilon.\mVV = \mVV$ and $\act\ftV.\mVV = \act.(\ftV.\mVV)$.
   If both $\min L_a(\mV)$ and $\min L_r(\mV)$ are finite, then we can define regular monitor
   \[
    \mVV = \sum_{\ftV \in \min L_a(\mV)} \ftV.\yes + \sum_{\ftV \in \min L_r(\mV)} \ftV.\no,
   \]
   and it is not hard to see that $\mVV$ accepts and rejects exactly the same traces as $\mV$:
   if $\mV$ rejects $\tV$, then it rejects a finite prefix $\ftV$ of $\tV$,
   % that $\mV$ rejects,
   so $\ftV \in L_r(\mV)$, and therefore $\ftV = \ftVV\ftVV'$ for some $\ftVV \in \min L_r(\mV)$, which is then rejected by $\mVV$.
   The other direction and the case for acceptance are similar.

   Therefore, it suffices to prove that $\min L_a(\mV) \cup \min L_r(\mV)$ is finite.
   If $\varepsilon \in \min L_a(\mV) \cup \min L_r(\mV)$, then we can immediately see that $\min L_a(\mV) \cup \min L_r(\mV) = \{\varepsilon\}$, which is a finite set.
   Otherwise, let $L' = \{ \ftV \mid \exists \act \in \Act.~\ftV\act \in \min L_a(\mV) \cup \min L_r(\mV) \}$.
   We can observe that $\mV$ can neither accept nor reject  $\ftV \in L'$, because otherwise, without loss of generality, $\ftV \in L_a(\mV)$, so if $\ftV \act \in \min L_a(\mV)$, then $\ftV\act$ is not minimal in $L_a(\mV)$ and we have a contradiction, while if $\ftV \act \in \min L_r(\mV)$, then $\ftV\act \in L_a(\mV) \cap L_r(\mV)$, which is also a contradiction.
   Therefore $L' \subseteq L$, where $L \subseteq \Act^*$ is the set of finite traces that $\mV$ does not accept or reject.
   Therefore, it suffices to prove that $L$ is finite, which we proceed to do.

   To reach a contradiction, we assume that $L$ is infinite. Let $T = (L,\reduc,\Act)$ be
   the tree-LTS where for every $\act \in \Act$ and all $\ftV,\ftVV$, $\ftV \traS{\act} \ftVV$ if and only if $\ftVV = \ftV\act$.
   As we have established above, $\mV$ does not accept or reject $\varepsilon$, therefore $\varepsilon \in L$.
   It is not hard to see that for every $\ftV,\ftVV \in \Act^*$, if $\varepsilon \traS{\ftV} \ftVV$, then $\ftV = \ftVV$, by easy induction on $\ftV$.
   Similarly, if $\ftV \in L$, then $\varepsilon \traS{\ftV} \ftV$.
   Since $\Act$ is finite, $T$ is finitely-brancing, and therefore, by K\"{o}nig's Lemma, since $L$ is infinite, it must be the case that there is an infinite path (trace) $\tV$ in $T$ from $\varepsilon$.
   For every finite prefix $\ftV$ of $\tV$, $\varepsilon \traS{\ftV} \ftV$, and therefore, $\ftV \in L$.
   Therefore, $\mV$ neither accepts nor rejects $\tV$, which is a contradiction, because $\mV$ is complete for $\hV$.
   % Furthermore, for every
%
% \af{Antonis's proof}.
  % Assumption about monitor : 1 reaches a verdict after a finite sequence of observations/actions and then the verdicts are irrevocable.
  %
  % Complete
  %
  % L_a(m) = suffix closed - implied by the fact that verdicts are irrevocable.  (luca)
  %
  % L_b(m) = same.
  %
  % take min -m capture path leading to a verdict.
  %
  % one less
  %
  % L  all (finite) paths that do not lead to a verdict.
  %
  % You construct a tree T from L
  %
  % {Give definitions}
  %
  % Argue that this set must be finite.
  %
  %  from finite L_a(m) and L_r(m) create a finite monitor.
  %
  %  Then use our existing construction.
%
%
\end{proof}

\begin{lem}\label{lem:determinism-and-taus}
	For any deterministic monitor $\mV$, if $\mV \traS{\tau} \mV_1$ and $\mV \traS{\mu} \mV_2$, then $\mu = \tau$ and $\mV_1 = \mV_2$.
\end{lem}

\begin{proof}
	From   \Cref{def:determinism}, if $\mV \traS{\tau}$, then $\mV = \rec x \mVV$ for some $\mVV$, and thus the only transition \mV can perform is $\mV\traS{\tau} \mVV[\mV/X]$.
\end{proof}

\begin{cor}\label{lem:strong-determinism}
	If \mV is deterministic, and $\mV \etraS{\etV} \mV_1$ and $\mV \etraS{\etV} \mV_2$, then $\mV_1 = \mV_2$.
	\qedd
\end{cor}

\begin{cor}\label{lem:strong-determinism2}
	If \mV is deterministic and $\mV \wtraS{} \vV$, and $\mV \wtraS{} \mVV \traS{\act}$, then $\mVV = \vV$.
	\qedd
\end{cor}

\begin{proof}
%	Let $\ftV = \act \ftVV$ and
	Let $k,k' \geq 0$ be such that
	$\mV (\traS{\tau})^k \vV$ and $\mV (\traS{\tau})^{k'} \mVV.$
	If $k'<k$, then by \Cref{lem:strong-determinism},
	$$\mV \  (\traS{\tau})^{k'} \  \mVV \   (\traS{\tau})^{k-k'} \  \vV.$$
	But then, $\mVV \traS{\tau}$ because $k-k'\geq 1$ and $\mVV \traS{\act}$ by our assumptions, and by \Cref{lem:determinism-and-taus}, $\tau = \act$, which is a contradiction.
	Therefore, $k \leq k'$, and by \Cref{lem:strong-determinism},
	$\mV\ (\traS{\tau})^k \ \vV \ (\traS{\tau})^{k'-k} \ \mVV$, and by \Cref{lem:ver-persistence}, $\mVV = \vV$.
%	We prove that $\mVV' = \vV$ and by \Cref{lem:ver-persistence}, this is enough to complete the proof.
\end{proof}

\begin{cor}\label{lem:strong-determinism3}
	If \mV is deterministic and $\mV \wtraS{} \vV$, and $\mV \wtraS{\ftV} \mVV$, where $\ftV \neq \varepsilon$, then $\mVV = \vV$.
	\qedd
\end{cor}

\begin{proof}
	A consequence of \Cref{lem:strong-determinism2,lem:ver-persistence}.
\end{proof}

\input{app-pspace-hardness.tex}

%% file: app-pspace-hardness.tex
% !TEX root = main.tex

\subsection{The PSPACE-hardness of \ltmuS}
\label{subsec:pspace-hardness}

Here we prove that satisfiability for \ltmuS is PSPACE-hard. 
The reduction that we use is from the one-variable, diamond-free fragment of $D \oplus_\subseteq K4$, 
which is PSPACE-complete \cite{achilleos16}. 

$D \oplus_\subseteq K4$ is a modal logic with two modalities, $[1]$ and $[2]$, based on a serial transition relation 
$\traS{1}$ (\ie $\forall x \exists y.x \traS{1} y$) and a  transitive transition relation $\traS{2}$ (\ie $\forall x,y,z.(x \traS{2} y \traS{2} z \Rightarrow x \traS{2} z)$), such that $\traS{1} \subseteq \traS{2}$.
Given a set of propositional variables $\textsf{Prop}$,
$D \oplus_\subseteq K4$ is interpreted over Kripke structures of the form $(W,(\tra{\act})_{\act \in \{1,2\}},V)$, where $W$ is a non-empty set of states, the transition relations satisfy the above-mentioned properties, and $V: W \to 2^{\textsf{Prop}}$ maps states to sets of propositional variables.
Here, we focus on the one-variable, diamond-free fragment of $D \oplus_\subseteq K4$, and therefore $\textsf{Prop} = \{p\}$ and the syntax of the fragment is
given by the following grammar:
\[ \varphi, \psi ::= p ~~~\mid~~~ \neg p ~~~\mid~~~ \varphi \vee \psi ~~~\mid~~~ \varphi \land \psi ~~~\mid~~~ [1]\varphi ~~~\mid~~~ [2] \varphi  .\]
The semantics of $D \oplus_\subseteq K4$ is defined in terms of a satisfaction relation $\models$, 
where $M,w \models \varphi$ means that $\varphi$ is satisfied at state $w$ of $M$, 
in a similar way to the branching-time semantics $\hSemB{-}$ for \recHML, with the additional condition that $ M,w \models p $ if and only if $p \in V(w)$. 
Constants $\hTru,\hFls$ can be either included to the syntax or constructed as $p \vee \neg p$ and $p \land \neg p$, respectively.

Let $\acta,\actb \in \Act$, where $\acta \neq \actb$.
We can define the mapping from formulae without diamond modalities and that use only one propositional variable $p$ that maps $\hV$ to
$$trans(\hV) = \max X.(\hSuf{\acta}X \vee \hSuf{\actb}{\hSuf{\acta}X}) \land \hV^t , $$ where
$\hV^t$ is such that $p^t = \hSuf{\actb}\hTru$, it commutes with the boolean operators, and 
$$([1] \hVV)^t = \max X.([\actb] X \land [\acta]\hVV^t) \qquad \text{ and } \qquad 
  ([2] \hVV)^t = \max X.([\actb] X \land [\acta] X \land [\acta]\hVV^t).$$
The construction of $trans(\hV)$ ensures that it can only be satisfied by traces of the form $(\acta + \actb\acta)^\omega$.
In such a trace, a following $\actb$ action marks the satisfaction of propositional variable $p$, so states-transitions are 
represented by $\acta$ actions. As such, in the translation above, we use greatest fixed points (least fixed points would have worked too) to allow the modalities for $1$ to skip any occurrences of $\actb$ and only be affected by the occurrences of $\acta$.
The transition relation for $2$ can simply be the transitive closure of the one for $1$, and therefore in the translation, the $2$ modalities are allowed to skip any finite prefix and activate right after any $\acta$ occurrence.

\begin{lem}\label{lem:reduction-pspace}
For every formula $\hV$ from the one-variable, diamond-free fragment of $D \oplus_\subseteq K4$, 
$\hV$ is satisfiable if and only if $trans(\hV)$ is satisfiable over \Trc.
\end{lem}	
\begin{proof}
Given a trace $\tV \in (\acta + \actb\acta)^\omega$,
let $M_\tV = (W,(\tra{\act})_{\act \in \{1,2\}},V)$ be a Kripke structure, where 
$W$ is the set of finite prefixes of $\tV$ that do not end with $\actb$, 
$\ftr \traS{1} \ftrr$ iff $\ftrr = \ftr\acta $ or $\ftrr =  \ftr\actb\acta$, and 
$V(p) = \{ \ftr \in W \mid 
\ftr \actb \acta \in W
% \ftr \actb \text{ is a prefix of } \tV 
\}$.

Given a Kripke structure $M = (W,(\tra{\act})_{\act \in \{1,2\}},V)$ and state $w \in W$, fix a path $w_0w_1w_2 \cdots$ in $M$, where $w = w_0$; $\tV_{M,w} = \ftr^{W,w}_0\ftr^{W,w}_1\ftr^{W,w}_2\cdots$, where if $w_i \in V(p)$, then $\ftr^{W,w}_i = \actb\acta$, and $\ftr^{W,w}_i = \acta$ otherwise.

It is not too hard to observe that the following hold for all $\varphi$ of $D \oplus_\subseteq K4$, $\tV$, $M$, and $w$:
\begin{enumerate}
\item $\tV \in \hSemL{trans(\hV)}$ if and only if $\tV$ is of the form $(\acta + \actb\acta)^\omega$ and $\tV \in \hSemL{\hV^t}$, by the definition of $trans(\hV)$;
\item $M,w \models p$ if and only if $\tV_{M,w} \in \hSemL{\hSuf{\actb}\hTru}$, by the construction of $\tV_{M,w}$;
\item $\tV \in  (\acta + \actb\acta)^\omega \cap \hSemL{\hSuf{\actb}\hTru}$  if and only if $M_\tV,\varepsilon \models p$;
\item $M,w \models [1] \hVV$ if and only if $\tV_{M,w} \in \hSemL{\max X.([\actb] X \land [\act]\hVV^t)}$;
\item $\tV \in  (\acta + \actb\acta)^\omega \cap \hSemL{\max X.([\actb] X \land [\act]\hVV^t)}$  if and only if $M_\tV,\varepsilon \models [1]\hVV$;
\item $M,w \models [2] \hVV$ if and only if $\tV_{M,w} \in \hSemL{\max X.([\actb] X \land [\acta] X \land [\acta]\hVV^t)}$;
\item $\tV \in  (\acta + \actb\acta)^\omega \cap \hSemL{\max X.([\actb] X \land [\acta] X \land [\acta]\hVV^t)}$ if and only if $M_\tV,\varepsilon \models [2]\hVV$.
\end{enumerate}
From these observations, it is not hard to conclude, by induction on $\hV$, that for every $\hV$ of $D \oplus_\subseteq K4$, if $M,w \models \varphi$, then $\tV_{M,w} \in \hSemL{trans(\varphi)}$.
Furthermore,  if $\tV \in \hSemL{trans(\varphi)}$, then by the first observation, $\tV \in (\acta + \acta\actb)^\omega$ and 
$\tV \in \hSemL{\hV^t}$.
Therefore, it suffices to prove that for all subformulae $\hVV$ of $\hV$ and $\ftr \trr = \tr$, where $\ftr$ does not end with $\actb$,
if $\trr \in \hSemL{\hVV^t}$, then $M, \ftr \models \hVV$.
This can be done by induction on $\hVV$, using the observations above.
\end{proof}

% $M_\tV,\varepsilon \models \varphi^t$, so it suffices to prove that 
% \begin{itemize}
% \item ; and
% \item if $\tV \in \hSemL{trans(\varphi)}$, then $\tV \in (\acta + \acta\actb)^\omega$ and $M_\tV,\varepsilon \models \varphi$
% \end{itemize}
% is satisfiable if and only if $trans(\hV)$ is satisfiable.

Then, the PSPACE-hardness of \ltmuS follows as a corollary of Lemma \ref{lem:reduction-pspace}.

%% file: tight-complete-monitorability-appendix.tex
% !TEX root = main.tex

In the following, we use the notations $\bigodot_{\act \in A} \mV_\act$ and $\mV_1 \paralG \cdots \paralG \mV_i  \paralG \cdots \paralG \mV_k$  to denote a combination of
monitors
using the parallel operator $\paralG$
%
% In the context that the notation is used,
since the particular way the monitors are combined does not matter.
Furthermore, since we are dealing with reactive monitors---and, as a consequence of \Cref{prop:monitor2automaton}, the parallel operators are associative with respect to verdict-equivalence---any way we combine the monitors with $\paralG$ will reach the same verdict for the same (finite) trace.

\begin{lemmabiss}{\ref{lem:slim_violations}}
If $\hV \in \HML$ is slim and
$\hSemL{\hV} = \emptyset$
% for every trace $\tr$, $\tr\notin \hmeaning{\hV}_L$
(\resp $\hSemL{\hV} = \Trc$), then $\hV = \hFls$ (\resp $\hV = \hTru$).
\end{lemmabiss}
\begin{proof}
We prove the contrapositive statement, that if $\hV \neq \hFls$, then there is some $\tV \in \hSemL{\hV}$.
The proof is by induction on $\hV$. The case for $\hTru$
% and $\hFls$ are
is
immediate.
If $\hV \equiv \bigwedge_{\act \in A} \hNec{\act}\hV_\act$, then we have two cases.
% ; then, we have the following cases:
\begin{description}
\item [Case $A = \Act$:]
% If $\hV \equiv \bigwedge_{\act \in A} \hNec{\act}\hV_\act$, then we have two sub-cases.
% If $A = \Act$,
Then there must be some $\act \in A$, such that $\hV_\act \neq \hFls$. By the inductive hypothesis, there is some $\tV \in \hSemL{\hV_\act}$, and therefore, $\act \tV \in \hSemL{\hV}$, which completes the proof.
% and $\forall \tr. \tr \notin \hmeaning{\hV}_L$,
% then for any $\actb \in \Act$, $\forall \tr. \actb\tV \notin \hmeaning{\hV}_L$.
\item [Case $A \neq \Act$:] Then there is some
% $\actb \in A$, then $\hV_\act \equiv \hFls$; if
$\act \notin A$,
and therefore, $\act^\omega \in \hSemL{\hV}$, which
% , again,
completes the proof.
% then for any continuation $\tr$, $\actb \tr \in \hmeaning{\hV}$, which is a contradiction. Therefore, $A = \Act$ and by the inductive hypothesis, for every $\act \in A$, $\hV_\act = \hFls$, and therefore, we have reached another contradiction, as $\hV = [\Act] \hFls$, so it is not slim.
% \item
% If $\hV \equiv \bigwedge_{\act \in A} \hNec{\act}\hV_\act$ and $\forall \tr. \tr \in \hmeaning{\hV}_L$,
% then for any $\actb \in \Act$, $\forall \tr. \actb\tV \in \hmeaning{\hV}_L$.
% % If $\actb \in A$, then for any continuation $\tr$, $\actb \tr \in \hmeaning{\hV}$, which is a contradiction.
% % Therefore, $A = \Act$ and by the inductive hypothesis, for every $\act \in A$, $\hV_\act = \hFls$, and therefore, we have reached another contradiction, as $\hV = [\Act] \hFls$, so it is not slim.
% If $\actb \in A$, then for any continuation $\tr$, $\actb \tr \in \hmeaning{\hV}$, and thus $\hV_\actb \equiv \hTru$, so
% by the inductive hypothesis, $\hV_\actb = \hTru$. Thus,
% $\hV$ is not slim, which is a contradiction.
% Therefore, $A = \emptyset$, so $\hV = \hTru$ \ac{or we can say that $\hV = \hTru$, but I'm not sure what conventions make sense here}.
% \item The cases for $\hV \equiv \bigvee_{\act \in A} \hSuf{\act} \hV_\act$ are symmetric.
% \qedhere
\end{description}
If $\hV \equiv \bigvee_{\act \in A} \hSuf{\act} \hV_\act$, then $A \neq \emptyset$. Let $\act \in A$. Since $\hV$ is slim, $\hV_\act \neq \hFls$ and by the inductive hypothesis there is some $\tV \in \hSemL{\hV_\act}$. Therefore, $\act \tV \in \hSemL{\hV}$, which completes the proof.
\end{proof}

\begin{lemmabiss}{\ref{lem:slim-implies-tight}}
If $\hV$ is a slim HML formula, then $m(\hV)$ is tight. % monitor for $\hV$.
\end{lemmabiss}
\begin{proof}
  By \Cref{prop:hml-monitorable},  $\tV \notin \hmeaning{\hV}_L$ implies that there is a finite prefix $\ftV$  of $\tV$ such that
  $m(\hV)\wtraS{\ftV}\no$.
  We  prove, by induction on \ftV, that
  % if $\mV \not \wtraS{\ftr} \vV$, then there is some $\ftrr$, such that $\mV \wtraS{\ftr\ftrr} \vV' \neq \vV$.
  % This is proven for the case of $\vV = \no$, as the case of $\yes$ is symmetric.
  % The contrapositive statement, is that if
  % for no finite trace $\ftrr$, $\mV \wtraS{\ftr\ftrr} \yes$, then $\mV \wtraS{\ftr} \vV$. This, in turn, due to Proposition \ref{prop:hml-monitorable}, is equivalent to saying that
  if $\forall\tV.
  % $, $
  % \ftr \tr \notin \hmeaning{\varphi}_L$,
  \rej{\mV,\ftV\tV}$,
  then $\mV \wtraS{\ftr} \no$ (the case for acceptance is symmetric).
  % if a finite prefix $\ftr$ violates $\hV$, then $m(\hV) \xRightarrow{\ftr} \no$.
  % The proof is by induction on $\ftr$ and we assume that $\ftr$ is minimal wrt to the prefixes that violate $\hV$.
  If $\ftr = \varepsilon$, then by Lemma \ref{prop:hml-monitorable}, for every trace $\tr$, $\tr\notin \hmeaning{\hV}_L$, thus by Lemma \ref{lem:slim_violations}, $\hV = \hFls$, and therefore $m(\hV) = \no$.
  Since $\no \wtraS{} \no$, we are done.
  If $\ftr = \actb \ftrr$, if $\hV \neq \hTru,\hFls$, then we have two cases:
  \begin{itemize}
  \item If $\hV \equiv \bigwedge_{\act \in B} \hNec{\act}\hV_\act$, then $\actb \in B$, $\ftrr \notin \hmeaning{\hV_\actb}_L$, and
  $$m(\hV) = \bigotimes_{\act \in B}(\act.m(\hV_\act) + \overline{\act}.\yes).$$
  Therefore, using the inductive hypothesis and rule \rtit{mVrC1},
  $$m(\hV) \xrightarrow{\actb} \yes \paralC \cdots \paralC \yes \paralC   m(\hV_\actb) \paralC  \yes \paralC \cdots \paralC \yes \Rightarrow m(\hV_\actb)\xRightarrow{\ftrr} \no.$$
  \item
  If $\hV \equiv \bigvee_{\act \in D} \hSuf{\act}\hV_\act$, then
  $$m(\hV) = \bigoplus_{\act \in D}(\act.m(\hV_\act) + \overline{\act}.\no).$$
  If $\actb \notin D$, then $m(\hV) \xrightarrow{\actb} \bigoplus_{\act\in D} \no \Rightarrow \no \wtraS{\ftrr}\no$.
  If  $\actb \in D$, then,
  using the inductive hypothesis and rule \rtit{mVrD1},
  $$m(\hV) \xrightarrow{\actb} \no \paralD \cdots \paralD \no \paralD   m(\hV_\actb) \paralD  \no \paralD \cdots \paralD \no \Rightarrow m(\hV_\actb)\xRightarrow{\ftrr} \no,$$
  and the proof is complete.
  \qedhere
  \end{itemize}
\end{proof}

\begin{propbiss}{\ref{prop:substitute-to-slim}}{\HML normalisation}
For every formula $\hV \in \HML$, there exists $k \leq l(\hV)$ such that $\hV = \hV_0 \rewriteL \hV_1 \rewriteL \ldots \rewriteL \hV_k=\hVV$ where $\hVV$ is slim and $\hSemL{\hV} = \hSemL{\hVV}$.
% % that uses $k$ symbols,
% we can construct an equivalent slim formula by
% repeatedly substituting a subformula that has the form of one of the formulas
% on the left-hand-side of one of the equivalences in \Cref{fig:fun-excercise} with
% the corresponding right-hand side of the equivalence.
% Furthermore, these substitutions need only be applied at most
% % some of the subformulas of $\hV$
% % up to
% $l(\hV)$ times
% to reach a slim formula.
%where $k$ is the length of $\hV$ as a string  of symbols.
\end{propbiss}

\begin{proof}
We observe that if $\hV$ is not slim, then one of its subformulas is not in a form that can be produced by the grammar of \Cref{def:slim}, and therefore it must have the form of one of the left-hand-side formulas from \Cref{fig:fun-excercise}.
Therefore, for the proposition it suffices to prove for each of these equivalences that it is sound and that the
left-hand-side has a smaller length than the right-hand-side, which ensures that the rewriting of formulae terminates after at most $l(\hV)$ substitutions.
The cases for \Cref{eq:modal-trivial,eq:absorb-trivial} are immediate.
The remaining cases are also not that hard to handle, and we describe the representative case of \Cref{eq:diamond-and-box}.

%For a formula $\hVV$, let  $l(\hVV)$ be the length of $\hVV$ as a string of symbols.
We can observe that $\bigwedge_{\act \in C} \chi_\act$ and $\bigvee_{\act \in C} \chi_\act$ represent respectively a sequence of $|C|$ conjunctions and  disjunctions.
Therefore, $l(\bigvee_{\act \in C} \chi_\act) = |C| - 1 + \sum_{\act \in C} l(\chi_\act)$.
As such,
% I have no explanation why these align this well:
\[
l\left(\bigwedge_{\act \in A}\hNec{\act}\hV_\act \land \bigvee_{\act \in B}\hSuf{\act}\hVV_\act \right) ~~~=
|A|+|B|-2 +   \sum_{\act \in A}(1 + l(\hV_\act)) + 1 + \sum_{\act \in B}(1 + l(\hVV_\act)), \text{ and} \]
\begin{align*}
l\left(\bigvee_{\act \in A \cap B} \hSuf{\act} (\hV_\acta \land \hVV_\acta) \lor \bigvee_{\act \in  B \setminus A} \hSuf{\act} \hVV_\acta\right) ~~~=
\qquad\qquad\qquad\qquad\qquad\qquad \qquad\qquad\qquad
% \hfill
\\
=~~~
|A \cap B|+|B \setminus A|-2 +   \sum_{\act \in A \cap B}(2 + l(\hV_\act) + l(\hVV_\act)) + 1 + \sum_{\act \in B \setminus A}(1 + l(\hVV_\act)) ~~~ \\
=~~~
|B|-2 +   \sum_{\act \in A \cap B}(1 + l(\hV_\act)) + 1 + \sum_{\act \in B }(1 + l(\hVV_\act)) ~~~ \\
<~~~
|A|+|B|-2 +   \sum_{\act \in A}(1 + l(\hV_\act)) + 1 + \sum_{\act \in B}(1 + l(\hVV_\act)),
\end{align*}
because $A \neq \emptyset$.
To prove that \Cref{eq:diamond-and-box} is sound, we observe that
$\tV \in \hSemL{\bigwedge_{\act \in A}\hNec{\act}\hV_\act \land \bigvee_{\act \in B}\hSuf{\act}\hVV_\act}$
if and only if
$\tV \in \hSemL{\bigwedge_{\act \in A}\hNec{\act}\hV_\act}$ and $\tV \in \hSem{\bigvee_{\act \in B}\hSuf{\act}\hVV_\act}$
if and only if
$\tV = \act \tVV$ and
$\act \in A \Rightarrow \tVV\in \hSemL{\hV_\act}$ and $\act \in B$ and $\tVV \in \hSemL{\hSuf{\act}\hVV_\act}$
if and only if
$\tV = \act \tVV$ and
% there is some
$\act \in B \setminus A$ and $\tVV \in \hSemL{\hSuf{\act}\hVV_\act}$,
or
$\act \in A \cap B$ and $\tVV \in \hSemL{\hSuf{\act}\hVV_\act}$
if and only if
$\tV \in \hSemL{ \bigvee_{\act \in A \cap B} \hSuf{\act} (\hV_\acta \land \hVV_\acta) \lor \bigvee_{\act \in  B \setminus A} \hSuf{\act} \hVV_\acta }$.
% , and this completes the proof.
%
% |A|+|B|-1 +   \sum_{\act \in A}(1 + l(\hV_\act)) + \sum_{\act \in B}(1 + l(\hVV_\act)) \leq
% |B| +  \sum_{\act \in A}(1 + l(\hV_\act)) + \sum_{\act \in B \setminus A}(1 + l(\hVV_\act)) + \sum_{\act \in A \cap B}(1 + l(\hVV_\act)) \leq
% |B| +  \sum_{\act \in B \setminus A}(1 + l(\hVV_\act)) + \sum_{\act \in A \cap B}(2 + l(\hV_\act) + l(\hVV_\act)) \leq
%
% \]
% \bigvee_{\act \in A \cap B} \hSuf{\act} (\hV_\acta \land \hVV_\acta) \lor \bigvee_{\act \in  B \setminus A} \hSuf{\act} \hVV_\acta
\end{proof}

%% file: partial-complete-monitorability-appendix.tex
% !TEX root = main.tex

A useful measure for guarded formulae is $ms(\hV)$ that measures the longest distance from the root of the syntax tree of $\hV$ to either a constant $\hTru,\hFls$, or to a modality.

\begin{defn}[Measure for guarded \recHML formulae] \label{def:formula-measure}
  \begin{align*}
    & ms(\hNec{A}\hV) = ms(\hSuf{A}\hV) = ms(\hTru) = ms(\hFls) = 0\\
    & ms(\max X. \hV) = ms(\min X. \hV) = ms(\hV) + 1\\
    & ms(\hV\land \hVV) = ms(\hV \lor \hVV) = \max\{ ms(\hV), ms(\hVV) \} + 1 \tag*{\qedd}
  \end{align*}
\end{defn}
% It is defined recursively thus: $ms(\hNec{A}\hV) = ms(\hSuf{A}\hV) = ms(\hTru) = ms(\hFls) = 0$; $ms(\max X. \hV) = ms(\min X. \hV) = ms(\hV) + 1$; and $ms(\hV\land \hVV) = ms(\hV \lor \hVV) = \max\{ ms(\hV), ms(\hVV) \} + 1$.
% We will be using $ms(-)$ in the proofs that follow.

\begin{propbis}{\ref{prop:partial-reactive}}
	For any $\hV \in \ltmuS \cup \ltmuC$, $\hSyn{\hV}$ is reactive.
\end{propbis}

\begin{proof}
	The proof is by straightforward induction on $ms(\hV)$, using the fact that
	\hV is guarded, and that therefore for the case of $\hV = \max X.\hVV$, $ms(\hVV[\hV/X]) < ms(\hV)$.
\end{proof}

\begin{propbis}{\ref{lem:sound-and-rej-compl}}
For every $\hV\in \ltmuS$, $\hSyn{\psi}$ is a sound and violation-complete monitor for $\hV$. For every  $\hV\in \ltmuC$, $\hSyn{\psi}$ is a sound and satisfaction-complete monitor for $\hV$.
\end{propbis}
\begin{proof}

We prove the lemma for the $\ltmuS$ fragment; the proof for $\ltmuC$ is dual.
For \emph{soundness}, we need to show that if \rej{\hSyn{\hV},\tV} (\resp \acc{\hSyn{\hV},\tV}) then $t \notin \hSem{\hV}$ (\resp $t \in \hSemL{\hV}$).
We here show the case for rejection; the case for acceptance is symmetric.

If $\rej{\hSyn{\hV},\tV}$ then there is an explicit trace $\etV$ that agrees with a prefix of $\tV$ on external actions, such that $\hSyn{\hV} \etraS{\etV} \no$.
By structural induction on $\etV$, we prove that for every $\etV$ and $\tV$, if $\hSyn{\hV} \etraS{\etV} \no$ and $\filter{\etV}$ is a prefix of \tV, then $\tV \notin \hSemL{\hV}$.
 \begin{description}
    \item[Case $\etV = \varepsilon$:]
 $\hSyn{\hV} = \no$ and thus, $\hV = \hFls$ where $\tV \notin \hSemL{\hFls}$ holds trivially.
  \item[Case $\etV = \actu \etV'$:]
  We take cases for $\hV$.
  Since $\hV$ is closed, we do not consider the case for $\hV=X$.
  \begin{description}
  \item[Case $\hV = \hFls$:] Immediate.
  \item[Case $\hV = \hTru$:] We have $\hSyn{\hTru} = \yes$  and \Cref{lem:ver-persistence} ensures that the premise $\hSyn{\hTru}=\yes \etraS{\etV} \no$ cannot ever hold.
  \item[Case $\hV = \hNec{\ASet}{\hV}$:] By \Cref{def:formula-to-monitor-part}, we have
  $\hSyn{\hNec{\ASet}{\hV}} = \ch{\prf{\ASet}{\hSyn{\hV}}}{\uprf{\ASet}{\yes}}$.
  This monitor cannot take a $\tau$-transition, so it must be the case that $\actu = \acta$.
  Therefore , it must be the case that $\tV = \act \tV'$ where $\acta \in \ASet$  and $\hSyn{\hNec{\ASet}{\hV}} \traS{\act} \hSyn{\hV} \etraS{\etV'} \no$ for some $\etV'$ such that $\filter{\etV'}$  is a prefix of $\tV'$.
  From the IH,  $\tV' \notin \hSemL{\hV}$, and thus, by $\acta \in \ASet$, we obtain $\tV \notin \hSemL{\hNec{\ASet}{\hV}}$.
  \item[Case $\hV = \hSuf{\ASet}{\hV}$:]
  By \Cref{def:formula-to-monitor-part}, we have $\hSyn{\hSuf{\ASet}{\hV}} =\ch{\prf{\ASet}{\hSyn{\hV}}}{\uprf{\ASet}{\no}}$.
  Similar to the previous case,  if $\hSyn{\hSuf \ASet \hV} \etraS{\actu\etV'} \no$, then $\mu = \act$ and $t = \act t'$ where
\begin{itemize}
  \item either $\act\notin \ASet$, which immediately gives us $t \notin \hSemL{\hSuf{\ASet}{\hV}}$;
  \item or $\act\in \ASet$ where $\hSyn \hV\etraS{\etV'}\no$.  By the IH, we obtain $t' \notin \hSemL{{\hV}}$ and thus $t \notin \hSemL{\hSuf{\ASet}{\hV}}$.
\end{itemize}
     \item[Case $\hV = \hMaxX \hVV$:]
  $\hSyn \hV = \rec x \hSyn \hVV$ so $\etV = \tau \etV'$ and $\hSyn \hV \xrightarrow{\tau} \hSyn{\hVV} [\rec x \hSyn \hVV/x] \etraS{\etV'} \no$. Noting that $\hSyn{\hVV}[\rec x \hSyn{\hVV}/x]=\hSyn{\hVV[\max X \hVV/X]}$, and
  since $\filter{\etV'}$ is a prefix of $\tV$,
%  agrees with $\tV$ on external actions,
  by the inductive hypothesis, $\tV \notin \hSemL{\hVV[\hMaxX \hVV /X]} = \hSemL{\hV}$.
  \item[Cases \hAndF and \hOrF:]
We proceed by induction on the number of boolean connectives in the formula. If $\hV$ has no boolean connectives, this is handled by one of the previous cases.
If $\hV=\hVV_1 \land \hVV_2$ then $\hSyn{\hV}=\hSyn{\hVV_1} \paralC \hSyn{\hVV_2} $. From \cref{cor:parallel-and-monitors}, $\rej{\hSyn{\hVV_1}\paralC\hSyn{\hVV_2},t}$ if and only if either $\rej{\hSyn{\hVV_1},t}$ or $\rej{\hSyn{\hVV_2},t}$. By the inductive hypothesis, this is the case only if $t\notin \hSemL{\hVV_1}$ or $t \notin \hSemL{\hVV_2}$, and therefore $t\notin \hSemL{\hV}$. The case for $\lor$ is similar.
  \end{description}
\end{description}

For \emph{completeness:} we need to show that if  $t \notin \hSemL{\hV}$ (\resp $t \in \hSemL{\hV}$) then \rej{\hSyn{\hV},\tV} (\resp \acc{\hSyn{\hV},\tV}). Again, we prove the case for rejection since the case for acceptance is symmetric.

%%%%%%%%%%%%%%%   NEW APPROACH. OLD APPROACH STILL EXISTS AS BACKUP   %%%%%%%%%%%%%%%%%

Since $\hV$ is closed, we can assume that each formula variable $X$ appears in the scope of a \emph{unique}
greatest-fixed-point operator $\max X$.
We assume a mapping $un(-)$ of variables that appear in $\hV$ to subformulae of $\hV$, such that
for every $X$, $un(X) = \max X.\hVV$ for some $\hVV$.
We can extend the definition of monitor synthesis from \Cref{def:formula-to-monitor-part} to also apply on pairs $(\hVV,S)$, where $\hVV \in \ltmuS$ and $S$ is a set of formula variables that appear in $\hV$, by altering the case for $X$, so that
\begin{align*}
  \hSyn{X,S} =&
   \begin{cases}
     \hSyn{un(X),S\setminus \sset{X}} & \text{ if } X \in S\\
     x & \text{ otherwise }
   \end{cases}
\end{align*}
We observe that $\hSyn{\hVV,\emptyset} = \hSyn{\hVV}$. Therefore, to complete the proof of completeness, it suffices to prove that for all (possibly open) subformulae of $\hV$, if $S$ is the set of free variables in $\hVV$ and
$t \notin \hSemL{\hVV,\rho}$ for some environment $\rho$ such that for all $X \in S$, $\rho(X)$ is the set of traces that $\hSyn{X,S}$ does not reject,
then $\hSyn{\hVV,S}$ rejects $t$.
We proceed to prove this claim by induction on $\hVV$.

%We define a monitor environment $\sigma$ to be a mapping from monitor-variables to
%  (possibly open) monitors.
%  Similarly to the case of a formula environment, $\sigma[x \mapsto m]$ is the monitor environment that agrees with $\sigma$ on all variables, except, possibly on $x$, and $\sigma[x \mapsto m](x) = m$. We say that a monitor $m$ rejects a trace $t$ under an monitor environment $\sigma$ when $m$ rejects $t$, using the additional rule $\frac{\sigma(x) \traS{\mu} n}{x \traS{\mu} n}$.
% Recall that without monitor environment $x$ cannot transition, so under $\sigma$, $x$ and $\sigma(x)$ have exactly the same transitions. Furthermore, if $x$ is not free in $m$, then $m$ rejects the same traces, both under $\sigma$ and under $\sigma[x \mapsto n]$, for every monitor $n$.
%
%
%  We prove for every (possibly open) formula $\hV\in \ltmuS$, environment $\rho$ and monitor environment $\sigma$ such that for all $X$ free in $\hV$
%  $\rho(X)$ has exactly the traces that $\sigma(x)$ does not reject under $\sigma$ that:
%  $t \notin \hSemL{\psi,\rho}$ implies that
%  $m(\hV)$ rejects $t$ under $\sigma$.
%  This claim implies completeness.
%
%    The
%  claim
%  is proven by induction on $\hV$.

\begin{description}
\item[cases $\hVV \in \{X, \hFls, \hTru\}$:] immediate.
\item[case $\hVV=\hNec \ASet \hVV'$:]
Note that $t \notin \hSemL{\hNec \ASet \hVV',\rho}$
if and only if $\tV = \act \tV'$, $\act\in \ASet$, and $\tV' \notin \hSemL{\hVV',\rho}$.
By the IH,  $\tV' \notin \hSemL{\hVV',\rho}$ implies that $\hSyn{\hVV',S}$ rejects $\tV'$.
As a result, the monitor $\hSyn{\hNec \ASet \hVV',S} = \ch{\prf{\ASet}{\hSyn{\hVV',S}}}{\uprf{\ASet}{\yes}}$ rejects $\tV$.
\item[case $\hVV=\hSuf \ASet \hVV'$:]
 $t \notin \hSemL{\hSuf{\ASet}{\hVV'},\rho}$ if and only if either $\tV = \act \tV'$ and $\act\notin \ASet$, or else $\tV' \notin \hSemL{{\hVV'},\rho}$.
 In the former case $\hSyn{\hVV,S}$ clearly rejects $\tV$; in the later case, by the IH we know that $\hSyn{\hVV',S}$ rejects $\tV'$, in which case $\hSyn{\hSuf \ASet \hVV',S} = \ch{\prf{\ASet}{\hSyn{\hVV',S}}}{\uprf{\ASet}{\no}}$ rejects $\tV$.
 \item[cases $\hVV= \hVV_1\vee \hVV_2$ and $\hVV_1 \land \hVV_2$:] We proceed by induction on the number of boolean connectives.  The case without boolean connectives is handled by one of the previous cases.
 If $\hVV=\hVV_1\wedge \hVV_2$, then $\hSyn{\hVV}=\hSyn{\hVV_1} \paralC \hSyn{\hVV_2} $.
 If $\tV\notin \hSemL{\hVV}$, then it must be the case that either
 $\tV\notin \hSemL{\hVV_1}$ or $\tV\notin \hSemL{\hVV_2}$.
 By the IH we obtain either $\rej{\hSyn{\hVV_1},t}$ or $\rej{\hSyn{\hVV_2},t}$.
 Therefore, from \Cref{prop:partial-reactive,cor:parallel-and-monitors}, $\rej{\hSyn{\hVV'_1,S}\paralC\hSyn{\hVV'_2,S},\tV}$.
 The disjunctive case is similar.

 \item[case $\hVV = \hMaxX \hVV'$:]
   From \Cref{fig:recHML}, $\tV \in \hSemL{\hVV,\rho}$ if and only if there is some set of traces $T$, such that $\tV \in T$ and $T \subseteq \hSemL{\hVV',\rho[X \mapsto T]}$.
   Let $T$ be the set of traces not rejected by $\hSyn{\hVV,S}$.
  By the IH, for every trace $\tV'$,
  \[
  \text{if } \tV' \notin \hSemL{\hVV',\rho[X \mapsto T]} \text{ then }
   \hSyn{\hVV',S \cup \{X\}} \text{ rejects } \tV' .
%   \text{ under } \sigma[x \mapsto \hSyn{\hVV}].
  \]
By \Cref{def:formula-to-monitor-part}, we have $\hSyn{\hMaxX \hVV', S} = \rec{x}{\hSyn{\hVV',S}}$
where every transition sequence must therefore start as $\rec{x}{\hSyn{\hVV',S}} \traS{\actt} \hSyn{\hVV',S}\subS{\hSyn{\hVV,S}}{x}$.
We also
%Then,
%if $\tV' \notin \hSemL{\hVV',\rho[X \mapsto T]}$, we also
have that $\hSyn{\hVV',S}\subS{\hSyn{\hVV,S}}{x} = \hSyn{\hVV',S \cup \{X\}}$.
This means that, if $\tV' \notin \hSemL{\hVV',\rho[X \mapsto T]}$,
then
$\hSyn{\hVV',S}\subS{\hSyn{\hVV,S}}{x}$
rejects $\tV'$, which in turn yields that $\hSyn{\hVV,S}$ rejects $\tV'$, which is the result that we want.
%Then, as $x$ is not free in $\hSyn{\hVV',S \cup \{X\}}[\hSyn{\hVV}/x]$, it follows that $\hSyn{\hVV'}[\hSyn{\hVV}/x]$  -- that is $\hSyn{\hVV'}[\rec x \hSyn{\hVV'}/x]$ -- rejects $\tV'$ under $\sigma$.
%Then $\rec x \hSyn{\hVV'}$, which has a $\tau$-transition to $\hSyn{\hVV'}[\rec x \hSyn{\hVV'}/x]$, also rejects $\tV'$ under $\sigma$. Then, finally, $m(\hVV)$ rejects $\tV'$ under $\sigma$.
%
%
Therefore, for every trace $\tV'$,
  \[
  \text{if }  \tV' \notin \hSemL{\hVV',\rho[X \mapsto T]}, \text{ then }
  m(\hVV,S ) \text{ rejects } \tV'
  ,
  \]
  hence $T \subseteq \hSemL{\hVV',\rho[X \mapsto T]}$.
  % }
  If $m(\hVV,S)$ does not reject $\tV$, then $\tV \in T$ and thus, $\tV \in \hSemL{\psi,\rho}$; in other words, if $\tV \notin \hSemL{\psi,\rho}$, then $m(\hVV,S)$ rejects $\tV$.
  \qedhere
\end{description}

\end{proof}

The following definition, \Cref{def:introducing-redundancies}, lead up to \Cref{lem:reject-same}, which handles the discrepancies between our synthesis functions.

% We formalize this modification in the following definition.

\newcommand{\red}{\textsf{red}}

\begin{defn}\label{def:introducing-redundancies}
	For parallel monitor \mV, we define $\red(\mV)$ recursively on $\mV$, such that
	$\red(\mV) = \mV$ for $\mV = \yes,\no,x$, it commutes with the parallel composition operators,
	$\red(\mV + \mVV) = \red(\mV) \paralC \red(\mVV)$, $\red(\stp) = \yes$, and
	$\red({\prf{\ASet}\mV}) = {\prf{\ASet}\red(\mV) + \prf{\coASet}\yes}$.
%	\begin{align*}
%	\red(\no) &= \no %\\
%	&\red(\stp) &= \stp
%	&\red(\yes) &= \yes
%	&\red(x) &= x
%	  \\
%	\red(\recX{\mV}) &= \recX{\red(\mV)}%\\
%	&\red({\esel{\mV}{\mVV}}) &= \hAnd{f(\mV)}{\red({\mVV})} %\\
%	&\red({m\paralC n}) &= \hAnd{\red({\mV})}{\red({\mVV})} %\\
%	&\red({m\paralD n}) &= \hOr{\red({\mV})}{\red({\mVV})} %\\
%	&\red({\prf{\ASet}\mV}) &= \hNec{\ASet}{\red({\mV})}  %\\
%	%  \mSyn{x} &= X
%	\tag*{\qedd}
%	\end{align*}
\qedd
\end{defn}

\begin{lem}\label{lem:substitute}
For every \mV, $\hSyn{\mSyn{\mV}}= \red(\mV)$.
% \mV[\yes/\stp][\paralC/+][\ch{\prf{\ASet}{\mVV}}{\uprf{\ASet}{\yes}}/\prf{\ASet}{\mVV}]$.
\end{lem}

\begin{proof}
	By straightforward induction on \mV.
\end{proof}

\begin{lemmabiss}{\ref{lem:reject-same}}
$\hSyn{\mSyn{\mV}}$
%$\red(\mV)$
rejects the same traces as $\mV$.
\end{lemmabiss}

\begin{proof}
From \cref{lem:substitute}, $\hSyn{\mSyn{\mV}}= \red(\mV)$.
% $\mV[\yes/\stp][\paralC/+][\ch{\prf{\ASet}{\mVV}}{\uprf{\ASet}{\yes}}/\prf{\ASet}{\mVV}]$.
% $\mV[\yes/\stp]$ obviously rejects the same traces as $m$.
 We can decompose $\red(-)$ to three separate operators, $\red_e(-)$, $\red_+(-)$, and $\red_a(-)$ that respectively
 replace $\stp$ with $\yes$, $+$ with $\paralC$, and $A.\mVV$ with $A.\red_a(\mVV) + \coASet.\yes$, and commute with all other monitor operations and leave other constants unchanged.
 We can see that $\red(\mV) = \red_a(\red_+(\red_e(\mV)))$.
 Thus, it suffices to prove that for all $\mV$ and $o \in \{e,+,a\}$,
  $\red_o(\mV)$ rejects the same finite traces as $\mV$.
  But this can be proven by straightforward induction on the finite trace.
%
% $\mV_1 \paralC \mV_2$ rejects if and only if either $\mV_1$ or $\mV_2$ rejects; this is also exactly when $\ch{\mV_1}{\mV_2}$ rejects, so $ \mV[\yes/\stp][\paralC/+]$ rejects the same traces as $\mV$, too. Finally $\ch{\prf{\ASet}{\mVV}}{\uprf{\ASet}{\yes}}$ rejects the same traces as $\prf{\ASet}{\mVV}$, so $\mV[\yes/\stp][\paralC/+][\ch{\prf{\ASet}{\mVV}}{\uprf{\ASet}{\yes}}/\prf{\ASet}{\mVV}]$ rejects the same traces as $\mV$ as well.
\end{proof}

%% file: tight-partial-complete-monitorability-appendix.tex
% !TEX root = main.tex

\begin{lemmabiss}{\ref{lem:partially-complete-tight}}
Let $\mV$ be a deterministic regular monitor, where $\sum_{\act \in \Act} \act.\no$, $\rec x \no$, $\sum_{\act \in \Act} \act.\yes$, and $\rec x \yes$ do not occur as submonitors. Then, $\mV$ is tight.
\end{lemmabiss}

\begin{proof}
Let $\mV$ be a deterministic regular monitor, where $\sum_{\act \in \Act} \act.\no$, $\rec x \no$, $\sum_{\act \in \Act} \act.\yes$, and $\rec x \yes$ do not occur as submonitors, and let \ftV\ be
% the smallest finite trace
such that $\mV$ rejects $\ftV\tV$ for every $\tV$.
We prove that $m \wtraS{\ftV} \no$.
For this, we use the alternative monitor rules
% of System N,
that were
introduced in Subsection \ref{subsec:transformations} and the following auxiliary lemma.
\begin{lem*}[\cite{determinization}]\label{lem:find_px_determ}
	In a transition-sequence $\mV \xRightarrow{\ftV} x$, such that $x$ is bound in $\mV$ and $\mV$ is deterministic, $p_x$
	must appear.
\end{lem*}
Let $\mVV$ be such that $\mV \wtraS{\ftV} \mVV$.
We prove that if $\mVV \neq \no$, then there is some $\tr$ that $\mVV$ does not reject, and this suffices due to Lemma \ref{lem:det_mon_is_det}. % (pending to copy).
%, by induction on $\mV$.
We use
induction on $\mVV$.
If $\mVV$ is a verdict, then the proof is complete.
If $\mVV = x$, then $\mVV$ can only transition to $p_x$; but then, by the lemma, there are some $\ftV_1\ftV_2 = \ftV$, such that $\ftV_2 \neq \varepsilon$ and $\mV \wtraS{\ftV_1} p_x \wtraS{\ftV_2} \mVV = x \traS{\tau} p_x$, and therefore $\mVV$ does not reject $\ftV_2^\omega$.
% so if $\mVV$ rejects all traces, then so does $p_x$, and thus, $\ftV$ is not the smallest trace with this property, a contradiction.
If $\mVV = \sum_{\act \in A} [\act] \mV_\act$, then if $\actb \notin A$, $\mVV$ does not reject $\actb \tr$, therefore we assume that $A = \Act$.
If $\mVV$ rejects all traces, then so do all of $\mV_\act$, and therefore by the inductive hypothesis, $\mVV = \sum_{\act \in \Act} [\act] \no$, which contradicts the lemma's assumptions.
Finally, if $\mVV = \rec x \mVV'$, then by the inductive hypothesis, either $\mVV$ does not reject all traces, or $\mVV' = \no$, which is a contradiction.
\end{proof}

%% file: branching-appendix.tex
% !TEX root = main.tex

\subsection{The Finfinite Domain}

\begin{lemmabiss}{\ref{lem:infinite-to-finfinite-semantics}}
For all $ \hV\in \UHML$, $\hSemF{\hV}\cap \Trc = \hSemL{\hV}$

%\hSemL{\hV} \subseteq \hSemF{\hV}$ and $\Trc \setminus \hSemL{\hV} \subseteq \fTrc \setminus \hSemF{\hV}$.
\end{lemmabiss}

\begin{proof}
Given an environment $\sigma$ on finfinite traces, let $\sigma'$ be the restriction of $\sigma$ on $\Trc$.
Then, we can prove by induction on $\hV$ that $\tr \in \hSemF{\hV,\sigma}$ if and only if $\tr \in \hSemL{\hV,\sigma'}$, for all $\sigma$ and $\tr$.
\end{proof}

\subsection{Monitorability over Finfinite Traces}

\begin{lemmabiss}{\ref{lem:finfinite-trivial-mon}}
Over finfinite traces, if $\mV$ is sound and complete for $\hV$, then $\hV$ is equivalent to either $\hTru$ or to $\hFls$.
\end{lemmabiss}

\begin{proof}
If $\mV$ is complete for $\hV$, then it must either accept or reject $\varepsilon$ and thus 
all of its extensions, that is all finfinite traces. If $\mV$ is also sound, then $\hV$ is equivalent to $\hTru$ or $\hFls$.
\end{proof}

To facilitate some of the proofs to follow, we define $ms(\hV)$ to measure the distance from the root of the syntax tree of $\hV$ to either a constant $\hTru,\hFls$, or to a modality.
\begin{defn}
\begin{align*}
 ms([\act]\hV) = ms(\hSuf{\act}\hV) = ms(\hTru) = ms(\hFls) &= 0 \\
  ms(\max X. \hV) = ms(\min X. \hV) &= ms(\hV) + 1 \\
  ms(\hV\land \hVV) = ms(\hV \lor \hVV) &= \max\{ ms(\hV), ms(\hVV) \} + 1.
\end{align*}
\end{defn}

\begin{lemmabiss}{\ref{lem:finfinite-suffix-closed}}
For all $\ftV \in \Act^*$ and $\fftV \in \fTrc$, if $\hV \in \ftmuS$ and $\ftV\fftV \in \hSemF{\hV}$, then $\ftV \in \hSemF{\hV}$; 
if $\hV \in \ftmuC$ and $\ftV \in \hSemF{\hV}$, then $\ftV\fftV \in \hSemF{\hV}$.
\end{lemmabiss}

\begin{proof}
We prove the lemma for the case of $\hV \in \ftmuS$, as the case of $\hV \in \ftmuC$ is dual, and, as usual, we assume that $\hV$ is a guarded formula.
We use induction on 
% $\ftV$ and then on 
$ms(\hV) + |\ftV|$.
Let $\ftV\fftV \in \hSemF{\hV}$.
We proceed by a case analysis
on the form of $\hV$. The interesting cases are the ones for $\hV = [\ASet] \hVV$ and $\hV = \max X.\hVV$.
If $\hV = [\ASet] \hVV$, then, if $\ftV = \act \ftV'$ for some $\act \in \ASet$, then $\ftV\fftV = \act \ftV'\fftV$, so it must be the case that $\ftV'\fftV \in \hSemF{\hVV}$, 
therefore by the inductive hypothesis, $\ftV' \in \hSemF{\hVV}$, so $\ftV \in \hSemF{\hV}$; otherwise, immediately by the finfinite semantics, $\ftV \in \hSemF{\hV}$.
If $\hV = \max X.\hVV$, then $\hSemF{\hV} = \hSemF{\hVV[\hV/X]}$ and since $\hV$ is guarded, $ms(\hSemF{\hVV[\hV/X]}) < \hSemF{\hV}$, so the proof is complete by the inductive hypothesis.
\end{proof}

\begin{defn}
We say that $\hV$ is propositionally inconsistent if $\hV = \hFls$, or $\hV = \hV_1 \land \hV_2$ and one of $\hV_1,\hV_2$ is propositionally  inconsistent, or $\hV = \hV_1 \lor \hV_2$ and both of $\hV_1,\hV_2$ are propositionally  inconsistent, or $\hV = \max X. \hV_1$ or $\hV = \min X. \hV_1$, and $\hV_1$ is propositionally  inconsistent.
We can dually define that $\hV$ is a propositional tautology.
\end{defn}
% A useful measure for guarded formulae is $ms(\hV)$ that measures the longest distance from the root of the syntax tree of $\hV$ to either a constant $\hTru,\hFls$, or to a modality. It is defined recursively: $ms([\act]\hV) = ms(\hSuf{\act}\hV) = ms(\hTru) = ms(\hFls) = 0$; $ms(\max X. \hV) = ms(\min X. \hV) = ms(\hV) + 1$; and $ms(\hV\land \hVV) = ms(\hV \lor \hVV) = \max\{ ms(\hV), ms(\hVV) \} + 1$.

% \ac{here I mean and closed; do we repeat this everywhere, or ...}

\begin{lem}\label{lem:prop-inconst-is-ff}
If a guarded (closed) $\hV \in \ftmuS$ is equivalent to $\hFls$ under finfinite semantics, then it is propositionally inconsistent.
If a guarded (closed) $\hV \in \ftmuC$ is equivalent to $\hTru$ under finfinite semantics, then it is a propositional tautology.
\end{lem}

\begin{proof}
% We extend the definition of propositional inconsistent formulae to work on formulae with free variables.
% We define that $\mV$ is propositionally inconsistent under $\sigma$ the same way as before, only by extending the base cases by $X$, where $\sigma(X) = \emptyset$.
We prove by induction on $ms(\hV)$ that if $\hV$ is not propositionally inconsistent, then % under $\sigma$, then %there is an environment $\sigma$, such that 
$\varepsilon \in \hSemF{\hV}$.
The cases for %$\hV = X$, %then we can just define $\sigma(X)=\fTrc$; t 
$\hV = \hFls$ or $\hV = \hTru$ are vacuous or trivial.
If $\hV = [\ASet]\hVV$, then by definition, $\varepsilon \in \hSemF{\hV}$.
If $\hV = \max X. \hVV$, then, since $\hV$ is not propositionally inconsistent, neither is $\hVV[\hV/X]$; but $ms({\hVV[\hV/X]}) < ms(\hV)$, 
and by the inductive hypothesis $\varepsilon \in \hSemF{\hVV[\hV/X]} = \hSemF{\hV}$.
We can similarly prove that if $\hV$ is not a propositional tautology, then $\varepsilon \notin \hSemF{\hV}$.
\end{proof}

\begin{lem}\label{lem:immediate-rejection}
If $\hV$ is propositionally inconsistent, then $\hSyn{\hV} \wtraSS{} \no$.
\end{lem}
\begin{proof}
Straightforward induction on $ms(\hV)$, using \Cref{lem:mon-combinators,lem:mon-combinators2,lem:sound-and-rej-compl}. 
% \Cref{lem:mon-combinators} and \ref{lem:synt-complete-mon-reactive}.
% \ac{actually needs the variant for recursion, but I could not find it}
\end{proof}

\begin{propbis}{\ref{lem:sound-and-compl-finfinite}}
Given a closed formula $\hV\in \ftmuS$, $\hSyn{\hV}$ is sound and violation-complete for $\hV$ over finfinite traces.
For $\hV\in \ftmuC$, $\hSyn{\hV}$ is sound and satisfaction-complete for $\hV$ over finfinite traces.
\end{propbis}

\begin{proof}
Let $\hV \in \ftmuS$ --- the case for $\hV \in \ftmuC$ is similar.
The proof for Soundness is the same as in the proof of \Cref{lem:sound-and-rej-compl}.
To prove Completeness, if $\fftV \notin \hSemF{\hV}$,
% If $\mV$ rejects $\fftV$, 
then we have two cases. The first is that $\fftV \in \Trc$, in which case, by \Cref{lem:infinite-to-finfinite-semantics}, $\fftV \notin \hSemL{\hV}$, and therefore, by \Cref{lem:sound-and-rej-compl}, $\hSyn{\psi}$ rejects $\fftV$.
The second case is that $\fftV \in \Act^*$, in which case
we use induction on $\fftV$.
The base case is that $\fftV = \varepsilon$, which, by \Cref{lem:finfinite-suffix-closed}, implies that 
$\hV$ is equivalent to $\hFls$, which in turn, by \Cref{lem:prop-inconst-is-ff,lem:immediate-rejection}, implies that $\hSyn{\hV}$ rejects $\varepsilon$.
For $\fftV = \acta \ftV$, we use induction on $ms(\hV)$.
The cases for $\hV = \hTru$ or $\hFls$ are immediate. Since $\hV$ is closed, $\hV \neq X$.
If $\hV = [\ASet] \hVV$, then $\hSyn{\hV} = \ASet.\hSyn{\hVV} + \bar{\ASet}.\yes$, $\acta\in\ASet$ and $\ftV \notin \hSemF{\hVV}$, and by the inductive hypothesis on $\fftV$, $\hSyn{\hVV}$ rejects $\ftV$, therefore $\hSyn{\hV}$ rejects $\fftV$.
If $\hV = [\actb] \hVV$, then it is satisfied by $\fftV$.
The boolean operator cases follow from \Cref{lem:mon-combinators,lem:synt-complete-mon-reactive}.
Finally, if 
$\hV = \max X. \hVV$, then $\hSyn{\hV} = \rec x \hSyn{\hVV}$ and $\hSyn{\hVV[\hV/X]} = \hSyn{\hVV}[\hSyn{\hV}/x]$; 
but then, because of guardedness, $ms({\hVV[\hV/X]}) < ms(\hV)$, so 
by the inductive hypothesis on $ms(\hV)$, 
since $\hV$ is equivalent to $\hVV[\hV/X]$,
$\hSyn{{\hVV[\hV/X]}} = \hSyn{\hVV}[\hSyn{\hV}/x]$ rejects $\fftV$. 
% therefore 
As
$\hSyn{\hV} \traS{\tau} \hSyn{\hVV}[\hSyn{\hV}/x]$, $\hSyn{\hV}$ rejects $\fftV$ too.

\end{proof}

\begin{lem}{\label{lem:trace-processes-are-deterministic}}
If $p$ represents $\act \fftV$ and  $p \reduc q$, then $q$ represents $\fftV$.
\end{lem}
\begin{proof}
If $p \reduc q$, then $p \traS{\act} q$.
If $q \wtraS{\ftV}$, then $p \wtraS{\act\ftV}$, so $\ftV$ is a prefix of $\fftV$.
If $q (\reduc)^k q'$ and $q (\reduc)^k q''$, then $p (\reduc)^{k+1} q'$ and $p (\reduc)^{k+1} q''$, so $q' = q''$.
\end{proof}

\begin{lemmabiss}{\ref{lem:branching-finf-match}}
If $p$ represents $\fftV$, then 
$\fftV \in \hSemF{\hV}$ iff $p \in \hSemB{\hV}$.
\end{lemmabiss}

\begin{proof}
% Without loss of generality, we can assume that all trace-processes appear in the LTS, 
% and by \Cref{lem:trace-processes-are-deterministic} that for all $\act, \fftV$, 
% $p \traS{\act} p_\fftV$.
Given an environment $\sigma$ for finfinite traces, let $\sigma_B$ be an environment on processes, such that for every $X$, 
$$\sigma_B(X) = \{ p \mid 
% \fftV \in \sigma(X) 
% \text{ and } 
p \text{ represents some } \fftV \in \sigma(X) \};$$ 
given an environment $\rho$ on processes, we can similarly define
$$\rho_L(X) = \{ \fftV \mid p \in \rho(X) \text{ and } p \text{ represents } \fftV \}.$$
We prove that if $\fftV \in \hSemF{\hV,\sigma}$, then $p \in \hSemB{\hV,\sigma_B}$, by induction on $\hV$.
The cases for $\hV = \hTru, \hFls, X$, and the boolean operators are immediate.
The cases for $\hV = \hSuf{\act}\hVV$ or $\hV = \hNec{\act}\hVV$ follow from \Cref{lem:trace-processes-are-deterministic}.
If $\hV = \max X. \hVV$  and $\fftV \in \hSemF{\hV}$, then there is some $T \subseteq \hSemF{\hVV,\sigma[X \mapsto T]}$ with $\fftV \in T$.
Let $T' = \{ q \mid q $ represents some $ \fftVV \in T \}$. We can see that $p \in T'$ and by the inductive hypothesis, $T' \subseteq \hSemB{\hVV,(\sigma[X \mapsto T])_B} = \hSemB{\hVV,\sigma_B[X \mapsto T']}$, which implies that $p \in \hSemB{\hV,\sigma_B}$. 
If $\hV = \min X.\hVV$, then $\forall T. (\hSemF{\hVV,\sigma[X \mapsto T]} \subseteq T \Rightarrow \fftV \in T)$.
If for a set of processes $R$, $\hSemB{\hVV,\sigma_B[X \mapsto R]} \subseteq R$, then due to the monotonicity of $\hVV$, for $R_t$ the set of trace-processes from $R$ and $R' = \{ \fftVV \mid \fftVV $ is represented by some $ p \in R \}$, 
$\hSemB{\hVV,(\sigma[X \mapsto R'])_B} = \hSemB{\hVV,\sigma_B[X \mapsto R_t]} \subseteq R$.
By the inductive hypothesis, if $q$ represents $\fftVV$, then $\fftVV \in \hSemF{\hVV,\sigma[X \mapsto R']}$ implies that $q \in \hSemB{\hVV,(\sigma[X \mapsto R'])_B}$, and therefore, $q \in R$, yielding that $\fftVV \in R'$. Therefore, $\fftV \in R'$, and so $p \in R$. Thus, we proved that $\forall R. (\hSemB{\hVV,\sigma_B[X \mapsto R]} \subseteq R \Rightarrow p \in R)$, and thus, $p \in \hSemB{\hV}$.

We can similarly prove that 
if $p \in \hSemB{\hV,\rho}$, then $\fftV \in \hSemF{\hV,\rho_L}$.
\end{proof}

%\subsection{A unifying monitorable fragment}

\subsection{Monitorable Formulae across Semantics}

\begin{lem}\label{prop:infinite-maximality}
If $\hV \in \UHML$ has a
% if $\mV$ is a
sound and violation-complete (\resp satisfaction-complete) regular or reactive parallel monitor over infinite traces, then there is some
$\hVV \in \SHML$ (\resp $\hVV \in \CHML$) that is equivalent to $\hV$ over infinite traces.
% $\fftV \in \hSemF{\hV}$ if and only if $p_\fftV \in \hSemB{\hV}$.
\end{lem}

\begin{proof}
% We observe that $\fftV \mapsto p_\fftV$ is a bijection from finfinite traces to trace-processes.
% Consider the LTS of trace-processes, that is, the processes $p_\fftV$ for finfinite $\fftV$ such that $p_\fftV \wtraS{\ftV}$ if and only if $\ftV$ is a prefix of $\fftV$.
Let $\mV$ be a sound and violation-complete regular or reactive parallel monitor for $\hV$ over finfinite traces.
By  Proposition \ref{prop:extended-mon-to-reg-mon}, there is a regular monitor $\mVV$ that is verdict-equivalent to $\mV$,
so it is also sound and violation-complete for $\hV$. %over finfinite traces.
From \Cref{thm:branching-monitorability}
there is a formula $\hVV \in \SHML$, such that $\mVV$ is sound and violation-complete for
$\hVV$ over all LTSs, including the LTS of trace-processes.
Since $\mVV$ is sound and violation-complete for $\hVV$ on trace processes,
$p_\fftV \in \hSemB{\hVV}$ is equivalent to claiming that $\mVV$ does not reject any trace that $p_\fftV$ can produce.
However, this is equivalent to saying that $\mVV$ does not reject $\tV$, which is equivalent to $\tV \in \hSemL{\hV}$.
% Therefore, to complete this proof it suffices to prove that
By \Cref{lem:branching-finf-match,lem:infinite-to-finfinite-semantics},
$\tV \in \hSemL{\hVV}$ iff $p_\tV \in \hSemB{\hVV}$,
and the proof is complete.
% which can be done by straightforward induction on $\hVV$.
The case for a satisfaction-complete monitor is similar.
\end{proof}

\begin{propbis}{\ref{prop:finfinite-max}}
If $\hV \in \ftmuS$ (\resp $\hV \in \ftmuC$), %has a
% if $\mV$ is a 
% sound and violation-complete (\resp satisfaction-complete) regular or reactive parallel monitor, 
then there is some
$\hVV \in \SHML$ (\resp $\hVV \in \CHML$) that is equivalent to $\hV$ over finfinite traces.
% $\fftV \in \hSemF{\hV}$ if and only if $p_\fftV \in \hSemB{\hV}$.
\end{propbis}

\begin{proof}

If $\hV \in \ftmuS$ (\resp $\hV \in \ftmuC$), then from \cref{lem:sound-and-compl-finfinite}, there is a monitor $\mV$ that is sound and violation-complete (\resp satisfaction-complete) for $\hV$ over finfinite traces. From \Cref{prop:finfinite-maximality}, $\hV$ is then equivalent to a formula in $\SHML$ (\resp \CHML).
\end{proof}

\begin{propbis}{\ref{prop:infinite-max}}
If $\hV \in \maxHML$ (\resp $\hV \in \minHML$), %has a
% if $\mV$ is a 
% sound and violation-complete (\resp satisfaction-complete) regular or reactive parallel monitor, 
then there is some
$\hVV \in \SHML$ (\resp $\hVV \in \CHML$) that is equivalent to $\hV$ over infinite traces.
\end{propbis}

\begin{proof}
$\SHML$ is just $\maxHML$ without disjunctions or $\hSuf{\act}$-operators. We first show that on infinite linear semantics, $\hSuf{\ASet}\hV$ can be rewritten as $\hNec{\coASet}{\hFls} \wedge \hNec{\ASet}{\hV}$.
Indeed, if $\tV\in \hSemL{\hSuf{\ASet}\hV}$ then $\tV=\act \tVV$ for some $\act \in \ASet$ and $\tVV\in \hSemL{\hV}$. 
Then $\tV\in \hSemL{\hNec{\coASet}{\hFls} \wedge \hNec{\ASet}\hV}$.
Conversely, if $\tV\notin \hSemL{\hSuf{\ASet}\hV}$ then $\tV=\act\tVV$ and either $\act\notin \ASet$ or $\tVV\notin \hSemL{\hV}$.
In either case, $\tV\notin \hSemL{\hNec{\coASet}{\hFls} \wedge \hNec{\ASet}\hV}$.

This gives up a formula in $\ftmuS$. From \Cref{prop:finfinite-maximality}, this formula is then equivalent to a $\SHML$ formula over finfinite traces. From \Cref{lem:infinite-to-finfinite-semantics} this equivalence also holds over infinite traces. 
\end{proof}

%\subsection{Comparing linear- and branching -time monitorability}

\begin{lemmabiss}{\ref{lem:subsume}}
If $p\in \hSemB{\hV}$ for $\hV\in \SHML$, and $p$ subsumes $p'$, then $p'\in \hSemB{\hV}$. Dually, if $p\notin \hSemB{\hV}$ for $\hV\in \CHML$, and $p$ subsumes $p'$ then $p'\notin \hSemB{\hV}$. 
\end{lemmabiss}

\begin{proof}
Let  $p\in \hSemB{\hV}$ for $\hV\in \SHML$, and suppose that $p$ subsume $p'$. We claim that $p'\in \hSemB{\hV}$. Indeed,
assume towards a contradiction, that $p'\notin \hSemB{\hV}$. Since $\hV\in \SHML$, from \Cref{thm:branching-monitorability} there is a monitor $\mV$ that is sound and
violation-complete for $\hV$ with respect to the branching-time semantics. As $\mV$ is
violation-complete for $\hV$ and $p'$ does not satisfy $\hV$, we have that
$\rej{m,p',\ftV}$ for some $\ftV$, that is $\sys{\mV}{p'} \wtraS{\ftr} \sys{\no}{q'}$. Then by unzipping (\cref{lem:unzipping}) $\mV \wtraS{\ftr} \no$ and $p'\wtraS{\ftr} q'$. Since $p$ subsumes $p'$, $p\wtraS{\ftr} q$ for some process $q$. Then, from \cref{lem:zipping}  $\sys{\mV}{p} \wtraS{\ftr} \sys{\no}{q}$, i.e., $\rej{m,p,s}$, which contradicts the soundness of $\mV$ for $\hV$.
The case of $\hV\in \CHML$ is obtained by duality.
\end{proof}

\begin{defn}
	We say that a process $p$ subsumes a process $p'$ if $p$ produces all the finfinite traces that $p'$ produces.
\end{defn}

\begin{lem}\label{lem:subsume}
	If $p\in \hSemB{\hV}$ for $\hV\in \SHML$, and $p$ subsumes $p'$, then $p'\in \hSemB{\hV}$. Dually, if $p\notin \hSemB{\hV}$ for $\hV\in \CHML$, and $p$ subsumes $p'$ then $p'\notin \hSemB{\hV}$. 
\end{lem}

\begin{propbis}{\ref{prop:shml-and-forall-traces}}
	For a process $p$ and a formula $\hV\in\SHML$, the following are equivalent:
	\begin{enumerate}
		\item $p\in \hSemB{\hV}$ \label{one}
		\item If $p$ produces a finfinite trace $g$, then $g \in \hSemF{\hV}$ \label{two}
	\end{enumerate}
\end{propbis}

\begin{proof}
	\begin{description}
		
		\item[Direction (\ref{one}) implies (\ref{two})]
		Assume $p\in \hSemB{\hV}$ and that $p$ produces a finfinite trace $g$. Then $p$ subsumes the trace-process $p_g$. From \cref{lem:subsume}, $p_g\in \hSemB{\hV}$.
		Since $p_g$ is a trace process, from \Cref{lem:branching-finf-match}, $g\in \hSemF{\hV}$.
		\item[Direction (\ref{two}) implies (\ref{one})] Assume that $p\notin \hSemB{\hV}$. 
		%Then $\hV$ has a $\CHML$ negation $\hVV$ and $p\in \hSemB{\hVV}$.  seems irrelevant to the rest of the proof
		Let $g$ be a trace produced by $p$. Again, $p$ subsumes $p_g$, and from \cref{lem:subsume} $p_g\notin\hSemB{\hV}$. Since $p_g$ is a trace process, from \Cref{lem:branching-finf-match}, $g\notin \hSemF{\hV}$.
		\qedhere 
	\end{description}
\end{proof}